\documentclass[11pt]{article}
\usepackage[margin=.7in,left=.7in]{geometry}
\usepackage{amsmath}
\usepackage{amsfonts}
\usepackage{amssymb}
\usepackage{amsthm}
\usepackage{subcaption}
\usepackage{float}
\usepackage{graphicx}
\usepackage{color}
\usepackage{xcolor}
\usepackage{hyperref}
\usepackage[all]{xy}

\usepackage{tikz}
\usepackage{tikz-cd}
\usepackage{lipsum}
\usepackage{adjustbox}

\usepackage{multirow}

\usepackage{stmaryrd}    

\usepackage{enumerate} 

\usepackage{helvet}   

\usepackage{floatflt}  
\usepackage{wrapfig} 
\usepackage{array}     
\newcolumntype{L}[1]{>{\raggedright\let\newline\\\arraybackslash\hspace{0pt}}m{#1}}
\newcolumntype{C}[1]{>{\centering\let\newline\\\arraybackslash\hspace{0pt}}m{#1}}
\newcolumntype{R}[1]{>{\raggedleft\let\newline\\\arraybackslash\hspace{0pt}}m{#1}}

\usepackage{diagbox}  

\usepackage{tabularx}   

\newdir{> }{{}*!/10pt/@{>}}

\newcommand{\dslash}{/\!\!/}

\makeatletter
\newif\if@sup
\newtoks\@sups
\def\append@sup#1{\edef\act{\noexpand\@sups={\the\@sups #1}}\act}%
\def\reset@sup{\@supfalse\@sups={}}%
\def\mk@scripts#1#2{\if #2/ \if@sup ^{\the\@sups}\fi \else%
  \ifx #1_ \if@sup ^{\the\@sups}\reset@sup \fi {}_{#2}%
  \else \append@sup#2 \@suptrue \fi%
  \expandafter\mk@scripts\fi}
\def\tensor#1#2{\reset@sup#1\mk@scripts#2_/}
\def\multiscripts#1#2#3{\reset@sup{}\mk@scripts#1_/#2%
  \reset@sup\mk@scripts#3_/}
\makeatother

\makeatletter
\newbox\slashbox \setbox\slashbox=\hbox{$/$}
\def\itex@pslash#1{\setbox\@tempboxa=\hbox{$#1$}
  \@tempdima=0.5\wd\slashbox \advance\@tempdima 0.5\wd\@tempboxa
  \copy\slashbox \kern-\@tempdima \box\@tempboxa}
\def\slash{\protect\itex@pslash}
\makeatother

\def\clap#1{\hbox to 0pt{\hss#1\hss}}

\def\mathrlap{\mathpalette\mathrlapinternal}
\def\mathclap{\mathpalette\mathclapinternal}

\def\mathrlapinternal#1#2{\rlap{$\mathsurround=0pt#1{#2}$}}
\def\mathclapinternal#1#2{\clap{$\mathsurround=0pt#1{#2}$}}

\let\oldroot\root
\def\root#1#2{\oldroot #1 \of{#2}}
\renewcommand{\sqrt}[2][]{\oldroot #1 \of{#2}}

\DeclareSymbolFont{symbolsC}{U}{txsyc}{m}{n}
\SetSymbolFont{symbolsC}{bold}{U}{txsyc}{bx}{n}
\DeclareFontSubstitution{U}{txsyc}{m}{n}

\DeclareSymbolFont{stmry}{U}{stmry}{m}{n}
\SetSymbolFont{stmry}{bold}{U}{stmry}{b}{n}

\DeclareFontFamily{OMX}{MnSymbolE}{}
\DeclareSymbolFont{mnomx}{OMX}{MnSymbolE}{m}{n}
\SetSymbolFont{mnomx}{bold}{OMX}{MnSymbolE}{b}{n}
\DeclareFontShape{OMX}{MnSymbolE}{m}{n}{
    <-6>  MnSymbolE5
   <6-7>  MnSymbolE6
   <7-8>  MnSymbolE7
   <8-9>  MnSymbolE8
   <9-10> MnSymbolE9
  <10-12> MnSymbolE10
  <12->   MnSymbolE12}{}

\makeatletter
\def\Decl@Mn@Delim#1#2#3#4{%
  \if\relax\noexpand#1%
    \let#1\undefined
  \fi
  \DeclareMathDelimiter{#1}{#2}{#3}{#4}{#3}{#4}}
\def\Decl@Mn@Open#1#2#3{\Decl@Mn@Delim{#1}{\mathopen}{#2}{#3}}
\def\Decl@Mn@Close#1#2#3{\Decl@Mn@Delim{#1}{\mathclose}{#2}{#3}}
\Decl@Mn@Open{\llangle}{mnomx}{'164}
\Decl@Mn@Close{\rrangle}{mnomx}{'171}
\Decl@Mn@Open{\lmoustache}{mnomx}{'245}
\Decl@Mn@Close{\rmoustache}{mnomx}{'244}
\makeatother

\makeatletter
\DeclareRobustCommand\widecheck[1]{{\mathpalette\@widecheck{#1}}}
\def\@widecheck#1#2{%
    \setbox\z@\hbox{\m@th$#1#2$}%
    \setbox\tw@\hbox{\m@th$#1%
       \widehat{%
          \vrule\@width\z@\@height\ht\z@
          \vrule\@height\z@\@width\wd\z@}$}%
    \dp\tw@-\ht\z@
    \@tempdima\ht\z@ \advance\@tempdima2\ht\tw@ \divide\@tempdima\thr@@
    \setbox\tw@\hbox{%
       \raise\@tempdima\hbox{\scalebox{1}[-1]{\lower\@tempdima\box
\tw@}}}%
    {\ooalign{\box\tw@ \cr \box\z@}}}
\makeatother

\makeatletter
\def\udots{\mathinner{\mkern2mu\raise\p@\hbox{.}
\mkern2mu\raise4\p@\hbox{.}\mkern1mu
\raise7\p@\vbox{\kern7\p@\hbox{.}}\mkern1mu}}
\makeatother

\newcommand{\R}{\ensuremath{\mathbb R}}
\newcommand{\Q}{\ensuremath{\mathbb Q}}

\renewcommand{\(}{\begin{equation*}}
\renewcommand{\)}{\end{equation*}}
\newcommand{\bea}{\begin{eqnarray*}}
\newcommand{\eea}{\end{eqnarray*}}

\theoremstyle{italics}
\newtheorem{theorem}{Theorem}[section]
\newtheorem{lemma}[theorem]{Lemma}
\newtheorem{prop}[theorem]{Proposition}

\theoremstyle{definition}
\newtheorem{defn}[theorem]{Definition}

\newtheorem{example}[theorem]{Example}

\newtheorem{remark}[theorem]{Remark}
\newtheorem{note[theorem]}{Note}

\usepackage{amsfonts}

\begin{document}

\title{Gauge enhancement of super M-branes\\ via parametrized stable homotopy theory}

\author{
Vincent Braunack-Mayer\thanks{ Institute of Mathematics,
University of Zurich, Winterthurerstrasse 190, CH-8057 Z\"urich, Switzerland. }, \;
  Hisham Sati\thanks{Division of Science and Mathematics, New York University, Abu Dhabi, UAE.}, \;
  Urs Schreiber\thanks{Division of Science and Mathematics, New York University, Abu Dhabi, UAE, on leave from Czech Academy of Science.}}

\maketitle
\abstract{
  A key open problem in M-theory is to explain the mechanism of \lq\lq gauge enhancement'' through which  M-branes
  exhibit the nonabelian gauge degrees of freedom seen perturbatively in the limit of 10d string theory.
  In fact, since only the twisted K-theory classes represented by nonabelian Chan--Paton gauge fields on D-branes have an invariant meaning,
  the problem is really the understanding the M-theory lift of the classification of D-brane charges by twisted K-theory.

  Here we show that this problem has a solution by universal constructions in rational super homotopy theory. We recall how
  double dimensional reduction of super M-brane charges is described by the cyclification adjunction applied to the 4-sphere,
  and how M-theory degrees of freedom hidden at ADE singularities are induced by the suspended Hopf action on the 4-sphere.
  Combining these, we demonstrate that, in the approximation of rational homotopy theory, gauge enhancement in M-theory is
  exhibited by lifting against the fiberwise stabilization of the unit of this cyclification adjunction on the
  A-type orbispace of the 4-sphere.
  This explains how the fundamental D6 and D8 brane cocycles can be lifted from twisted K-theory to a cohomology theory for M-brane charge, at least rationally.}

\tableofcontents

\newpage

\section{Introduction}
\label{Before}

The conjecture of \emph{M-theory} \cite{Townsend95, Witten95I}
(see \cite{Duff99B, BeckerBeckerSchwarz06}) says, roughly, that there exists a
non-perturbative physical theory, which makes the following schematic diagram commute:

\begin{equation}
  \label{TheMillionDollarQuestion}
  \raisebox{55pt}{
  \xymatrix@C=10pt{
    & \fbox{\footnotesize M-Theory}
    \ar[dl]_{ \mbox{ \tiny \begin{tabular}{c} double dimensional reduction \\ along $S^1$-fibration \end{tabular} } }
    \ar[dr]^{ \tiny \begin{tabular}{c} low energy \\  approximation \end{tabular} }
    \\
    \fbox{\footnotesize
      \begin{tabular}{c}
        Perturbative
        \\
        type IIA string theory
      \end{tabular}
    }
    \ar[dr]_{ \tiny\begin{tabular}{c} low energy \\  approximation \end{tabular} }
    &&
    \fbox{\footnotesize
      \begin{tabular}{c}
        11d supergravity
      \end{tabular}
    }
    \ar[dl]^-{~~ \mbox{ \tiny \begin{tabular}{c} dimensional KK-reduction \\ along $S^1$-fibration \end{tabular}   }   }
    \\
    &
    \fbox{\footnotesize
      \begin{tabular}{c}
        10d type IIA
        \\
        supergravity
      \end{tabular}
    }
  }
  }
\end{equation}
Here both \emph{perturbative string theory} as well as \emph{higher-dimensional quantum supergravity} may, with some effort,
be well-defined as \emph{perturbative} S-matrix theories (e.g. \cite[Sec. 12.5]{Polchinski01}\cite{Witten12} and \cite{Donoghue95}\cite{BFR13});
and the conjecture is that there is a joint \emph{non-perturbative} completion of 11-dimensional supergravity and of IIA string theory,
hence of any effective quantum field theory approximating the latter.
Even though the actual nature of M-theory has remained an open problem \cite[Sec. 12]{Moore14}\footnote{At least in mathematics it is not uncommon that a theory is conjectured to exist before its actual nature is
known---famous examples of this include the theory of \emph{motives}, which has meanwhile been discovered, and the \emph{field with one element.}},
there is a huge and steadily growing network of hints supporting the M-theory conjecture.
This should be regarded in light of the situation of perturbative quantum field theory as used in the Standard Model of
particle physics, where the identification of any aspect of its non-perturbative completion is a crucial but wide-open problem,
referred to as one of the ``millennium problems'' \cite{ClayInstitute}.

\medskip
In a series of articles \cite{FSS13, FSS16a, FSS16b, HuertaSchreiber17}, we have shown that a systematic analysis of the
\emph{Green--Schwarz sigma-models} (which define \emph{fundamental} super $p$-branes, such as the fundamental
membrane that gives M-theory its name) from the point of view of \emph{super homotopy theory} provides a concrete
handle on some previously elusive aspects of M-theory; see \cite{Schreiber17c, FSS19} for exposition of this perspective.
Specifically, in the companion article \cite{ADE} it is shown that the existence and classification of \emph{black M-branes at real
ADE singularities} can be systematically derived and analyzed in the supergeometric enhancement of \emph{equivariant
homotopy theory} (see \cite{Blu17}).
However, the M-theory folklore suggests (\cite[Sec. 2]{Sen97}, see e.g. \cite[Sec. 6.3.3]{IbanezUranga12}, also \cite{AcharyaGukov04})
that understanding black branes at ADE-singularities also holds  the key to the all-important, widely expected
yet still mysterious phenomenon of \emph{gauge enhancement} in M-theory. This suggests that the super homotopy theoretic analysis
may also shed light on the true nature of gauge enhancement in M-theory. Here we show that this is indeed the case.

\medskip
In the remainder of this introduction we review the issue of gauge enhancement
in various guises, survey what is known, what is conjectured, and which problems remain essentially unsolved.
After establishing some mathematical results in Sec. \ref{HomotopyTheory} and Sec. \ref{TheATypeOrbispaceOfThe4Sphere},
we explain our solution to the gauge enhancement problem in Sec. \ref{TheMechanism}.


\medskip
\noindent {\bf Double dimensional reduction.}
At the heart of the matter is \emph{double dimensional reduction}, originally due to \cite{DuffHoweInamiStelle87, Townsend95b} and
whose rigorous formulation from \cite[Sec. 3]{FSS16a} \cite[Sec. 3]{FSS16b} we recall and expand upon below in Sec. \ref{TheMechanism}.
 Going back to ideas of Kaluza and Klein a century ago, in ordinary \emph{dimensional reduction}, we take spacetime to be a
 pseudo-Riemannian $S^1$-fibration over a $(D-1)$-dimensional base space and consider the limit in which the circle fiber becomes infinitesimal.
In this limit, we obtain a field theory on the $(D-1)$-dimensional base space of the circle fibration, hence in a lower-dimensional spacetime, which has a larger space of field species including the Fourier modes of the original fields on the circle fiber.

\begin{equation}
  \label{DD}
  \raisebox{44pt}{
  \scalebox{.9}{
  \xymatrix@C=-32pt{
    Y
    \ar[dd]^{ \pi_{{}_{S^1}} }
    &\,\,\,\,\,\,\,\,\,\,\,\,\,\,\,\,\,\,\,\,\,\,\,\,\,\,\,\,\,\,\,\,\,\,\,\,\,\,\,&
    \fbox{
      \begin{tabular}{c}
        11-dimensional
        \\
        spacetime
      \end{tabular}
    }
    \ar[dd]|{ \mbox{ \footnotesize \begin{tabular}{c} \tiny circle bundle \\ \tiny projection  \end{tabular}  } }
    &\,\,\,\,\,\,\,\,\,\,\,\,\,\,\,\,\,\,\,\,\,\,\,\,\,\,\,\,\,\,\,\,\,\,\,\,\,\,\,&
    \fbox{ \begin{tabular}{c} $D = 11$  supergravity/ \\ M-theory\end{tabular} }
    \ar@{|->}[dd]|{ \vphantom{\big(}\mbox{ \tiny dimensional reduction } }
    &\,\,\,\,\,\,\,\,\,\,\,\,\,\,\,\,\,\,\,\,\,\,\,\,\,\,\,\,\,\,\,\,\,\,\,\,\,\,\,&
    &
    \fbox{ \begin{tabular}{c} $p$-brane  \end{tabular} }
    \ar@{|->}[ddl]|{ \phantom{A \atop A} }
    \ar@{|->}[ddr]|{ \phantom{ A \atop A }  }
    \ar@{}[dd]|{ \mbox{ \begin{tabular}{c} \tiny double dimensional reduction \end{tabular} } }
    \\
    \\
    X = Y \dslash S^1
    &\,\,\,&
    \fbox{
      \begin{tabular}{c}
        10-dimensional
        \\
        spacetime
      \end{tabular}
    }
    &\,\,\,&
    \fbox{ \begin{tabular}{c} $D=10$ supergravity / \\ type IIA string theory \end{tabular} }
    &\,\,\,&
    \fbox{ $(p-1)$-brane }
    &&
    \fbox{ $p$-brane }
  }}
  }
\end{equation}
Secondly, one finds higher-dimensional analogs of black holes in higher-dimensional supergravity,
with $(p+1)$-dimensional singularities, called \emph{black} (or ``solitonic'') \emph{$p$-branes} \cite{DIPSS88}.
Under dimensional reduction of the ambient supergravity theory, the singularity of a black $p$-brane may
or may not extend along the circle fiber that is  being shrunk away. If it does not, then the result is again a black $p$-brane
solution, now in the lower-dimensional supergravity theory. But if it does, then along with reduction in spacetime dimension
from $D$ to $(D -1)$, the black $p$-brane singularity effectively appears as a black $(p-1)$-brane solution in the
lower-dimensional supergravity theory; whence ``double dimensional reduction''.

\medskip
\noindent {\bf Chan--Paton gauge enhancement.} In its most immediate (albeit naive) form, the formulation of the problem
of \emph{gauge enhancement} in M-theory proceeds from the expected double dimensional reduction \eqref{DD}
of the black branes of M-theory (see \cite{ADE} for a precise account), to those of type IIA string theory,
which for black branes looks as follows \cite{Townsend95, Townsend95b, FSa, FS}:

$$
  \mathclap{
  \xymatrix@C=-1pt@R=1.2em{
    \fbox{\hspace{-3mm} \footnotesize
      \begin{tabular}{c}
        Black brane species
        \\
        in M-theory
      \end{tabular}
   \hspace{-3mm} }
    \ar[dd]|{\mbox{\tiny
      double dimensional reduction
    }}
    &
    &
    \mathrm{MW}
    \ar@{|->}[ddl]
    & &
    \mathrm{M2}
    \ar@{|->}[ddl]
    \ar@{|->}[ddr]
    && &&
    \mathrm{M5}
    \ar@{|->}[ddl]
    \ar@{|->}[ddr]
    &&
    \mathrm{MK6}
    \ar@{|->}[ddr]
    &&&
    &
    \mathrm{MO9}
    \ar@{|->}[ddl]
    &
    \fbox{
      \begin{tabular}{c}
        \multirow{2}{*}{ \color{gray} unknown }
        \\
        $\phantom{A}$
      \end{tabular}
     }
    \ar[dd]
    \\
    \\
    \fbox{\hspace{-3mm} \footnotesize
      \begin{tabular}{c}
        Black brane species in
        \\
        type IIA string theory
      \end{tabular}
    \hspace{-3mm} }
    &
    \mathrm{D0}
    &&
    \mathrm{NS1}
    &&
    \mathrm{D2}
    &&
    \mathrm{D4}
    &&
    \mathrm{NS5}
    &&
    {\mathrm{D6}}
    &&
    {\mathrm{O8}}
    &&
    \fbox{\hspace{-2mm}\footnotesize
      \begin{tabular}{c}
        Chan--Paton
        \\
        gauge enhancement
      \end{tabular}
    \hspace{-3mm}}
  }}
$$

From perturbative string theory one finds that open fundamental strings ending on the D-branes behave as quanta
for an abelian $U(1)$-gauge theory (i.e. electromagnetism) on the worldvolume of the D-brane. A widely
accepted but informal\footnote{ In \cite[first line on p. 8]{Witten96} the argument was introduced as an ``obvious guess''.
 Most subsequent references cite this as a fact, e.g. the review \cite[Sec. 3]{Myers03}, despite the lack of a formal argument.
 } argument \cite[Sec. 3]{Witten96} indicates that if $N$ such D-branes are \emph{coincident} then
the gauge group \emph{enhances} from the abelian group $(U(1))^N
$ to the non-abelian
group $U(N)$.
The idea is that massless open fundamental strings stretch in $N \times N$ possible ways between the
$N$ coincident D-branes,
thus constituting gauge bosons that organize in $N \times N$ unitary matrices, called \lq\lq Chan--Paton factors''.
However, it is non-trivial to check that scattering of these open strings reproduces the scattering
amplitudes of gauge bosons for non-abelian gauge theory (Yang-Mills theory).
First approximate numerical checks of this idea are due to \cite{ColettiSigalovTaylor03} and
similar numerical checks as well as an exact derivation under simplifying assumptions are given in \cite{BerkovitsSchnabl03};
a full derivation was claimed in \cite{Lee17}.

\medskip
This phenomenon of \emph{gauge enhancement on D-branes} is of paramount importance for string theory, in particular
as a candidate for a theory of realistic physics. The fundamental gauge fields observed in nature, per the Standard Model of
particle physics, do of course involve non-abelian gauge groups corresponding to the weak and strong nuclear forces.
While Kaluza--Klein dimensional reduction from 11-dimensional supergravity may
exhibit such non-abelian gauge groups, this happens in a manner incompatible with realistic coupling of fermionic matter fields \cite{Witten81}.
 Instead, realistic gauge fields must arise from gauge enhancement on coincident D-branes (see \cite[Sec. 10]{IbanezUranga12}).
 Moreover, all discussion of modern string theoretic topics such as \emph{AdS/CFT duality} (see \cite{AGMOO99}) or \emph{geometric
 engineering of gauge theories} (going back to \cite{HananyWitten97}, see e.g. \cite{Fazzi17}), such as for classification of 6d SCFTs
 (as in \cite{ZHTV15}), depend crucially on gauge enhancement on coincident D-branes. But, under the M-theory conjecture,
double dimensional reduction should exhibit an equivalence (``duality'') between (strongly coupled) type IIA string theory and
M-theory, in particular between the full non-perturbative theory of D-branes and their M-brane pre-images.
Hence if M-theory exists, then gauge enhancement on coincident D-branes must correspond to, and is potentially explained by,
a corresponding phenomenon on M-branes. The most immediate incarnation of the problem of gauge enhancement in M-theory
can, therefore, be succinctly phrased as:

\vspace{.2cm}

\noindent {\bf Open Problem, version 1:}

\vspace{-2mm}
\begin{quote}
{\it What is the lift to M-theory of the non-abelian Chan--Paton gauge field degrees of freedom on coincident D-branes?}
\end{quote}


\noindent Since the string theory literature tends to blur the distinction between what is known and what is conjectured, we briefly highlight
what the folklore on this problem (e.g. \cite[Sec. 6.3.3]{IbanezUranga12}, \cite{AcharyaGukov04}) does and does not achieve:

\begin{itemize}
\vspace{-2mm}
\item A celebrated recent result (see \cite{BLMP13} for a review) shows the existence of a class of non-abelian gauge field theories that are plausible candidates for the worldvolume theories expected to live on black M2-branes sitting at ADE singularities (the \lq\lq BLG-model''
\cite{BL, Gu} and, more generally, the \lq\lq ABJM-model'' \cite{ABJM08}).
However, a derivation of these field theories from M-theoretic degrees of freedom is missing; the argument works just by consistency checks.

\vspace{0mm}
\item On the other hand, a conjectural sketch of an explicit derivation does exist for the M-brane species $\mathrm{MK6}$,
whose image in the low-energy approximation provided by 11-dimensional supergravity
is supposed to be the \emph{Kaluza--Klein monopole spacetime} and which becomes the black D6-brane
under dimensional reduction \cite{Townsend95}, \cite[Sec. 1]{Sen97}. For an \emph{abelian} gauge field on the D6-brane,
a straightforward analysis shows that it is sourced by doubly dimensionally reduced M2-branes
ending on the KK-monopole \cite[Sec. 2]{GomezManjarin02}, \cite[Sec. 5.4]{Manjarin04}.
But more generally, KK monopoles may be argued to be the fixed point loci of spacetime orbifolds locally of the
form $\mathbb{R}^{6,1} \times \mathbb{C}^2 \dslash G_{\mathrm{ADE}}$
for a finite subgroup $G_{\mathrm{ADE}} \subset SU(2)$. A classical theorem \cite{DuVal} (see \cite{Reid} for a review)
 implies that such singularities are canonically resolved by spheres touching along the shape of a simply-laced Dynkin diagram:

\begin{center}
\includegraphics[width=.47\textwidth]{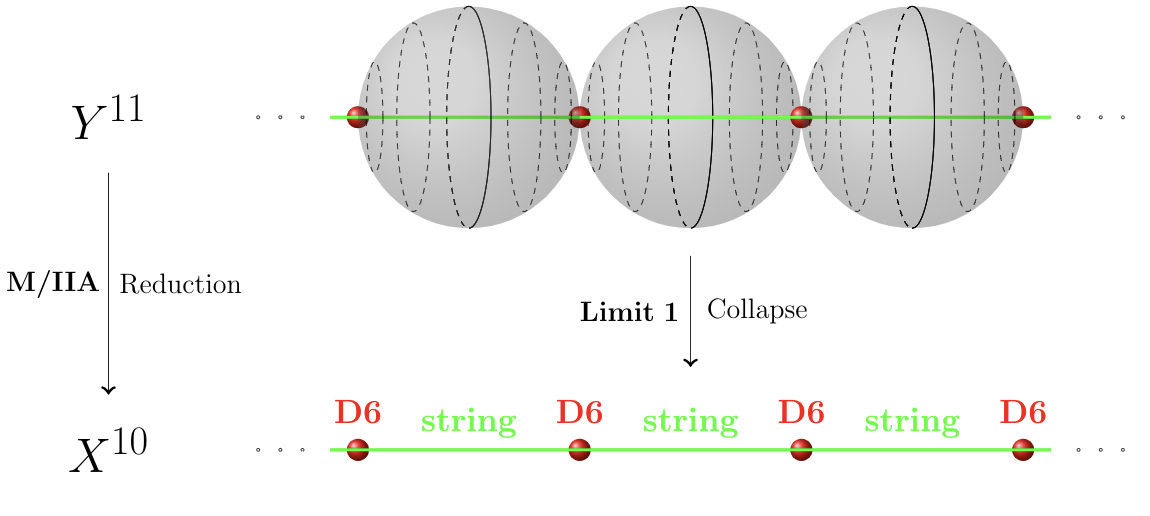}
\hspace{-1mm}
\includegraphics[width=.37\textwidth]{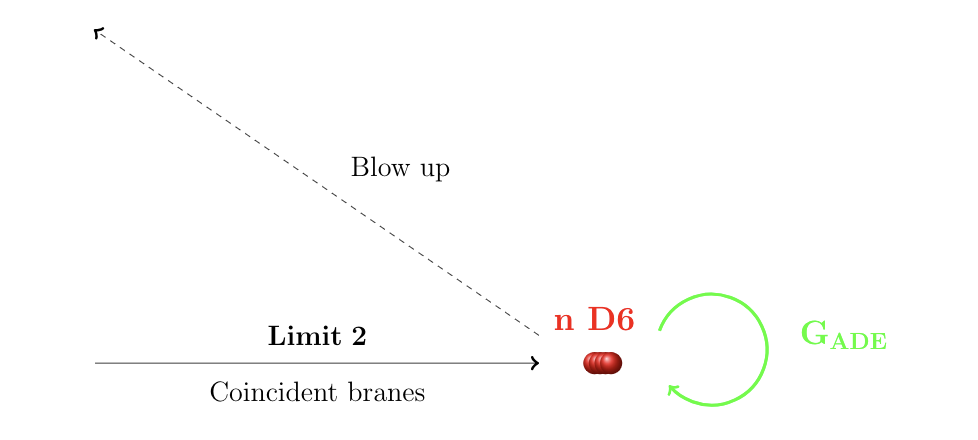}
\end{center}

If one here imagines that M2-branes wrap these \emph{vanishing 2-cycles} then, under double dimensional reduction,
this situation looks like an M-theoretic lift of the strings stretching between several coincident D6-branes, as indicated in the figure above,
and hence looks like an M-theoretic lift of the Chan--Paton gauge enhancement mechanism discussed above.

This argument goes back to \cite[Sec. 2]{Sen97}; for review see \cite[Sec. 6.3.3]{IbanezUranga12}.
While this story is appealing, it is unsatisfactory that it has to treat the membrane in 11-dimensional spacetime as
a direct analogue of the fundamental string in 10-dimensional spacetime; after all, the very term \lq\lq M-theory'' instead of \lq\lq Membrane theory''
was chosen as a reminder that this direct analogy is too naive.\footnote{\cite{HoravaWitten96a}:
``{\it As it has been proposed that [this] theory is a supermembrane theory but there are some reasons
to doubt that interpretation, we will  non-committedly call it the $M$-theory, leaving to the future the relation
 of $M$ to membranes.}'' }
\end{itemize}

In short, the traditional attempts to understand gauge enhancement on M-branes suffer from the lack of any handle on the
actual nature of M-theory.
But what is worse, these stories argue a point that has meanwhile come to be thought as being inaccurate:

\medskip

\noindent {\bf K-Theoretic gauge enhancement.}
A well-known series of arguments \cite{MinasianMoore97, Witten98, FreedWitten99, FreedHopkins00, MooreWitten00} shows that
the gauge fields carried by D-branes do not actually have well-defined existence in themselves. Instead,
one must view all the $\mathrm{D}p$-branes taken together and then their gauge fields
serve as representatives for classes in \emph{twisted K-theory}
\cite{BouwknegtMathai00, Witten00, Freed00, BCMMS02, BEV03, MathaiSati04, EvslinSati06, Evslin06}.
It is only these twisted K-theory classes that are supposed to have an intrinsic meaning. In other words, non-abelian
gauge fields on separate D-brane species are much like coordinate charts on spacetime:
a convenient but non-invariant means of presenting a specific structure. In actual reality, the D-brane species
$Dp$, $p \in \{0,2,4,6,8\}$ only have a unified joint existence and gauge enhancement to on separate D-brane species is only one presentation, out of many, of a unified
\emph{higher gauge field}: a cocycle in twisted K-theory.
In view of this, the problem of gauge enhancement in M-theory is really the following:\footnote{
While \emph{a derivation of K-theory from M-theory} is suggested by the title of \cite{ADerivationofK},
that article only checks that the behavior of the partition function of the 11d supergravity $C$-field is compatible
with the \emph{a priori} K-theory classification of D-branes. Seeking a generalized cohomology describing
the M-field and M-branes was originally advocated for in \cite{S1, S2, S3, S4}.
}

\vspace{.2cm}

\noindent {\bf Open Problem, version 2:}

\vspace{-2mm}
\begin{quote}
{\it What is the cohomology theory classifying M-brane charge, and how does double dimensional reduction
reduce it to the classification of D-brane charge in twisted K-theory?}
\end{quote}

\noindent Hence the refined perspective of twisted K-theory shifts the focus away from Chan--Paton-like gauge
enhancement on one particular D-brane species (which seems to have no invariant meaning in the full theory) and
instead highlights the problem of how the full list of D-brane species arises and carries a unified charge in twisted
K-theory.  However, this is in conflict with the M-theoretic origin of the D-branes in the traditional story recalled
above: only the $\mathrm{M2}$-brane and $\mathrm{M5}$-brane exist as \emph{fundamental branes}, meaning
that they have corresponding Green--Schwarz type sigma-models (we recall this in  detail below in Sec.
\ref{FundamentalpBranes}). The double dimensional reduction of these yields the fundamental brane species in
type IIA string theory, \emph{except} the $\mathrm{D6}$ and the $\mathrm{D8}$ (while the $\mathrm{D0}$
now encodes the circle fibration itself):

$$
  \xymatrix@C=6pt@R=1.2em{
    \fbox{
      \begin{tabular}{c}
     \footnotesize   Fundamental brane species
        \\
   \footnotesize     in M-theory
      \end{tabular}
    }
    \ar[dd]|{\mbox{\tiny
      double dimensional reduction
    }}
    &
    && &
    \mathrm{M2}
    \ar@{|->}[ddl]
    \ar@{|->}[ddr]
    && &&
    \mathrm{M5}
    \ar@{|->}[ddl]
    \ar@{|->}[ddr]
    &&
    \\
    \\
    \fbox{
      \begin{tabular}{c} \footnotesize
        Fundamental brane species
        \\
      \footnotesize  in type IIA string theory
      \end{tabular}
    }
    &
   ~~~~~~ \mathrm{D0}
    &&
    \mathrm{F1}
    &&
    \mathrm{D2}
    &&
    \mathrm{D4}
    &&
    \mathrm{NS5}
    &&
    {\color{gray} \mathrm{D6}}
    &&
    {\color{gray} \mathrm{D8}}
  }
$$

On the other hand, up to now the M-theoretic origin of the $\mathrm{D6}$-brane has only been argued in its \emph{black brane}
incarnation, which is supposed to be given by the $\mathrm{MK6}$ as recalled above. The nature of the
$\mathrm{D8}$-brane is yet more subtle \cite{BRGPT96}---an M-theory lift has been proposed in \cite{Hull98}, but the
proposal is not among the usual list of expected black M-branes. On the other hand, the reduction of the $\mathrm{MO9}$-brane,
whose existence in M-theory is solid (see \cite{ADE}), is not the black $\mathrm{D}8$-brane, but rather the $\mathrm{O}8$-plane
 (e.g. \cite[Sec. 3]{GKST01}).
This highlights that the core of the open problem, from the refined K-theoretic perspective, is really in the appearance of
 the D6- and D8-brane:

\vspace{.2cm}

\noindent {\bf Open Problem, version 3:}

\vspace{-2mm}
\begin{quote}
\noindent {\it What is the lift to M-theory of the fundamental D6-branes and D8-branes in type IIA string theory, such that the unified $\mathrm{D}p$-branes
are jointly classified by twisted K-theory?}
\end{quote}
\vspace{.2cm}

\noindent {\bf K-Theoretic gauge enhancement in rational approximation}.
Twisted K-theory is a comparatively complicated structure, and the fine detail of which of its various variants
really applies to D-branes is still the subject of discussion (\cite{KS2, DFM09, S4, GS19}). Of course the glaring problem here
is, once more, that the non-perturbative theory that ought to answer this question is missing.
It is worth highlighting that the issue of gauge enhancement is visible, and has remained unresolved,
already in the \emph{rational} approximation (i.e. ignoring all torsion-group effects), where a cocycle in twisted K-theory
reduces to a cocycle in \emph{twisted de Rham cohomology}. Associated to each brane species is a differential form on
spacetime---its \emph{flux form}---which corresponds to the brane in analogy to the correspondence between
the \emph{Faraday tensor} and charged particles (0-brane) in electromagnetism. The double dimensional
reduction of these flux forms is as follows (see \cite[Sec 4.2]{MathaiSati04}), parallel to the pattern of the double
dimensional reduction of the fundamental branes.

\begin{equation}
  \label{Gysin}
  \hspace{-5mm}
  \mathclap{
  \raisebox{37pt}{
  \xymatrix@C=-.1em@R=1.2em{
    \fbox{\hspace{-3mm}
      \begin{tabular}{c}
        \footnotesize Flux forms in
        \\
        \footnotesize $D = 11$, $\mathcal{N} = 1$ supergravity
      \end{tabular}
    \hspace{-3mm}}
    \ar[dd]|{ \vphantom{\big(}\mbox{\tiny
      double dimensional reduction
    }}
    &
    && &
    G_4
    \ar@{|->}[ddl]
    \ar@{|->}[ddr]
    && &&
    G_7
    \ar@{|->}[ddl]
    \ar@{|->}[ddr]
    &&
    &&&&
    \fbox{\hspace{-3mm}
      \begin{tabular}{c}
     \footnotesize   Cocycle on
        $D = 11$ spacetime
        \\
\footnotesize        in rational
        cohomotopy
      \end{tabular}
    \hspace{-3mm}}
    \ar[dd]|-{ \vphantom{\big(}\mbox{ \tiny Gysin sequence }  }
    \\
    \\
    \fbox{\hspace{-3mm}
      \begin{tabular}{c}
        \footnotesize Flux forms in
        \\
       \footnotesize  $D =10$, $\mathcal{N} = (1,1)$ supergravity
      \end{tabular}
    \hspace{-3mm}}
    &
    F_2
    &&
    H_3
    &&
    F_4
    &&
    F_6
    &&
    H_7
    &&
    {\color{gray} F_8}
    &&
    {\color{gray} F_{10}}
    &
    \fbox{\hspace{-3mm}
      \begin{tabular}{c}
     \footnotesize   Cocycle on
        $D = 10$ spacetime
        \\
     \footnotesize   in rational 6-truncated
        \\
     \footnotesize   twisted K-theory
      \end{tabular}
    \hspace{-3mm}}
  }}}
\end{equation}

\noindent Here the right hand sides have been recognized in \cite[Sec 2.5]{S-top} \cite{cohomotopy}
for the top part and derived in \cite{FSS16a, FSS16b} for the bottom part; we discuss this in detail below
in Sec. \ref{With}. Notice that the $F_2$-contribution does arise from plain double dimensional reduction,
not directly from the $(G_4,G_7)$, but as a rational image of the first Chern-class of the circle fibration itself.
 Moreover, these differential forms satisfy relations (``twisted Bianchi identities'') that
identify them, via the de Rham theorem,
as cocycles in the rationalization of a generalized cohomology theory:
\begin{center}
\begin{tabular}{|c||c|c|}
  \hline
  \begin{tabular}{c}
    Flux
    \\
    forms
  \end{tabular}
  &
  \begin{tabular}{c}
    \bf Twisted
    \\
    \bf Bianchi identity
  \end{tabular}
  &
  \begin{tabular}{c}
    \bf Rational
    \\
    \bf cocycle in
  \end{tabular}
  \\
  \hline
  \hline
  {\bf M-theory} &
  $
    \begin{aligned}
      d G_4 & = 0
      \\
      d G_7 & = -\tfrac{1}{2}G_4 \wedge G_4
    \end{aligned}
  $
  &
  \begin{tabular}{c}
    cohomotopy
    \\
    in degree 4
  \end{tabular}
  \\
  \hline
  \begin{tabular}{c}
    \bf Type IIA
    \\
    \bf string theory
  \end{tabular}
  &
  $
    \begin{aligned}
      d H_3 & =0
      \\
      d F_{2p + 4} & = H_3 \wedge F_{2 p + 2}
    \end{aligned}
  $
  &
  \begin{tabular}{c}
    twisted K-theory
    \\
    in even degree
  \end{tabular}
  \\
  \hline
\end{tabular}
\end{center}
In this rational approximation, the core of the gauge enhancement problem is still fully visible:

\vspace{.2cm}

\label{FirstRationalProblem}
\noindent \hypertarget{FirstRational}{{\bf Open Problem, rational version 1:}}

\vspace{-2mm}
\begin{quote}
{\it What is the origin of the RR-flux forms $F_8$ and $F_{10}$ in M-theory, such that these unify
with the double dimensional reduction of the M-flux $(G_4, G_7)$ to an un-truncated cocycle in
rational twisted K-theory (i.e. in twisted de Rham cohomology)?}
\end{quote}

\noindent  \rm We present a solution to this version of the problem in Sec. \ref{With}.
Though working in the rationalized setting means that we are disregarding all torsion
\footnote{This torsion is in the sense of cohomology or homotopy classes. In the following paragraph we use torsion
in the sense of differential (super)geometry. We hope that the distinction will be clear from the context.}
 information for the time being,
it has the striking advantage that these rationalized relations follow rigorously from a \emph{first principles}-definition
of M-branes (recalled below in Sec. \ref{FundamentalpBranes}), and hence serve as a starting point for
a systematic analysis of the problem of gauge enhancement.

\medskip
To properly take local supersymmetry into account, one has to consider the refinement of the plain flux forms to
super-flux forms on super-spacetime. The torsion-freeness constraints of supergravity geometrically require
 \cite{Lott90, EE}  the bifermionic component of these super-flux forms to be covariantly constant on each
super tangent space (see \cite[Sec. 1]{Higher-T}), where they correspond to those cocycles $\mu_{{}_{p+2}}$ in the
supersymmetry super Lie algebra cohomology defining the Green--Schwarz-type sigma-models for the fundamental
$p$-branes \cite{AETW87, AzTo89}.
This identification locates the problem in the precise context of \emph{super homotopy theory} of super-spacetimes.
The beauty of this is that homotopy theory is governed by \emph{universal constructions}, which, roughly, means that
it exhibits the emergence of \lq\lq god-given'' structures from a minimum of input.
In fact, in super homotopy theory the super-cocycles $\mu_{{}_{M2}}$ and $\mu_{{}_{M5}}$ witnessing the fundamental M-branes emerge by
a universal construction (a kind of equivariant Whitehead tower) from nothing but the superpoint \cite{FSS13, HuertaSchreiber17},
and their double dimensional reduction is reflected by another universal construction \cite[Sec. 3]{FSS16a} \cite[Sec. 3]{FSS16b}---the \emph{$\mathrm{Ext}/\mathrm{Cyc}$-adjunction} (discussed in detail in Sec. \ref{TheAdjunction} below).

\begin{equation}
  \label{DDCoc}
  \hspace{-3mm}
  \mathclap{
  \raisebox{45pt}{
  \xymatrix@C=-.1pt@R=1.2em{
    \fbox{\hspace{-3mm}
      \begin{tabular}{c}
    \footnotesize    Fundamental branes
        \\
         \footnotesize in M-theory
      \end{tabular}
    \hspace{-3mm}}
    \ar[dd]|{ \vphantom{\big(}\mbox{\tiny
      double dimensional reduction
    }}
    &
    && &
    \mu_{{}_{\mathrm{M2}}}
    \ar@{|->}[ddl]
    \ar@{|->}[ddr]
    && &&
    \mu_{{}_{\mathrm{M5}}}
    \ar@{|->}[ddl]
    \ar@{|->}[ddr]
    &&
    &&&&
    \fbox{\hspace{-3mm}
      \begin{tabular}{c}
        \footnotesize  Super-cocycle
        \\
       \footnotesize   on super-spacetime
        \\
       \footnotesize   in rational
        cohomotopy
      \end{tabular}
    \hspace{-3mm}}
    \ar[dd]|{ \vphantom{\big(} \mbox{ \tiny cyclification } }
    \\
    \\
    \fbox{\hspace{-3mm}
      \begin{tabular}{c}
        \footnotesize Fundamental branes
        \\
       \footnotesize  in type IIA string theory
      \end{tabular}
   \hspace{-4mm} }
    &
    \mu_{{}_{\mathrm{D0}}}
    &&
    \mu_{{}_{\mathrm{F1}}}
    &&
    \mu_{{}_{\mathrm{D2}}}
    &&
    \mu_{{}_{\mathrm{D4}}}
    &&
    \mu_{{}_{\mathrm{NS5}}}
    &&
    {\color{gray} \mu_{{}_{\mathrm{D6}}}}
    &&
    {\color{gray} \mu_{{}_{\mathrm{D8}}}}
    &
    \fbox{\hspace{-3mm}
      \begin{tabular}{c}
         \footnotesize Super-cocycle
        \\
         \footnotesize on super-spacetime
        \\
          \footnotesize in rational 6-truncated
        \\
          \footnotesize twisted K-theory
      \end{tabular}
    \hspace{-3mm}}
  }}}
\end{equation}

\begin{center}
\begin{tabular}{|c||c|c|}
  \hline
  \begin{tabular}{c}
    Fundamental brane
    \\
    super-cocycle
  \end{tabular}
  &
  \begin{tabular}{c}
    \bf Cocycle
    \\
    \bf condition
  \end{tabular}
  &
  \begin{tabular}{c}
    \bf in rational
    \\
    \bf cohomology theory
  \end{tabular}
  \\
  \hline
  \hline
  {\bf M-theory} &
  $
    \begin{aligned}
      d \mu_{{}_{M2}} & = 0
      \\
      d \mu_{{}_{M5}} & = -\tfrac{1}{2} \mu_{{}_{M2}} \wedge \mu_{{}_{M2}}
    \end{aligned}
  $
  &
  \begin{tabular}{c}
    cohomotopy
    \\
    in degree 4
  \end{tabular}
  \\
  \hline
  \begin{tabular}{c}
    \bf Type IIA
    \\
    \bf string theory
  \end{tabular}
  &
  $
    \begin{aligned}
      d \mu_{{}_{F1}} & =0
      \\
      d \mu_{{}_{2p+2}} & = \mu_{{}_{F1}} \wedge \mu_{{}_{2p}}
    \end{aligned}
  $
  &
  \begin{tabular}{c}
    twisted K-theory
    \\
    in even degree
  \end{tabular}
  \\
  \hline
\end{tabular}
\end{center}

\vspace{1mm}
\noindent Here (rational) \emph{cohomotopy} in degree 4 is the generalized \emph{non-abelian} cohomology theory represented
 by the (rationalized) 4-sphere, meaning that the joint $\mathrm{M2}/\mathrm{M5}$-brane cocycle is a morphism
 in the rational super homotopy category of the form \cite{cohomotopy,FSS16a}
$$
  \xymatrix{
    \mathbb{R}^{10,1\vert \mathbf{32}}
    \ar[rr]^-{ \mu_{{}_{M2/M5}} }
    &&
    S^4
  }
  \phantom{AAA}
  \in
  \mathrm{Ho}\left( \mathrm{SuperSpaces}_{\mathbb{R}}  \right).
$$
That cohomotopy governs the M-brane charges this way, at least rationally,  was first proposed and highlighted
in \cite[Sec. 2.5]{S-top}.  \emph{A priori}
there are many homotopy types that look like the 4-sphere in the rational approximation; however in \cite{ADE}
further precise evidence was provided to demonstrate the sense in which the 4-sphere is the correct coefficient for
M-brane charge. Namely, the 4-sphere coefficient is naturally identified with the 4-sphere around a black
M5-brane singularity in $D = 11$ supergravity and this identification induces a real structure on the 4-sphere together
with actions of the ADE subgroups of $SU(2)$ that are compatible with the corresponding BPS actions on super-spacetime.
It is therefore natural to ask for enhancements of the
$\mathrm{M2}/\mathrm{M5}$-brane cocycle to \emph{equivariant} cohomotopy
$$
  \xymatrix{
    \mathbb{R}^{10,1\vert \mathbf{32}}
    \ar@(ul,ur)[]^{ G_{\mathrm{ADE}} \times G_{\mathrm{HW}} }
    \ar[rr]^-{ \widehat{\mu}_{{}_{M2/M5}} }
    &&
    S^4
    \ar@(ul,ur)[]^{ G_{\mathrm{ADE}} \times G_{\mathrm{HW}} }
  }
  \phantom{AA}
  \in
  \mathrm{Ho}\big( \left( G_{\mathrm{ADE}} \times G_{\mathrm{HW}} \right)  \mathrm{\mbox{-}SuperSpaces}_{\mathbb{R}}  \big).
$$
Here $G_{\mathrm{ADE}} \subset SU(2)$ is a finite
subgroup as per the ADE-classification that acts by orientation-preserving
super-spacetime automorphisms, while $G_{\mathrm{HW}} = \mathbb{Z}_2$
is an orientation-reversing reflection as in Ho{\v r}ava--Witten theory.

\medskip
It is shown in \cite{ADE} that such an equivariant enhancement exists and makes the \emph{black branes} at
ADE singularities appear. This results in a unified framework for black and fundamental M-branes.
In particular, the corresponding A-series actions on the 4-sphere factor through the $U(1)$-action
$$
  \xymatrix{
    S^4
    \ar@(ul,ur)[]^{ S^1}
  }
  \;:=\;
  S( \mathbb{R} \oplus \!\!\!\!\xymatrix{\mathbb{C}^2 \ar@(ul,ur)[]^{\rm U(1) } }\!\!\! )
  \,\subset\,
  S( \mathbb{R} \oplus \!\!\!\!\xymatrix{\mathbb{C}^2 \ar@(ul,ur)[]^{\rm SU(2) } }\!\!\!\! )
$$
obtained as the suspension of the circle action on the complex Hopf fibration $H_\mathbb{C}\colon  S^3 \to S^2$. The projection to the corresponding homotopy quotient is identified with the M-theory $S^1$-fibration in the near horizon geometry of an M5-brane:

\hspace{-15mm}
\begin{equation}
  \label{SpacetimeAndATypeOrbispace}
  \mathclap{
  \raisebox{60pt}{\xymatrix@C=1em@R=1.3em{
    &
    \ar@{}[rr]|{ \fbox{ M2/M5-brane cocycle } }
    &&
    \\
    \fbox{  \hspace{-3mm}
      \begin{tabular}{c}
        11d super-spacetime
      \end{tabular}
   \hspace{-3mm} }
    &
    \mathbb{R}^{10,1\vert \mathbf{32}}
    \ar[dd]
    \ar[rr]^-{ \mu_{{}_{M2/M5}} }
    &&
    S^4
    \ar[dd]
    &
    \fbox{\hspace{-1mm} 4-sphere coefficient \hspace{-1mm}}
    \\
    \\
    \fbox{\hspace{-1mm}
      10d super-spacetime
    \hspace{-1mm}}
    & \mathbb{R}^{9,1\vert \mathbf{16} + \overline{\mathbf{16}}}
    \simeq \mathbb{R}^{10,1\vert \mathbf{32}} \dslash S^1
    &&
    S^4 \dslash S^1
    &
    \fbox{\hspace{-3mm}
      \begin{tabular}{c}
        A-type orbispace
        \\
        of the 4-sphere
      \end{tabular}
   \hspace{-3mm} }
  }}}
\end{equation}

\medskip\noindent In conclusion, the problem of gauge enhancement in M-theory in the rational approximation, but otherwise proceeding from \emph{first principles}, reads as follows:

\vspace{.2cm}

\label{OpenRationalPage}
\noindent \hypertarget{OpenRational}{{\bf Open Problem, rational version 2:}}
\vspace{-2mm}
\begin{quote}
{\it Which universal construction in rational super homotopy theory enhances the cyclification of the $\mathrm{M2}/\mathrm{M5}$-cocycle from
6-truncated to un-truncated rational twisted K-theory? }
\end{quote}

\vspace{.2cm}

\noindent {\bf Gauge enhancement explained.}
\rm
It is this version of the problem to which we present a solution.
First we explain and analyze two  relevant universal constructions
in homotopy theory:
\begin{enumerate}[{\bf (i)}]
\item \emph{Fiberwise stabilization} (in Sec. \ref{RationalParameterizedStableHomotopyTheory}) and
\item \emph{the $\mathrm{Ext}/\mathrm{Cyc}$-adjunction} (in Sec. \ref{TheAdjunction}).
\end{enumerate}
We then consider (in Sec. \ref{RationalHomotopyTypeOfATypeOrbispace})
the rational homotopy type of the \emph{A-type orbispace of the 4-sphere}, as in \eqref{SpacetimeAndATypeOrbispace} above,
and we apply the two aforementioned universal constructions to it (in Sec. \ref{RationalUnitOnAType}). Our main result, Theorem \ref{TwistedKTheoryInsideFiberwiseStabilizationOfATypeOrbispaceOf4Sphere}, shows that rational untruncated twisted K-theory
appears as a direct summand in the fiberwise stabilization of the $\mathrm{Ext}/\mathrm{Cyc}$-unit on the A-type orbispace of the 4-sphere.

\medskip
Since homotopy theory is immensely rich and computationally demanding (see. e.g. \cite{Ravenel03, HHR09}) one often simplifies
calculations by working in successive approximations, such as in the filtration by \emph{chromatic layers}. The first of these approximations is \emph{rational homotopy theory} (e.g. \cite{Hess06}) obtained by disregarding all torsion elements in homotopy and cohomology groups.
The model for \emph{rational parametrized stable homotopy theory} that we use in our computations had been conjectured in \cite[p. 20]{FSS16a}
and was subsequently worked out in \cite{Bra18}. This model allows us to reveal a deeper meaning
behind a curious dg-algebraic observation due to \cite{RS} (recalled as Prop. \ref{MinimalDGModels} below), culminating in our main
Theorem \ref{TwistedKTheoryInsideFiberwiseStabilizationOfATypeOrbispaceOf4Sphere}, below. However, since the gauge enhancement
mechanism that we present is obtained by \emph{universal constructions} (specifically the \emph{derived adjunctions} discussed in Sec. \ref{HomotopyTheory}), the lift of the mechanism beyond the rational approximation certainly exists, but is just much harder to analyze.

\medskip
Finally, in Sec. \ref{TheAppearanceFromMTheoryOfTheFundamentalD6AndD8} we apply Theorem \ref{TwistedKTheoryInsideFiberwiseStabilizationOfATypeOrbispaceOf4Sphere} to the double dimensional reduction of the fundamental
M-brane cocycles. We show that this solves \hyperlink{OpenRational}{\bf Open Problem, rational version 2} by making the D6- and
D8-brane cocycles appear and by exhibiting a single unified super-cocycle in rational un-truncated twisted K-theory:
$$
  \hspace{-0.1cm}
  \xymatrix@C=.5pt@R=18pt{
    \fbox{\hspace{-3mm}
      \begin{tabular}{c}
      \footnotesize  Fundamental branes
        \\
     \footnotesize   in M-theory
      \end{tabular}
    \hspace{-3mm}}
    \ar[dd]|{\mbox{\tiny
      \begin{tabular}{c}
        {\color{blue} enhanced}
        \\
        double dimensional reduction
      \end{tabular}
    }}
    &
    && &
    \mu_{{}_{\mathrm{M2}}}
    \ar@{|->}[ddl]
    \ar@{|->}[ddr]
    && &&
    \mu_{{}_{\mathrm{M5}}}
    \ar@{|->}[ddl]
    \ar@{|->}[ddr]
    &&
    &&&&
    \fbox{\hspace{-3mm}
      \begin{tabular}{c}
    \footnotesize    Super-cocycle
        \\
\footnotesize        on super-spacetime
        \\
\footnotesize        in rational
        cohomotopy
      \end{tabular}
   \hspace{-4mm} }
    \ar[dd]|{ \mbox{\tiny  \begin{tabular}{c} {\color{blue} fiberwise stabilized } \\ \tiny
    cyclification adjunction  \end{tabular} } }
    \\
    \\
    \fbox{\hspace{-3mm}
      \begin{tabular}{c}
      \footnotesize  Fundamental branes
        \\
       \footnotesize  in type IIA string theory
      \end{tabular}
    \hspace{-3mm}}
    &
    \mu_{{}_{\mathrm{D0}}}
    &&
    \mu_{{}_{\mathrm{F1}}}
    &&
    \mu_{{}_{\mathrm{D2}}}
    &&
    \mu_{{}_{\mathrm{D4}}}
    &&
    \mu_{{}_{\mathrm{NS5}}}
    &&
    {\color{blue} \mu_{{}_{\mathrm{D6}}}}
    &&
    {\color{blue} \mu_{{}_{\mathrm{D8}}}}
    &
    \fbox{\hspace{-3mm}
      \begin{tabular}{c}
 \footnotesize       Super-cocycle
        \\
   \footnotesize      on super-spacetime
        \\
   \footnotesize      in rational {\color{blue} un-}truncated
        \\
   \footnotesize      twisted K-theory
      \end{tabular}
    \hspace{-3mm}}
  }
$$
Notice how  all the folkloric ingredients recalled above do appear in this rigorous result,
albeit in a somewhat subtle way. First of all, the fact that fundamental branes and black branes
are closely related, while still crucially different (particularly in the matter of gauge enhancement),
is reflected by how super-cocycles interact with spacetime ADE singularities in the data specifying a real ADE-equivariant
cohomotopy class \cite{ADE}. Second, the claim that gauge enhancement in M-theory is connected to the
appearance of ADE singularities in spacetime is reflected here in the fact that the untruncated rational
twisted K-theory spectrum (which, as we have discussed, is the true rational coefficient for the gauge enhanced brane charges),
only appears from fiberwise stabilization of the \emph{equivariant}
4-sphere coefficient. This equivariant coefficient also induces the appearance of singular fixed point strata in
spacetime via equivariant enhancement:
$$
  \xymatrix@C=5em{
    &
    \fbox{ \hspace{-4mm}
      \begin{tabular}{c}
   \footnotesize     A-type $S^1$-action
        \\
\footnotesize        on coefficient 4-sphere
        \\
\footnotesize        of fundamental M-brane cocycle
      \end{tabular}
    \hspace{-4mm} }
    \ar@{|->}[ddr]|{ \mbox{
      \tiny
      \begin{tabular}{c}
        fiberwise stabilized
        \\
        $\mathrm{Ext}/\mathrm{Cyc}$-adjunction
        \\
        (Thm. \ref{TwistedKTheoryInsideFiberwiseStabilizationOfATypeOrbispaceOf4Sphere},
        Sec. \ref{TheAppearanceFromMTheoryOfTheFundamentalD6AndD8})
      \end{tabular}
    } }
    \ar@{|->}[ddl]|{ \mbox{
      \tiny
      \begin{tabular}{c}
        equivariant
        \\
        enhancement
        \\
        (\cite[Thm. 6.1, Sec. 2.2]{ADE})
      \end{tabular}
    } }
    \\
    \\
    \fbox{\hspace{-3mm}
      \begin{tabular}{c}
     \footnotesize   Black branes
        \\
      \footnotesize  at A-type singularities
      \end{tabular}
    \hspace{-3mm}}
    \ar@{<~>}[rr]^{ \mbox{\tiny  M-theory folklore } }
    &&
    \fbox{\hspace{-3mm}
      \begin{tabular}{c}
     \footnotesize   Gauge enhancement
        \\
\footnotesize        on M-branes
      \end{tabular}
    \hspace{-3mm}}
  }
$$

Of course it remains to lift our result beyond rational homotopy theory. However, we suggest
that the rational derivation of gauge enhancement in Theorem \ref{TwistedKTheoryInsideFiberwiseStabilizationOfATypeOrbispaceOf4Sphere}
and  Sec. \ref{TheAppearanceFromMTheoryOfTheFundamentalD6AndD8}
points to its own non-rational refinement. The reason is that the universal constructions that we have used also make sense
non-rationally -- they are just much harder to compute.
More concretely, we observe that the manner in which the rational version of twisted K-theory  appears below
is via a twisted, rational version of \emph{Snaith's theorem} (see Rem. \ref{InterpretationOfFiberwiseSuspensionOfATypeOrbispace}).
This theorem says that the K-theory spectrum $\mathrm{KU}$ is obtained from the suspension spectrum $\Sigma^\infty_{+} B S^1$
of the classifying space $B S^1$ by adjoining a multiplicative inverse of the Bott generator $\beta$:
$$
  \Sigma^\infty_{+} B S^1 [\beta^{-1}]
  \;\simeq_{\mathrm{swhe}}\;
  \mathrm{KU}
  \,.
$$
Rationally, Snaith's theorem is rather immediate, as is its rational twisted version (Ex. \ref{RationalSnaithTheorem} below)
that underlies our identification of rational twisted K-theory in Theorem
\ref{TwistedKTheoryInsideFiberwiseStabilizationOfATypeOrbispaceOf4Sphere}. Since rationalization is
the coarsest non-trivial approximation to full homotopy theory (and in this regard quite similar to taking the first derivative
of a non-linear function at a single point) it loses plenty of information.

\medskip
\emph{A priori}, what looks like twisted K-theory in the rational approximation could correspond non-rationally to
 different twisted cohomology theories (see \cite{S4, GS17, GS19} for discussions in this context). However, we do not just see  rational
twisted K-theory in isolation, but rather appearing after applying universal constructions  to the A-type orbispace of the 4-sphere.
For our main theorem on gauge enhancement, these
universal constructions are what really matter. 
Therefore, any non-rational lift of our gauge enhancement mechanism
should arise by the same universal construction, applied non-rationally, possibly in conjunction with other universal
constructions that are rationally invisible. This considerably constrains the possibilities. The conclusion we draw is that (fiberwise)
inversion of the Bott generator is a good candidate for lifting our gauge enhancement mechanism beyond the rational approximation.

\medskip
It is quite plausible that the gauge enhancement mechanism presented in Sec.
\ref{TheAppearanceFromMTheoryOfTheFundamentalD6AndD8}
generalizes beyond the rational approximation to a derivation of full twisted K-theory from degree 4 cohomotopy. We
may state this as the remaining part of the problem of gauge enhancement in M-theory:

\vspace{.2cm}

\noindent {\bf Open Problem, remaining part:}
\begin{quote}
{\it What is the non-rational lift of the gauge enhancement mechanism, by universal constructions in super homotopy theory,
from Sec. \ref{TheAppearanceFromMTheoryOfTheFundamentalD6AndD8}?}
\end{quote}

\vspace{.2cm}

\noindent We will return to this open problem elsewhere.

\newpage
\section{Two universal constructions in homotopy theory}
\label{HomotopyTheory}

Here we discuss two universal constructions in homotopy theory
(see e.g. \cite{Schreiber17a, Schreiber17b}),
which when applied to the A-type orbispace of the
4-sphere (Sec. \ref{TheATypeOrbispaceOfThe4Sphere}) reveal the mechanism of gauge enhancement on M-branes (Sec. \ref{TheMechanism}).
Firstly, in Sec. \ref{RationalParameterizedStableHomotopyTheory} we recall \emph{parametrized stable homotopy theory}
and recall an algebraic model for its rationalization from \cite{Bra18,Bra19b}, which enables us to effectively compute in this setting.
The main results of this section are
Theorem \ref{RationalParameterizedSpectradgModel}, which establishes differential-graded modules
as rational models for parametrized spectra, and Prop. \ref{FiberwiseSuspensionSpectrumdgModel} characterizing the fiberwise
stabilization adjunction in terms of these models.
In Sec. \ref{TheAdjunction} we demonstrate that forming cyclic loop spaces is one part of a
homotopy-theoretic adjunction, and we characterize the unit of the adjunction  (Theorem \ref{GCycExtAdjunction}). This universal construction is used in our formulation of double dimensional reduction.

\medskip
This section  may be read independently of the rest of the article and is of interest beside its application to M-brane phenomena.
Conversely, readers interested only in the application to M-theory and willing to accept our homotopy-theoretic machinery as a black box may be inclined to skip this section.

\subsection{Fiberwise stabilization}
\label{RationalParameterizedStableHomotopyTheory}

In \cite[p. 20]{FSS16a}, we found that the super $L_\infty$-algebraic F1/D$p$-brane cocycles organize
into a diagram as shown in (\hyperlink{fig:rattwistk}{b}) below. We further indicated that this ought to be thought of as the image in
supergeometric rational homotopy theory of a cocycle in twisted K-theory realized as a morphism of
parametrized spectra, as shown in (\hyperlink{fig:fulltwistk}{a}) below.
\begin{figure}[H]
\centering
\begin{subfigure}{0.35\textwidth}
\fbox{$
  \xymatrix{
    && \mathrm{KU}
     \ar[d]^{\mathrm{hofib}(p_\rho)}
    \\
    X \ar[dr]_\tau^{\ }="t" \ar@{-->}[rr]^c_{\ }="s" && \mathrm{KU} \dslash  BU(1) \ar[dl]^{p_\rho}
    \\
    & B^2 U(1)
    %
  }
$}
\caption{\hypertarget{fig:fulltwistk}A twisted K-theory cocycle according to \cite{AndoBlumbergGepner10, NSS12}.}
\end{subfigure}
\;\;\;\;\;\;\;\;\;\;\;\;
\begin{subfigure}{0.38\textwidth}
\fbox{$
  \xymatrix{
    && \mathfrak{l}(\mathrm{ku})
    \ar[d]^{\mathrm{hofib}(\phi)}
    \\
    \mathbb{R}^{9,1\vert \mathbf{16}+ \overline{\mathbf{16}}}
    \ar[rr]^{\mu_{{}_{F1/D}}^{\mathrm{IIA}}}
    \ar[dr]_{\mu_{{}_{F1}}}
    &&
    \mathfrak{l}( \mathrm{ku} \dslash  BU(1) )
    \ar[dl]^{\phi}
    \\
    & b^2 \mathbb{R}
  }
$}
\caption{\hypertarget{fig:rattwistk}The descended IIA F1/D$p$-brane cocycle according to \cite[Theorem 4.16]{FSS16a}.}
\end{subfigure}
\end{figure}

Roughly, a \emph{spectrum} is a kind of \emph{linearized} or \emph{abelianized} version of a topological space.
By the classical Brown Representability Theorem, maps into spectra represent cocycles in generalized cohomology
theories, such as K-theory. A \emph{parametrized spectrum} is a family of spectra that is parametrized in a
homotopy-coherent manner by a topological space (see \cite{MSi}). Roughly,  these are equivalent to bundles of spectra over the
parameter space, where maps into the total space of such a bundle represent cocycles in a \emph{twisted}
generalized cohomology theory.
$$
\xymatrix{
 {\mbox{Spectra}}
    \; \ar@{^{(}->}[rrr]^-{\mbox{ \tiny \begin{tabular}{c}Parametrized \\ over the point\end{tabular}}}
    &&&
    {\mbox{\begin{tabular}{c}Parametrized \\ spectra\end{tabular}}}
    \ar@{->>}[rrr]^-{\mbox{ \tiny \begin{tabular}{c}Underlying \\ parameter space \end{tabular}}}
    &&&
    {\mbox{Spaces}}}
$$

 Under Koszul duality, the conjecture of \cite{FSS16a} means, roughly, that there ought to be highlighted entries as
 in the following table. These entries unify Quillen--Sullivan's DG-models for rational homotopy theory
of topological spaces (the central result is recalled as Prop. \ref{SullivanEquivalence} below) with chain complex
models for stable rational homotopy theory (recalled as Prop. \ref{SchwedeShipleyEquivalence} below).

\vspace{3mm}
\hspace{-11mm}
 \begin{tabular}{|c||c|c|c|}
 \hline
 {\bf Homotopy theory}
    &
{\bf Stable}
   &
 {\bf Parametrized  stable  }
    &
    {\bf Plain}
\\ \hline \hline
{\bf Plain} &
  Spectra & Parametrized  spectra & Spaces
\\ \hline
{\bf Rational} & Cochain  complexes & {\it \color{blue} DG-modules} & DG-algebras
\\ \hline
{\bf Super rational} &
 Super cochain complexes     &
 {\it \color{blue} Super DG-modules} &
   Super DG-algebras  (``FDA''s)
   \\ \hline
\end{tabular}

\vspace{3mm}
This conjecture has been proven recently in \cite{Bra18} (see also the forthcoming articles \cite{Bra19a, Bra19b}).
In this paper we review those parts of the resulting
\emph{rational parametrized stable homotopy theory} that we need for the proof of Theorem \ref{TwistedKTheoryInsideFiberwiseStabilizationOfATypeOrbispaceOf4Sphere}.
The main results in this section are
Theorem \ref{RationalParameterizedSpectradgModel}, together with
Prop. \ref{FiberwiseSuspensionSpectrumdgModel}, which provide differential graded models
for  fiberwise suspension spectra.

%
%


To fix notation and conventions, we briefly recall some background on homotopy theory:

\begin{defn}[Classical homotopy theory (see e.g. \cite{Schreiber17a})]
  \label{ClassicalHomotopyCategories}
  We write
  \begin{enumerate}[{\bf (i)}]
    \item   $\mathrm{Ho}(\mathrm{Spaces})$ for the homotopy category of topological spaces, called the
    \emph{classical homotopy category}, which is the localization of the category of topological spaces
   at the \emph{weak homotopy equivalences} (those maps inducing isomorphisms on all homotopy groups,
   which we will denote by $\simeq_{\mathrm{whe}}$).


  \item $\mathrm{Ho}(\mathrm{Spaces})_{\mathbb{Q}, \mathrm{ft}}$
    for the full subcategory on homotopy types of \emph{finite rational type}, namely those spaces $X$ for
    which the homotopy groups $\pi_{k\geq 1}(X)$ are uniquely divisible (i.e., torsion-free and divisible),
    and $H^1(X,\mathbb{Q})$ and $\pi_{k\geq 2}(X) \otimes \mathbb{Q}$
    are finite-dimensional $\mathbb{Q}$-vector spaces;

  \item $\mathrm{Ho}(\mathrm{Spaces})_{\mathbb{Q},\mathrm{nil},\mathrm{ft}}\subset \mathrm{Ho}(\mathrm{Spaces})_{\mathbb{Q}, \mathrm{ft}}$ for the full subcategory of homotopy types that are moreover \emph{nilpotent}, meaning those $X$ that are connected
   and whose fundamental group $\pi_1(X)$ is a nilpotent group with each $\pi_n(X)$ a nilpotent $\pi_1(X)$-module.
   \end{enumerate}
\end{defn}

A nilpotent action is one whose iteration a certain number of times is the identity. Being simply-connected is a special case.
Similar constructions hold in the parametrized case (see \cite{Crab, MSi}):

\begin{defn}[Parametrized homotopy theory (e.g. {\cite[Sec. 1.1]{Bra18}})]
  \label{ParamaterizedHomotopyTheory}
  For any topological space $X$, we write
  \begin{enumerate}[{\bf (i)}]

    \item $\mathrm{Ho}\big( \mathrm{Spaces}_{/X} \big)$ for the homotopy category
     of spaces over $X$, that is, of spaces equipped with a map to $X$. We denote
      such objects by $[Y\xrightarrow{\pi} X]$;
      this is the  \emph{$X$-parametrized classical homotopy category};

    \item $\mathrm{Ho}\big(\mathrm{Spaces}_{\dslash X}\big)$
    for the homotopy category of spaces over $X$ that are equipped with a section. We denote such objects by
    $[X\xrightarrow{\sigma}Y\xrightarrow{\pi} X]$, where it is understood that the composite
    $\pi\circ \sigma = \mathrm{id}_X$; morphisms are maps of spaces over $X$ respecting the sections up to homotopy.

  \end{enumerate}
%
%
\end{defn}

\begin{remark}
  In the special case that $X \simeq \ast$ is the point, we simply have that
  \begin{enumerate}[{\bf (i)}]
    \item $\mathrm{Ho}\big( \mathrm{Spaces}_{/\ast}\big) \simeq \mathrm{Ho}\big( \mathrm{Spaces}\big)$
     is the classical homotopy category (Def. \ref{ClassicalHomotopyCategories}); and

    \item $\mathrm{Ho}\big( \mathrm{Spaces}_{\dslash \ast}\big)$
      is the homotopy category of \emph{pointed} spaces.
  \end{enumerate}
In general, it is sensible to think of $\mathrm{Ho}\big(\mathrm{Spaces}_{\dslash X}\big)$ as the homotopy category
  of \emph{$X$-parametrized pointed spaces}.
\end{remark}

Recall that there are two fundamental constructions associated to any pointed space $Y$: its based loop space $\Omega_\ast Y$
and its (reduced) suspension $\Sigma_\ast Y$.
There are analogous constructions in the parametrized setting over a fixed base space $X$ (see \cite{Crab, MSi}).
These constructions compute loop spaces and reduced suspensions fiberwise over $X$:
\begin{prop}[Looping and suspension]
  \label{ReducedSuspension}
 If $Y$ is
parametrized over a space $X$ (Def. \ref{ParamaterizedHomotopyTheory}), we can form its fiberwise loop space $\Omega_X Y$.
This construction is functorial and admits a left adjoint $\Sigma_X$, called the \emph{fiberwise
reduced suspension}:
$$
  \xymatrix{
    \mathrm{Ho}
    \left(
      \mathrm{Spaces}_{\dslash X}
    \right)
    \;
    \ar@{<-}@<+6pt>[rr]^-{\Omega_X}
    \ar@{->}@<-6pt>[rr]_-{\Sigma_X}^{\top}
    && \;
    \mathrm{Ho}
    \left(
     \mathrm{Spaces}_{\dslash X}
    \right).
  }
$$
\end{prop}
A fiberwise loop space $\Omega_X Y$ carries the structure of a fiberwise homotopical group (or, more precisely,
a fiberwise \emph{grouplike $A_\infty$-space})
given by concatenating loops. In the case of a fiberwise \emph{double} loop space $\Omega^2_X Y = \Omega_X \Omega_X Y$,
the Eckmann--Hilton argument implies an additional fiberwise first-order homotopy commutativity structure given by twisting
based loops around each other. Increasing the fiberwise loop order, as $n$ goes to infinity the fiberwise homotopcial group structure
on $\Omega^n_Y X$ becomes increasingly homotopy-commutative.
The $n\to \infty$ limit therefore provides a useful heuristic for obtaining \emph{$X$-parametrized homotopical abelian groups}.

\medskip
One way to formalize the idea that increasing loop order \emph{stabilizes} to abelian homotopy theory is to exhibit
a homotopy category for which the adjunction of Prop. \ref{ReducedSuspension}
is an equivalence. In the unparametrized setting, the homotopy category obtained in this manner is the \emph{stable homotopy category},
the objects of which are called \emph{spectra}. By definition, a spectrum $P$ is a sequence of pointed topological spaces
$\{P_n\}_{n\in \mathbb{N}}$ equipped with structure maps $\Sigma P_n \to P_{n+1}$.
To each spectrum $P$ is naturally assigned a sequence of abelian groups $\{\pi_k(P)\}_{k\in \mathbb{Z}}$ called the
\emph{stable homotopy groups} of $P$.
There is a natural notion of a morphism of spectra with respect to which the assignment of stable homotopy groups is functorial.
A map of spectra is a \emph{stable weak equivalence} if it induces an isomorphism on all stable homotopy groups, and localizing
the category of spectra at the class of stable weak equivalences produces the \emph{stable homotopy category}.
Every spectrum is stably weakly equivalent to a spectrum $P$ for which the adjuncts of the structure maps are weak homotopy
equivalences $P_n \cong_{\mathrm{whe}} \Omega P_{n+1}$.
Spectra of this special type are called \emph{$\Omega$-spectra}, and for an $\Omega$-spectrum $P$ there are weak homotopy
equivalences $P_k \cong \Omega^n P_{n+k}$ for all $n,k\geq 0$, exhibiting each $P_k$ as an infinite loop space
(or \lq\lq homotopical abelian group''). Finally, the looping and suspension operations on spaces prolong to spectra and for any
spectrum $P$ we have natural isomorphisms
\[
\pi_{k+1}(\Sigma P)\cong \pi_k (P) \cong \pi_{k-1}(\Omega P)
\]
for all $k\in \mathbb{Z}$. Spectra are primarily of interest since they represent generalized cohomology theories on topological spaces.
It is a consequence of the previously-mentioned Brown representability theorem that \emph{all} generalized cohomology theories arise
 from spectra in this way. See \cite{Schreiber17b} for a review of classical stable homotopy theory.

\medskip
The above notions for s[aces have analogues for spectra.

\begin{defn} [Stable homotopy theory]
  \label{StableHomotopyTheory}
  We write
  \begin{enumerate}[{\bf (i)}]
  \item $\mathrm{Ho}( \mathrm{Spectra} )$ for the \emph{homotopy category of spectra}, also
   called the \emph{stable homotopy category};
  \item $\mathrm{Ho}(\mathrm{Spectra})_{\mathbb{Q}}$
  for the \emph{rational stable homotopy category}, hence the localization of the category of spectra
  at the maps inducing isomorphisms on rationalized stable homotopy groups $\pi_\ast \otimes \mathbb{Q}$;
  \item $\mathrm{Ho}(\mathrm{Spectra})_{\mathbb{Q},\mathrm{ft}}$ for the full subcategory
    on those spectra $P$ which are of \emph{finite rational type}, meaning that $\pi_k(P)\otimes \mathbb{Q}$ is a
    finite-dimensional $\mathbb{Q}$-vector space for all $k\in\mathbb{Z}$;
  \item $\mathrm{Ho}(\mathrm{Spectra})_{\mathbb{Q},\mathrm{bbl}}$ for the full subcategory of those
  spectra which are \emph{rationally bounded below}, hence whose rationalized stable homotopy groups all vanish below
    some given dimension.
  \end{enumerate}
\end{defn}

We are mainly concerned here with \emph{parametrized} spectra, which are families of spectra parametrized by a topological
space. Equivalently, these are bundles of spectra over that base space. Given a family of spectra $P$ parametrized by a topological
space $X$, for any point $x\in X$ we can extract a \emph{stable homotopy fiber} $\mathbb{R}x^\ast P$ in a way that depends
 functorially on $P$ (and on $x$ in the appropriate homotopy-coherent sense).
There is a natural notion of a map of $X$-parametrized spectra, and we declare a map $P\to Q$ to be a \emph{fiberwise stable
equivalence} if the induced map $\mathbb{R}x^\ast Q\to \mathbb{R}x^\ast P$ is a stable weak equivalence for all $x\in X$.
For a rigorous approach to the theory using simplicial homotopy theory see \cite{Bra18, Bra19a}.

\begin{defn}[Parametrized stable homotopy theory]
  \label{ParamerizedStableHomotopyTheory}
  For a fixed parameter space $X$, we write
  \begin{enumerate}[{\bf (i)}]
    \item $\mathrm{Ho}\left( \mathrm{Spectra}_{X} \right)$ for the \emph{homotopy category of
    spectra parametrized by $X$}, hence the localization of the category of $X$-parametrized spectra at the fiberwise stable equivalences;
    \item $\mathrm{Ho}\left( \mathrm{Spectra}_{X} \right)_{\mathbb{Q}}$ for the \emph{rational homotopy category of
    spectra parametrized by $X$}, hence the localization of category of $X$-parametrized spectra at the maps inducing
    isomorphisms on rationalized stable homotopy groups on all homotopy fiber spectra;

    \item $\mathrm{Ho}\left( \mathrm{Spectra}_{X} \right)_{\mathbb{Q}, \mathrm{ft},\mathrm{bbl}}\subset \mathrm{Ho}\left( \mathrm{Spectra}_{X} \right)_{\mathbb{Q}}$ for the full subcategory on those $X$-spectra $P$ whose homotopy fiber spectra $\mathbb{R}x^\ast(P)$ are
    \emph{of finite rational type} and  \emph{are rationally bounded below} for all $x\in X$.
  \end{enumerate}
\end{defn}

A key point about parametrized spectra (Def. \ref{ParamerizedStableHomotopyTheory}) is that they represent \emph{twisted}
generalized cohomology theories, generalizing the fact that plain spectra represent generalized cohomology theories (e.g. \cite{ABGHR14}).
We will be particularly interested in the (rational version of the) twisted cohomology theory called twisted K-theory
 (Lemma \ref{TwistedKModel} below).
The fundamental relation between unstable  and stable homotopy theory in the parametrized setting is captured by the following:

\begin{prop}[Fiberwise stabilization adjunction]
 \label{AdjunctionStabilization}
 For any space $X$ there are pairs of adjoint functors
 \begin{equation}
   \label{FiberwiseStabilizationAdjunction}
   \xymatrix{
     \mathrm{Ho}\big( \mathrm{Spaces}_{/X}\big)
     \ar@/^1.4pc/@{<-}@<+8pt>[rrrr]^{\Omega^\infty_{X}}
     \ar@/_1.4pc/@{->}@<-8pt>[rrrr]_{\Sigma^\infty_{+,X}}
     \ar@{<-}@<+6pt>[rr]
     \ar@{->}@<-6pt>[rr]_-{(-)_{+,X}}^-{\bot}
     &&
     \mathrm{Ho} \big(\mathrm{Spaces}_{\dslash X}  \big)
     \ar@{<-}@<+6pt>[rr]^-{\Omega^\infty_X}
     \ar@{->}@<-6pt>[rr]_-{\Sigma^\infty_X}^-{\bot}
     &&
     \mathrm{Ho}\big( \mathrm{Spectra}_{X}  \big)
   }
 \end{equation}
 between the classical parametrized homotopy categories (Def. \ref{ParamaterizedHomotopyTheory})
 and the parametrized stable homotopy category (Def. \ref{ParamerizedStableHomotopyTheory}).
 Here $(-)_{+,X}$ adjoins a copy of $X$, e.g. $[Y\to X]\mapsto[X\to X\coprod Y\to X]$, while $\Omega^\infty_{X}$ sends an $X$-parametrized spectrum to its \emph{fiberwise infinite loop space}, and the operations $\Sigma^\infty_X$ and/or $\Sigma^\infty_{+,X}$
 are called forming \emph{fiberwise suspension spectra}.

 Moreover, the adjunction \eqref{FiberwiseStabilizationAdjunction} stabilizes the looping/suspension adjunction from Prop. \ref{ReducedSuspension} in that
 there is a diagram, commuting up to natural isomorphism, as follows
 $$
   \xymatrix@R=1.6em@C=5em{
    \mathrm{Ho}\big(\mathrm{Spaces}_{\dslash X}\big)
    \ar@{<-}@<+6pt>[rr]^-{\Omega_X}
    \ar@{->}@<-6pt>[rr]_-{\Sigma_X}^-{\top}
    \ar@{<-}@<+6pt>[dd]^-{\Omega^\infty_X}
    \ar@{->}@<-6pt>[dd]_-{\Sigma^\infty_X}^-{ \dashv }
    &&
    \mathrm{Ho}\big(\mathrm{Spaces}_{\dslash X}\big)
    \ar@{<-}@<+6pt>[dd]^-{\Omega^\infty_X}
    \ar@{->}@<-6pt>[dd]_-{\Sigma^\infty_X}^-{ \dashv }
    \\
    \\
    \mathrm{Ho}\big(\mathrm{Spectra}_X\big)
    \ar@{<-}@<+6pt>[rr]^-{\Omega_X}
    \ar@{->}@<-6pt>[rr]_-{\Sigma_X}^-{\simeq}
    &&
    \mathrm{Ho}\big(\mathrm{Spectra}_X\big).
   }
 $$
\end{prop}
\begin{remark}[Units]
We denote the unit morphism of the adjunction \eqref{FiberwiseStabilizationAdjunction}
on $Y \in \mathrm{Ho}\big( \mathrm{Spaces}_{/X} \big)$ by
\begin{equation}
  \label{UnitOfStabilizationAdjunction}
  \xymatrix@R=1em{
    Y
    \ar[dr]
    \ar[rr]^-{\mathrm{st}_X(Y)}
    &&
    \Omega^\infty_X \Sigma^\infty_{+,X}(Y)\;.
    \ar[dl]
    \\
    & X
  }
\end{equation}
\end{remark}
Both the classical and stable homotopy categories are extremely rich mathematical settings. In order to get a better handle
on these categories, one may filter them in various ways so as
to study (stable) homotopy types in controlled approximations. A particularly useful approximation of this sort is provided by
rational homotopy theory, which discards all torsion information carried by the homotopy groups.
The main reason that rational homotopy theory is so tractable is that both the unstable and stable variants can be completely
described in terms of algebraic data.
We recall this as Prop. \ref{SullivanEquivalence} and Prop. \ref{SchwedeShipleyEquivalence} below, but first we must recall
some terminology:
\begin{defn}[DG-algebraic homotopy theory]
  \label{dgAlgebrasAnddgModules}
  We write
\begin{enumerate}[{\bf (i)}]
  \item $\mathrm{Ho}(\mathrm{DGCAlg})$ for the \emph{homotopy category of connective differential graded (unital) commutative
  algebras (DG-algebras)} over $\mathbb{Q}$. The connectivity condition means that the underlying cochain complex vanishes identically in negative degree, and working in the homotopy category means that we localize the category of DG-algebras with respect to the class of quasi-isomorphisms;
  \item $\mathrm{Ho}(\mathrm{DGCAlg})_{\mathrm{cn}}$ for the full subcategory on the \emph{(cohomologically) connected}
  DG-algebras; those $A$ for which the algebra unit $\mathbb{Q} \to A$
   induces an isomorphism $\mathbb{Q} \simeq H^0(A)$;
  \item $\mathrm{Ho}\left(\mathrm{DGCAlg}\right)_{\mathrm{ft}}$ for the full subcategory of DG-algebras $A$ of \emph{finite type},
  so that $A$ is cohomologically connected and quasi-isomorphic to a DG-algebra that is degreewise finitely generated.
\end{enumerate}
\end{defn}

\begin{remark}[Differentials]
{\bf (i)} We write \lq\lq DG'' throughout to indicate that we are working with \emph{co}homological grading conventions, so that in particular all differentials  increase degree by $+1$.

\item {\bf (ii)} We will write free graded commutative algebras (without differentials) as polynomial algebras
$$
  \mathbb{Q}[\alpha_{k_1}, \beta_{k_2}, \dotsc]
  \,,
$$
where the subscript on the generator will always indicate its degree. Differentials on such free algebras are fully determined by
their actions on generators by the graded Leibniz rule, and so we will denote DG-algebras obtained this way by
$$
  \mathbb{Q}[\alpha_{k_1}, \beta_{k_2}, \dotsc]\Bigg/
  \left(
    \begin{aligned}
      d \alpha_{k_1} & = \cdots
      \\
      d \beta_{k_2} & = \cdots
      \\
      & \;\;\vdots
    \end{aligned}
  \right).
$$
\end{remark}
The main result of Sullivan's approach to rational homotopy theory \cite{Su} (a detailed treatment is the subject of the monograph \cite{BG}) is a characterization of certain well-behaved rational homotopy types in terms of DG-algebras:

\newpage
\begin{prop}[DG-models for rational homotopy theory (e.g. \cite{BG,Hess06})]
  \label{SullivanEquivalence}

  \item {\bf (i)} There is an adjunction
  $$
    \xymatrix{
      \mathrm{Ho}(\mathrm{Spaces})
      \ar@{->}@<+6pt>[rr]^-{\mathcal{O}}
      \ar@{<-}@<-6pt>[rr]_-{\mathcal{S}}^-{\bot}
      &&
      \mathrm{Ho}( \mathrm{DGCAlg})^{\mathrm{op}}
    }
  $$
  between the classical homotopy category of topological spaces
  (Def. \ref{ClassicalHomotopyCategories}) and the opposite of the
  homotopy category of DG-algebras (Def. \ref{dgAlgebrasAnddgModules}), where
  $\mathcal{O}$ denotes the derived functor of forming the DG-algebra of rational polynomial differential
  forms.

 \item  {\bf (ii)} This adjunction restricts to an equivalence of categories
  \begin{equation}
    \label{SullivanEquivalenceAdjunction}
    \xymatrix{
      \mathrm{Ho}(\mathrm{Spaces})_{\mathbb{Q}, \mathrm{nil}, \mathrm{ft}}
      \ar@{->}@<+6pt>[rr]^-{\mathcal{O}}
      \ar@{<-}@<-6pt>[rr]_-{\mathcal{S}}^-{\simeq}
      &&
      \mathrm{Ho}( \mathrm{DGCAlg})^{\mathrm{op}}_{\mathrm{ft}}
    }
  \end{equation}
  between the rational homotopy category of nilpotent spaces of finite type (Def. \ref{ClassicalHomotopyCategories})
  and the homotopy category of DG-algebras of finite type (Def. \ref{dgAlgebrasAnddgModules}).

  \item {\bf (iii)} The rational cohomology of a space is computed by
  the cochain cohomology of any one of its DG-algebra models:
  $$
    H^\bullet(X,\mathbb{Q}) \;\simeq\; H^\bullet(\mathcal{O}(X))
    \,.
  $$

  \item {\bf (iv)}   Under  the equivalence of \eqref{SullivanEquivalenceAdjunction}, every space $X$ on the left
    has a model by a \emph{minimal} DG-algebra, whose underlying graded algebra
    is the free graded commutative algebra on the dual rational homotopy groups of $X$:
    $$
      \mathcal{O}(X) \;\simeq\; \mathbb{Q}\left[ (\pi_\bullet(X) \otimes \mathbb{Q} )^\ast\right] \big/ \big(d(\,\cdots) = (\,\cdots)\big)
      \,.
    $$
    Minimal models are unique up to isomorphism, with the isomorphism between any two minimal models unique up to homotopy.
\end{prop}

\begin{example}[Minimal model for the 3-sphere]
  \label{MinimalDgcAlgebraModelFor3Sphere}
  The minimal model for the 3-sphere is given by
  $
    \mathcal{O}(S^3)
    \;\simeq\;
    \mathbb{Q}[h_3]\,/ \left(
        d h_3  = 0
    \right).
  $
  Observe that $S^3$ is rationally indistinguishable from a $K(\mathbb{Z},3)$.
  This is true for all odd-dimensional spheres, since $\pi_\ast (S^{2k+1})\otimes \mathbb{Q}$ is a one-dimensional graded
  vector space concentrated in dimension $2k+1$.
\end{example}
\begin{example}[Minimal  model for the 4-sphere]
  \label{MinimalDgcAlgebraModelFor4Sphere}
In contrast, the minimal model for the 4-sphere is not free, and is given by
  $$
    \mathcal{O}(S^4)
    \;\simeq\;
    \mathbb{Q}[ \omega_4, \omega_7  ]\Big/ \left(
      {\begin{aligned}
        d \omega_4 & = 0 \
        \\[-2mm]
        d \omega_7 & = -\tfrac{1}{2} \omega_4 \wedge \omega_4
      \end{aligned}}
    \right).
  $$
\end{example}
\begin{example}[Minimal model for $B S^1$]
  \label{MinimalDGCAlgebraModelForClassifyingSpace}
  The minimal model for the classifying space $B S^1$ of the circle is given by
  $
    \mathcal{O}(B S^1)
    \;\simeq\;
    \mathbb{Q}[\omega_2]\,\big/ \left(
        d \omega_2  = 0
    \right).
  $
  Note that this is decidedly not a minimal DGC-algebra model for the $2$-sphere, since $\pi_2(S^2)\cong \mathbb{Z}\cong\pi_3(S^2)$.
\end{example}

We will model the rational homotopy theory of parametrized spectra in terms of DG-modules over DG-algebras:
\begin{defn}[DG-modular homotopy theory]
  \label{RatioaldgModules}
  Given a DG-algebra $A$ (Def. \ref{dgAlgebrasAnddgModules}), we write
\begin{enumerate}[{\bf (i)}]
   \item $A\mbox{-}\mathrm{Mod}$ for the category of unbounded DG-modules over $A$;
   \item $A\mbox{-}\mathrm{Mod}_{\mathrm{ft}}$ for the subcategory of those DG-$A$-modules of \emph{finite type}, hence those whose
     cochain cohomology is finite-dimensional in each degree;
   \item $A\mbox{-}\mathrm{Mod}_{\mathrm{bbl}}$ for the subcategory of \emph{bounded below} DG-$A$-modules, namely those whose
     cochain cohomology vanishes identically below some degree;
   \item $\mathrm{DGCAlg}^{A/}$ for the slice category of \emph{$A$-algebras}.
   The objects of this category are simply DG-algebras equipped with an algebra morphism from $A$, which we will frequently denote $B \leftarrow A\colon \pi^\ast$.

   \item $\mathrm{DGCAlg}_{\dslash A}$ for the category of
   \emph{augmented $A$-algebras}, whose objects are diagrams of DG-algebras
   \[
   \xymatrix{
     A\ar@{<-}[r]^{\;\;\sigma^\ast}
     &
      B \ar@{<-}[r]^{\;\;\pi^\ast}
      & A}
    \] such that $\sigma^\ast \circ \pi^\ast$ is the identity on $A$.
%
%
%
    The morphism $\sigma^\ast$ is called the \emph{augmentation}, and its kernel  $\mathrm{ker}(\sigma^\ast) \in A\mbox{-}\mathrm{Mod}$ is the \emph{augmentation ideal}.

    \item Passing to the homotopy category in any of the above cases {\bf(i)}-{\bf(v)} (e.g. $A\mathrm{\mbox{-}Mod}\mapsto \mathrm{Ho}\big(A\mathrm{\mbox{-}Mod}\big)$ means that we localize with respect to the class of quasi-isomorphisms.
  \end{enumerate}
\end{defn}

\begin{remark}
We emphasise that we always work with \emph{connective} DG-algebras, however modules over these algebras may be unbounded.
\end{remark}

The stable analogue of Sullivan's rational homotopy theory equivalence Prop. \ref{SullivanEquivalence} {\bf (i)} in the unparametrized context is the following:
\begin{prop}[DG-models for rational stable homotopy theory]
  \label{SchwedeShipleyEquivalence}
  There is an equivalence
  $$
    \xymatrix{
      \mathrm{Ho}(\mathrm{Spectra})_{\mathbb{Q},\mathrm{ft}}
      \ar@{->}@<+6pt>[rr]^-{}
      \ar@{<-}@<-4pt>[rr]_{}^-{\simeq}
      &&
      \mathrm{Ho}(\mathrm{Ch}(\mathbb{Q}))^{\mathrm{op}}_{\mathrm{ft}}
    }
  $$
  between the rational homotopy category of spectra of finite type (Def. \ref{ClassicalHomotopyCategories}) and the opposite homotopy category
  of rational cochain complexes of finite type
  (as in Def. \ref{RatioaldgModules} for $A=\mathbb{Q}$).
\end{prop}
\begin{proof}[Sketch of proof.]
This is a well-known fact in stable homotopy theory (and used extensively in differential cohomology, see \cite{GS17}),
but we sketch a proof for completeness.
Passing to the rational stable homotopy category is implemented by smashing with the Eilenberg-MacLane spectrum
$H\mathbb{Q}$, so that $ \mathrm{Ho}(\mathrm{Spectra})_\mathbb{Q} \cong \mathrm{Ho}(H\mathbb{Q}\mathrm{\mbox{-}Mod})$.
But the latter homotopy category is equivalent to the homotopy category of rational \emph{chain} complexes
$\mathrm{Ho}(H\mathbb{Q}\mathrm{\mbox{-}Mod})\cong \mathrm{Ho}(\mathrm{ch}(\mathbb{Q}))$ (see \cite{Shipley07}).
Under the assumption of finite type, dualizing then gives
\[
\mathrm{Ho}(\mathrm{Spectra})_{\mathbb{Q},\mathrm{ft}} \cong \mathrm{Ho}(\mathrm{Ch}(\mathbb{Q}))^\mathrm{op}_\mathrm{ft}.
\]
An alternative proof which does not appeal to dualization is given in \cite[Sec. 2.7.4]{Bra18}; see also \cite{Bra19b}.
%
\end{proof}
\begin{remark}[Operations]
\label{SuspensionandShifting}
Under the above equivalence, suspension of spectra corresponds to shifting the corresponding cochain complex up by one.
Looping corresponds to the cochain complex down by one. 
\end{remark}

The unification of Prop. \ref{SullivanEquivalence} with Prop. \ref{SchwedeShipleyEquivalence} established in \cite{Bra18, Bra19b} is:
\begin{theorem}[DG-models for rational parametrized spectra]
   \label{RationalParameterizedSpectradgModel}
   For any connective DG-algebra $A$ (Def. \ref{dgAlgebrasAnddgModules}),   there is a pseudo-natural transformation
   $$
     \xymatrix{
       \mathrm{Ho}\big( \mathrm{Spectra}_{\mathcal{S}(A)}\big)
       \ar[rr]^-{\mathcal{M}_A}
       &&
       \mathrm{Ho}\left( A\mbox{-}\mathrm{Mod}\right)^{\mathrm{op}}
     }
   $$
   from the homotopy category of parametrized spectra (Def. \ref{ParamerizedStableHomotopyTheory})
   parametrized by the rational space $\mathcal{S}(A)$ (Prop. \ref{SullivanEquivalence})
   to the opposite homotopy category of DG-modules over $A$ (Def. \ref{RatioaldgModules})
   with the following properties:
   \begin{enumerate}
     \item
      There is a factorization over the rational homotopy category of parametrized spectra:
       $$
         \xymatrix{
           \mathrm{Ho}\big(
             \mathrm{Spectra}_{\mathcal{S}(A)}
           \big)
           \ar[dr]_{\mathcal{M}_A}
           \ar[r]
           &
           \mathrm{Ho}\big(
             \mathrm{Spectra}_{\mathcal{S}(A)}
           \big)_{\mathbb{Q}}
           \ar[d]
           \\
           &
           \mathrm{Ho}\left(
             A\mbox{-}\mathrm{Mod}
           \right)^{\mathrm{op}}.
         }
       $$
     \item
       If the space $\mathcal{S}(A)$ is simply-connected, then $\mathcal{M}_A$ restricts to an equivalence of rational homotopy categories
       $$
         \xymatrix{
           \mathrm{Ho}\big( \mathrm{Spectra}_{\mathcal{S}(A)} \big)_{\mathbb{Q},\mathrm{ft}, \mathrm{bbl}}
           \ar[rr]^-{\mathcal{M}_A}_-\simeq
           &&
           \mathrm{Ho}\big(
             A\mbox{-}\mathrm{Mod}
           \big)^{\mathrm{op}}_{\mathrm{ft}, \mathrm{bbl}}
         }\,
       $$
       between finite-type, bounded below objects.
     \item
       For $A = \mathbb{Q}$, this extends to the
       equivalence of rational stable homotopy theory from Prop. \ref{SchwedeShipleyEquivalence}:
       $$
         \xymatrix{
           \mathrm{Ho}\big( \mathrm{Spectra} \big)_{\mathbb{Q},\mathrm{ft}}
           \ar[rr]^-{\mathcal{M}_A}_-\simeq
           &&
           \mathrm{Ho}\big(
             \mathrm{Ch}\left(\mathbb{Q}\right)
           \big)^{\mathrm{op}}_{\mathrm{ft}}.
         }
       $$
   \end{enumerate}
\end{theorem}
\begin{proof}[Sketch of proof.]
The functor $\mathcal{M}_A$ is constructed by stabilizing the Sullivan--de Rham adjunction of Prop.  \ref{SullivanEquivalence} {\bf (i)}.
The various properties of the stabilized functor are established in \cite{Bra18, Bra19b}.
Specifically, pseudo-naturality is \cite[Cor. 2.7.26]{Bra18}, the first item is \cite[Cor. 2.7.31]{Bra18},
  the second item is \cite[Theorem 2.7.42]{Bra18}, and the final item is \cite[Rem. 2.7.44]{Bra18}.
\end{proof}
\begin{remark}[Fiberwise operations]
\label{SuspensionandShiftingParam}
The obvious parametrized analogue of Rem. \ref{SuspensionandShifting} holds, namely \emph{fiberwise} suspension
of a parametrized spectrum corresponds to shifting the corresponding DG-module up by one, and forming fiberwise loop
spaces corresponds to shifting the DG-module down by one.
\end{remark}

In the proof of Theorem \ref{TwistedKTheoryInsideFiberwiseStabilizationOfATypeOrbispaceOf4Sphere} below, we will need to know explicitly how the functor $\mathcal{M}_A$ behaves on fiberwise suspension spectra.
\begin{prop}[DG-models for fiberwise suspension spectra {\cite{Bra18, Bra19b}}]
  \label{FiberwiseSuspensionSpectrumdgModel}
  Let $A$ be a DG-algebra (Def. \ref{dgAlgebrasAnddgModules})
  and let
  $$
  \big[
  \xymatrix{
    Y\ar[r]^-{\pi} & \mathcal{S}(A)
    }
  \big]
    \;\in\;
    \mathrm{Ho}\big( \mathrm{Spaces}_{/\mathcal{S}(A)} \big)_{\mathbb{Q}}
  $$
  be a space over (Def. \ref{ParamaterizedHomotopyTheory}) the rational space $\mathcal{S}(A)$ determined by $A$ via Prop. \ref{SullivanEquivalence}.
  Then, after passing to DG-models via $\mathcal{M}_A$ (Theorem \ref{RationalParameterizedSpectradgModel}),
  the fiberwise suspension spectrum (Prop. \ref{AdjunctionStabilization}) is modeled by $\mathcal{O}(Y)$ (according to Prop. \ref{SullivanEquivalence}), regarded as an $A$-module via $\pi^\ast$; that is,
  \begin{equation}
    \label{FormulaFiberwiseSuspensionSpectra}
    \mathcal{M}_A
    \big(
      \Sigma^\infty_{+,\mathcal{S}(A)}(Y)
    \big)
    \;\simeq\;
    \mathcal{O}(Y)
    \,.
  \end{equation}
\end{prop}
\begin{proof}[Sketch of proof.]
  In the proof of Theorem \ref{RationalParameterizedSpectradgModel} (see \cite[Lemma 2.7.25]{Bra18}, also \cite{Bra19b}), we encounter the commutative diagram of left Quillen functors:
  $$
    \xymatrix@R=1.2em{
      \mathrm{Spaces}_{/\mathcal{S}(A)}
      \ar[dd]_{\mathcal{O}_A}
      \ar@/^2pc/[rrrrrr]^-{\Sigma^\infty_{+,\mathcal{S}(A)}}
      \ar[rrr]^-{(-)_{+,\mathcal{S}(A)}}
      &&&
      \mathrm{Spaces}_{\dslash \mathcal{S}(A)}
      \ar[dd]_-{ \mathcal{O}_A }
      \ar[rrr]^-{ \Sigma^\infty_{\mathcal{S}(A)} }
      &&&
      \mathrm{Spectra}_{\mathcal{S}(A)}
      \ar[dd]^-{\mathcal{M}_A}
      \\
      \\
      \big( \mathrm{DGCAlg}^{A/} \big)^{\mathrm{op}}
      \ar[rrr]_-{ (-)\oplus A }
      &&&
       \big(\mathrm{DGCAlg}_{\dslash A} \big)^{\mathrm{op}}
      \ar[rrr]_-{\mathrm{aug}_A}
      &&&
      A\mbox{-}\mathrm{Mod}^{\mathrm{op}}
    }
  $$
  Here the left and bottom horizontal functors send
  $$
    {
      \left[\hspace{-2mm}
      \raisebox{20pt}{
      \xymatrix@C=4pt{
        Y
        \ar[d]^\pi
        \\
        \mathcal{S}(A)
      }
      }
      \right]
    }
        \xymatrix{\ar@{|->}[r]^-{\mathcal{O}_A} &}
      {
        \left[\hspace{-2mm}
        {
        \raisebox{40pt}{
          \xymatrix@C=4pt{
            \mathcal{O}(Y)
            \\
            \mathcal{O}(\mathcal{S}(A))
            \ar[u]^{\pi^{\ast}}
            \\
            A
            \ar[u]^\eta
          }
        }}
        \right]
      }
     \;
     \xymatrix{\ar@{|->}[r]^-{(-)\oplus A} &}
     \;
      {
        \left[\hspace{-2mm}
        {
        \raisebox{40pt}{
          \xymatrix@C=4pt{
            A
            \\
            \mathcal{O}(Y) \oplus A
            \ar[u]^{0 \oplus \mathrm{id}}
            \\
            A
            \ar[u]^{ (\pi^\ast \circ \eta) \oplus \mathrm{id}}
          }
        }}
        \right]
      }
       \xymatrix{\ar@{|->}[r]^-{\mathrm{aug}_A} &}
      \mathrm{ker}(0 \oplus \mathrm{id})
      =
      \mathcal{O}(Y)\;,
  $$
from which the assertion follows.
\end{proof}

Theorem \ref{RationalParameterizedSpectradgModel} allows us to make use of the established theory
of \emph{minimal DG-modules} in order to obtain models for parametrized spectra in the rational approximation:
\begin{defn}[Minimal DG-modules (see {\cite{Roig94}\cite[Sec. 1]{RS}})]
  \label{MinimalDGModule}
  Let $A$ be a DG-algebra.
  Write
  $$
    A[n] \simeq A \otimes \langle c_n \rangle \in A\mbox{-}\mathrm{Mod}
  $$
  for the DG-module over $A$ that is freely generated by a single
  generator $c_n$ in degree $n\in \mathbb{Z}$. The underlying graded vector space is simply $A$ shifted up (or down, if $n$ is negative) in degree by $|n|$ and $c_n$ is identified with shifted algebra unit.
  The differential is same as $A$, shifted in degree, and the module structure is given by left multiplication with elements in $A$.

\item {\bf (i)}  Given $N \in A\mbox{-}\mathrm{Mod}$ and an element $\alpha \in N$ of degree $n+1$ such that $d \alpha = 0$, we construct a new DG-module
  $$
    N \oplus_{\alpha} \left( A\otimes \langle c_n\rangle \right)
    \;\in\;
    A\mbox{-}\mathrm{Mod}
  $$
  whose underlying graded vector space is the direct sum $N \oplus (A\otimes \langle c_n\rangle)$, equipped with the
  evident $A$-module structure.
  The differential is induced from the differentials of $N$ and $A$, with the additional condition that
   $ d c_n := \alpha.
  $
  Hence, the differential is specified by
  $$
    d(n + a \otimes c_n)
    \;=\;
    d_N n + (d_A a) \otimes c_n + (-1)^{\mathrm{deg}(a)} a \cdot \alpha
  $$
  for all $n \in N$ and $a \in A$.
  An $A$-module of the form $N \oplus_{\alpha} (A\otimes \langle c_n\rangle)$ is called an \emph{$n$-cell attachment} of $N$.

\item {\bf (ii)}  A \emph{relative cell complex over $A$ } is an inclusion
  $
    N \hookrightarrow \widehat N
  $
  of $A$-modules such that $\widehat{N}$ arises from $N$ via a countable sequence of  cell attachments as in {\bf (i)}.
 If $N = 0$ is the zero module, then such a $\widehat N$ is called a \emph{cell complex over $A$}.
 A cell complex  over $A$ is therefore equivalent to the data of
 \begin{itemize}
   \item a graded vector space $V$ equipped with a countable, ordered linear basis $\{e^{(i)}\}_{i\in \mathbb{N}}$; and

   \item a differential on $A\otimes V$ making this graded vector space a DG-$A$-module and such that
   \[
   de^{(i)} \in A\otimes \langle e^{(j)}\rangle_{j\leq i}\;.
   \]
 \end{itemize}

\item {\bf (iii)} Furthermore, a cell complex over $A$ is \emph{minimal} the ordered basis $\{e^{(i)}\}_{i\in \mathbb{N}}$
satisfies the additional condition that $i\leq j$ if and only if $\mathrm{deg}(e^{(i)})\leq \mathrm{deg}(e^{(j)})$.
\end{defn}

\begin{remark}[Existence of minimal models  {\cite{Roig94}}]
   \label{MinimalModelsExist}
 For a (necessarily connective) DG-algebra $A$, any $A$-module $M$ admits a \emph{minimal model}.
 This means that there is a minimal $A$-module $N$ together with a quasi-isomorphism of $A$-modules $N\to M$.
Any two minimal models of $M$ are unique up isomorphism, and this isomorphism is unique up to homotopy.
\end{remark}

\begin{example}[Minimal cochain complexes have vanishing differential]
  \label{MinimalModelsForCochainComplexes}
  If $A = \mathbb{Q}$, so that $A\mbox{-}\mathrm{Mod} \simeq \mathrm{Ch}(\mathbb{Q})$ is the category of
  rational cochain complexes,
  then the minimal DG-modules (Def. \ref{MinimalDGModule}) are precisely the cochain complexes of finite type with \emph{vanishing} differential.
  This is because a non-vanishing differential would have to take a generator $e^{(i)}$ in some degree $n$
  to a $\mathbb{Q}$-linear multiple of a generator $e^{(< i)}$ of degree $n+1$, but this is ruled out by the degree condition on the generators.
\end{example}

According to Theorem \ref{RationalParameterizedSpectradgModel}, minimal DG-modules over DG-algebra $A$ determine rational $\mathcal{S}(A)$-spectra.
Just as for minimal DG-algebras, the fiberwise rational stable homotopy groups of a parametrized spectrum can be read off directly from a minimal module:

\begin{lemma}
  \label{RationalizedSpectraModeledByRationalizedHomotopyGroups}
  For $E\in \mathrm{Ho}\left(\mathrm{Spectra}\right)_{\mathbb{Q},\mathrm{ft}}$ any spectrum of finite rational type, the minimal DG-module model for its rationalization (Rem. \ref{MinimalModelsExist}) is the cochain complex
  \[
  (\pi_\bullet(E)\otimes \mathbb{Q})^\ast
  \]
  equipped with vanishing differential.
\end{lemma}

\begin{lemma}[Minimal models for fiber spectra (see {\cite[Rem. 2.7.35]{Bra18}; also \cite{Bra19b}})]
  \label{MinimalDgModulesForFiberSpectra}
  Let $X \in \mathrm{Ho}\left(\mathrm{Spaces}\right)$ be simply-connected of finite rational type,
  and let $E \in \mathrm{Ho}\left(\mathrm{Spectra}_X\right)_{\mathbb{Q},\mathrm{ft}, \mathrm{bbl}}$ be a bounded-below $X$-spectrum of finite rational type (Def. \ref{ParamerizedStableHomotopyTheory}).
  If, moreover,
  $$
    A \simeq \mathcal{O}(X) \;\;\; \in \mathrm{DGCAlg}
  $$
  is any DG-algebra model for $X$ under Prop. \ref{SullivanEquivalence}, and if
  $$
    A \otimes V \;\simeq\; \mathcal{M}(E) \;\;\; \in A\mbox{-}\mathrm{Mod}
  $$
  is a \emph{minimal} DG-module model (Def. \ref{MinimalDGModule}) for $E$ under Theorem \ref{RationalParameterizedSpectradgModel},
  then for every $x\in X$ the cochain complex with vanishing differential
  $$
    V \;\simeq\; \mathcal{M}_{\mathbb{Q}}(E_x) \;\;\; \in \mathrm{Ch}(\mathbb{Q})
  $$
  is a minimal model for the fiber spectrum $E_x$.
  \end{lemma}

\begin{example}[Minimal model for suspension spectrum of $B S^1$]
  \label{MinimalModelForSuspensionSpectrumOfCircleClassifyingSpace}
  The minimal model for the suspension spectrum (Prop. \ref{AdjunctionStabilization}) of the classifying space $B S^1$ (viewed as a parametrized spectrum over the point)
  is the graded vector space spanned by one generator in every non-negative even degree:
  \begin{equation}
  \label{BS1MinimalModule}
    \mathcal{M}_{\mathbb{Q}}\left(\Sigma^\infty_+ B S^1 \right)
    \;\simeq\;
    \mathbb{Q}[\beta_2] = \langle 1, \beta_2, \beta_2^2 ,\dotsc\rangle
    \;\;\;
    \in \mathrm{Ch}(\mathbb{Q})
    \,.
  \end{equation}
Indeed, a minimal DG-algebra model for $B S^1$ is the symmetric graded algebra
  on a single generator in degree 2:
  $$
    \mathbb{Q}[ \beta_2 ]
    \;\simeq\;
    \mathcal{O}(B S^1)
    \;\;\;\in \mathrm{DGCAlg}
    \,,
  $$
  which necessarily has trivial differential;
  this implies the claim with Prop. \ref{FiberwiseSuspensionSpectrumdgModel}.
  Note that in \eqref{BS1MinimalModule} the generator in degree $2n$ is $\beta^n_2$
 ( however, this notation should be taken with a grain of salt since we have forgotten the algebra structure at this point).
\end{example}

\begin{example}[Rational Snaith theorem]
  \label{RationalSnaithTheorem}
  Let $\mathrm{KU} \in \mathrm{Ho}\left(\mathrm{Spectra}\right)$ be the spectrum representing complex K-theory
  and write $\mathrm{ku} \in \mathrm{Spectra}$ for its connective cover (obtained by killing negative-dimensional homotopy groups).
  Minimal DG-module models for  $\mathrm{KU}$ and $\mathrm{ku}$ are, by Lemma \ref{RationalizedSpectraModeledByRationalizedHomotopyGroups},
  given by
  $$
    \mathcal{M}_{\mathbb{Q}}\left( \mathrm{ku} \right)
    \;\simeq\;
    \mathbb{Q}[\beta_2]
    \; \in
    \mathrm{Ch}(\mathbb{Q})
 \qquad
\text{ and}
\qquad
     \mathcal{M}_{\mathbb{Q}}\left( \mathrm{KU} \right)
    \;\simeq\;
    \mathbb{Q}[\beta_2, \beta_2^{-1}]
    \; \in
    \mathrm{Ch}(\mathbb{Q})\;.
  $$
  In particular, rationally there is no difference between $ku$ and $\Sigma^\infty_+ BS^1$:
  $$
    \Sigma^\infty_+ B S^1 \;\simeq_{\mathbb{Q}}\; \mathrm{ku}.
  $$
  On the other hand, if we remember the algebra structure on $\mathbb{Q}[\beta_2]$, the full non-connective
  K-theory is obtained by multiplicatively inverting the element $\beta_2$
  $$
    \Sigma^\infty_+ B S^1[\beta_2^{-1}] \;\simeq_{\mathbb{Q}}\; \mathrm{KU}
    \,.
  $$
\end{example}

\begin{remark}[Full Snaith theorem]
  As a matter of fact, the last statement in the previous example is still true \emph{non-}rationally:
  this is the content of Snaith's theorem \cite{Snaith81}. After rationalization, Snaith's theorem essentially reduces to a
  triviality, however keeping in mind that this is a \lq\lq rational shadow'' may help identify the non-rational situation
  approximated by our main Theorem \ref{TwistedKTheoryInsideFiberwiseStabilizationOfATypeOrbispaceOf4Sphere}
  below. We conjecture that Theorem \ref{TwistedKTheoryInsideFiberwiseStabilizationOfATypeOrbispaceOf4Sphere}
  remains true non-rationally by a generalization of Snaith's theorem to twisted K-theory obtained by \emph{fiberwise}
  inversion of the Bott generator.   We will return to this point elsewhere.
\end{remark}

We have seen in Prop. \ref{FiberwiseSuspensionSpectrumdgModel} how stabilization -- the process of passing from spaces
to spectra -- works in terms of rational models by taking augmentation ideals.
Conversely, the \emph{de}stabilization process that extracts an infinite loop space from a spectrum also has a straightforward
incarnation in terms of algebraic models.
For connective parametrized spectra, extracting fiberwise infinite loop spaces is represented by taking the free algebra of the
corresponding DG-module:
\begin{prop}[Rational models for fiberwise infinite loop spaces (see {\cite[Sec. 2.7]{Bra18}; \cite{Bra19b}})]
   \label{RationalInfiniteLoopSpace}
Let $A$ be a DG-algebra of finite type such that $\mathcal{S}(A)$ is simply-connected.
If $M$ is a connective $A$-module, then under the inverse equivalence of
$$
         \xymatrix{
           \mathrm{Ho}\big( \mathrm{Spectra}_{\mathcal{S}(A)} \big)_{\mathbb{Q},\mathrm{ft}, \mathrm{bbl}}
           \ar[rr]^-{\mathcal{M}_A}_-\simeq
           &&
           \mathrm{Ho}\big(
             A\mbox{-}\mathrm{Mod}
           \big)^{\mathrm{op}}_{\mathrm{ft}, \mathrm{bbl}}
         },
$$
the fiberwise infinite loop space (Prop. \ref{AdjunctionStabilization}) is modelled by the augmented
$A$-algebra $\mathrm{Sym}_A (N)$ where $N$ is a minimal model of $M$ (Rem. \ref{MinimalModelsExist}).
\end{prop}

\begin{lemma}[Minimal model for twisted connective K-theory]
  \label{TwistedKModel}
  Denote the parametrized spectrum representing
  general twisted connective K-theory (e.g. \cite{AndoBlumbergGepner10}) by
  $$
    \xymatrix{
      \mathrm{ku} \dslash \mathrm{GL}_1(\mathrm{ku})
      \ar[r]&
      B \mathrm{GL}_1( \mathrm{ku} )
    }
  $$
  and its restriction to the twist by ordinary degree-3 cohomology
  by \footnote{Here and elsewhere, ``(pb)" denotes a (homotopy-)pullback square.}
  $$
    \xymatrix{
      \mathrm{ku} \dslash BS^1
        \ar[r]
        \ar@{}[dr]|{\mbox{\emph{\footnotesize{(pb)}}}}
        \ar[d]
       &
      \mathrm{ku} \dslash \mathrm{GL}_1(\mathrm{ku})
      \ar[d]
      \\
      K(\mathbb{Z},3)
        \ar[r]
        &
      B \mathrm{GL}_1(\mathrm{ku}).
    }
  $$
  The minimal DG-algebra model for the base space is
  the graded symmetric algebra freely generated by a single generator in degree $h_3$ with vanishing differential:
  $$
    \mathcal{O}\big(K(\mathbb{Z},3)\big)
    \simeq
    \mathbb{Q}[h_3]
    \,.
  $$
  The corresponding minimal DG-module model (Def. \ref{MinimalDGModule}) for
  the rationalization of $\mathrm{ku}\dslash BS^1$
  is
  \begin{equation}
    \label{MinimaldgModelForTwistedK}
    \mathcal{M}_{\mathbb{Q}[h_3]}\left(
      \mathrm{ku} \dslash BS^1
    \right)
    \;\simeq\;
    \mathbb{Q}[h_3] \otimes
    \left\langle
      \omega_{2k} \,\vert\, k \in \mathbb{N}
    \right\rangle
    \Big/
    \left(\!\!
          \begin{array}{c}
       \hspace{-6mm} d \omega_0  = 0
        \\
        d \omega_{2k + 2}  = h_3 \wedge \omega_{2k}
      \end{array}
      \!\!
    \right).
  \end{equation}
  The module structure over $\mathbb{Q}[h_3]$ is given by the evident action on this tensor factor.

 Moreover, for each $n\geq 0$  the fiberwise infinite loop space of the fiberwise suspension
 $\Sigma_{B^3\mathbb{Q}}^{2n} (\mathrm{ku}\dslash BS^1)$ has minimal DG-model
  \begin{equation}
    \label{OmegaInfinityOfTwistedk}
    \mathcal{O}\big(
      \Omega^{\infty-2n}_{B^3 \mathbb{Q}} \left( \mathrm{ku} \dslash BS^1 \right)
    \big)
    \;\simeq\;
    \mathbb{Q}[h_3, \omega_{2k+2n} \,\vert\, k \in \mathbb{N}]
    \Big/
    \left(\!\!
          \begin{array}{c}
       \hspace{-6mm} d \omega_{2n}  = 0
        \\
        d \omega_{2k + 2n+2}  = h_3 \wedge \omega_{2k + 2n}
      \end{array}
      \!\!
    \right).
  \end{equation}
\end{lemma}

\begin{proof}
For the structure as a $\mathbb{Q}[h_3]$-module, we appeal to Lemma \ref{MinimalDgModulesForFiberSpectra} and Ex.
 \ref{RationalSnaithTheorem}.
  It only remains to determine the differential.
  By \cite[Prop. 3.9]{AtiyahSegal05} the degree-3 twist on complex K-theory is non-trivial in every degree.
  But by the degrees of the generators in \eqref{MinimaldgModelForTwistedK}, the given differential is degreewise
  and up to isomorphism the only possible non-trivial differential.
  The minimal models \eqref{OmegaInfinityOfTwistedk} are obtained by using Rem. \ref{SuspensionandShiftingParam}  and Prop. \ref{RationalInfiniteLoopSpace}.
\end{proof}

A form  of this statement also appears as \cite[Ex. 12.5]{BunkeNikolaus14}, \cite[Sec. 3.1]{GS17}.




\subsection{The $\mathrm{Ext}$/$\mathrm{Cyc}$-adjunction}
\label{TheAdjunction}

Our mathematical formalization of double dimensional reduction and gauge enhancement in Sec.
\ref{TheMechanism} involves at its core particular universal construction -- the
\emph{ $\mathrm{Ext}$/$\mathrm{Cyc}$-adjunction}.
While we ultimately work with the rational homotopy theory version of this construction, we would
like to amplify that this  adjunction can be formulated much more generally.
In particular, for $\mathbf{H}$ any \lq\lq good'' homotopy theory
(for instance an $\infty$-topos \cite{Lurie06}) and for $G$ a strong homotopy group in $\mathbf{H}$
(namely, a grouplike $A_\infty$-monoid) with delooping $\mathbf{B}G$, there is
a duality (an $\infty$-adjunction)
  $$
    \xymatrix{
      \mathbf{H}
      \ar@{<-}@<+6pt>[rr]^-{ \mathrm{Ext}_G }
      \ar@<-6pt>[rr]_-{ \mathrm{Cyc}_G }^-{\bot}
      &&
      \mathbf{H}_{/\mathbf{B}G}
    }
  $$

  \vspace{-3mm}
\noindent between

\vspace{-2mm}
\begin{enumerate}[{\bf (i)}]
\item  the operation $\mathrm{Ext}_G$ of forming $G$-extensions; and

\item   the operation of \emph{$G$-cyclification}; the result of first forming the space of maps out of $G$ and then taking the homotopy
quotient by the $G$-action rigidly reparametrizing these maps.
\end{enumerate}
In terms of abstract homotopy theory, this adjunction turns out to be right base change along the essentially unique point inclusion
map $\ast \to \mathbf{B}G$. For the reader familiar with abstract homotopy theory this fully defines the adjunction, and the only point
to check is that this right adjoint is indeed obtained by forming cyclifications as claimed.

\medskip
For any object $X\in \mathbf{H}$, specifying a $G$-principal bundle on $X$ is equivalent to the data of a map
\[
\tau\colon X\longrightarrow \mathbf{B}G.
\]
The $G$-principal bundle associated to $\tau$ is obtained by computing the homotopy fiber at the essentially
unique point in $\mathbf{B}G$:
\[
\xymatrix@C=5em{
  \mathrm{Ext}_G (\tau)\ar[r]\ar[d]
  \ar@{}[dr]|{\mbox{\footnotesize{(pb)}}}
 & \ast\ar[d]
  \\
  X\ar[r]^-{\tau} & \mathbf{B}G,
  }
\]
(see \cite{NSS12} for an exposition of the general theory).
Importantly for our purposes, for each map $\tau$ the $\mathrm{Ext}/\mathrm{Cyc}$-adjunction provides us
with a natural morphism -- the unit of the adjunction -- which fits into a (homotopy) commutative diagram:
\[
\xymatrix@R=1em{
  X\ar[rr]\ar[dr]_-{\tau}
  && \mathrm{Cyc}_G\mathrm{Ext}_G (\tau). \ar[dl]
  \\
  &\mathbf{B}G&
  }
\]
The map $X\to \mathrm{Cyc}_G\mathrm{Ext}_G (\tau)$ is the operation that takes any point in $X$ to the map
$G\to \mathrm{Ext}_G(\tau)$ which winds identically around the extension fiber over that point.
This is only well-defined up to a choice of base point in the fiber, but this is precisely the ambiguity that this
quotiented out in the definition of $\mathrm{Cyc}_G$.
\begin{remark}[Notation]
We will often abuse notation and write $\mathrm{Ext}_G(X)$ instead of $\mathrm{Ext}_G(\tau)$ when it is
understood that we are considering a particular map $\tau\colon X\to \mathbf{B}G$.
\end{remark}

We now describe the $\mathrm{Ext}/\mathrm{Cyc}$-adjunction in some detail in the setting of classical homotopy theory (Def. \ref{ClassicalHomotopyCategories}).
More precisely, for any strict topological group $G$ we exhibit an ordinary adjunction (Theorem \ref{GCycExtAdjunction} below)
between categories of topological spaces which for $G=S^1$ presents  $\mathrm{Ext}/\mathrm{Cyc}$-adjunction in homotopy (see Rem. \ref{ExtCycAdjunctiononHomotopyCats}).
Let us first recall some  preliminaries to establish our conventions:
\begin{defn}[Group action on topological space]
  \label{GSpace}
  For $X$ a topological space and $G$ a topological group, a \emph{(right) action} of $G$ on $X$
  is a continuous function
  \begin{align*}
  X\times G &\longrightarrow X\\
  (x,g)&\longmapsto x\cdot g
  \end{align*}
  such that
  $
  x\cdot e = x
  $
  and
$(x\cdot g_1)\cdot g_2) = x\cdot (g_1\cdot g_2)
 $  for all $x\in X$ and $g_1, g_2\in G$.
  One also refers to this situation by saying that $X$ is a \emph{$G$-space}.
\end{defn}

\begin{defn}[Quotients by group actions]
\label{Quotients}
For 
a $G$-space $X$
, we write
$$
  X/G
  \;:=\;
  X/( x \sim x \cdot g )
$$
for the (ordinary) quotient space, which comes with the quotient projection
\begin{equation}
  \label{QuotientProjection}
  \xymatrix@R=1.5em{
    X
    \ar[r]^-{\pi_{G}}
    &
    X/G
    \,.
  }
\end{equation}
\item For $G$-spaces $X$ and $Y$, we write
\begin{equation}
  \label{QuotientByDiagonalAction}
  X \times_G Y := (X \times Y)/G := \big(X \times Y\big)/\big( (x,y) \sim (x \cdot g , y \cdot g) \big)
\end{equation}
for the quotient by the diagonal action.
\end{defn}

\begin{remark}[Comparison map for free actions]
  \label{ComparisonMapForFreeActions}
Recall that the $G$-action on $X$ is called \emph{free} if for every pair of points $(x_1,x_2) \in X \times X$ there is at most one $g \in G$ such that $x_2 = x_1 \cdot g$.
For a free action there is a well-defined \emph{comparison map} which we suggestively write as
\begin{align}
\label{ComparisonMap}
X\times_{X/G} X&\longrightarrow G\\
\notag
[x_1, x_2] &\longmapsto x_1^{-1} \cdot x_2
\end{align}
such that
$
  y_1 \cdot (x_1^{-1} \cdot x_2) = y_2
$
whenever $[y_1, y_2 ] = [x_1, x_2]$.
The comparison map determines a homeomorphism $X\times_{X/G}X\to X\times G$ via $[x_1, x_2]\mapsto (x_1, x_1^{-1}\cdot x_2)$.
\end{remark}

\begin{defn}[$G$-Extension functor]
\label{GExt}
For any topological group $G$, there exists a topological space $EG$
such that $EG$ is a free $G$-space
which is weakly contractible: the map $EG\to \ast$ is a weak homotopy equivalence.
The quotient space
\begin{equation}
  \label{ClassifyingSpace}
  B G := (E G)/G
\end{equation}
is the \emph{classifying space} of $G$, and the quotient projection
\begin{equation}
  \label{UniversalGPrincipal}
  \xymatrix@R=1.2em{
    E G
    \ar[r]^-{\pi_G}
    &
    B G
  }
\end{equation}
this is called the \emph{universal $G$-principal bundle}.
The $G$-bundle $EG\to BG$ is determined by this specification uniquely up to homotopy equivalence. 

Given any topological space $X$ equipped with a map $X \overset{\phi}{\to} B G$, we can pull back the universal bundle \eqref{UniversalGPrincipal}
to obtain a space $X\times_{BG} E G$ with free $G$-action whose quotient space is $X$:
\begin{equation}
  \label{GBundleByPullback}
  \raisebox{20pt}{
  \xymatrix@C=5em{
    X \underset{B G}{\times} E G \ar@{}[dr]|{\mbox{\footnotesize{(pb)}}}\ar[d]_-{\pi_G} \ar[r] & E G \ar[d]^-{\pi_G}
    \\
    X \ar[r]^-{\phi} & B G
    \,.
  }
  }
\end{equation}
This is the $G$-principal bundle \emph{classified by $\phi$}.
For our purposes, it is useful to think of $X\times_{BG} EG$
as the \emph{extension} that is classified by the \emph{cocycle} $\phi$.
Therefore, we write $\mathrm{Ext}_G$ for the functor that computes these fiber products:
\begin{align}
  \label{GExtF}
  \mathrm{Ext}_G\colon \mathrm{Spaces}_{/BG}
  &
  \longrightarrow
  \mathrm{Spaces}\\
  \notag
  (Y\to BG)
  &\longmapsto
  Y\underset{BG}{\times} EG.
\end{align}
\end{defn}

\begin{defn}[Homotopy quotient]
\label{HomotopyQuotient}
Given a $G$-space $X$ (Def. \ref{GSpace}), write $X \dslash G$ for the
\emph{homotopy} quotient space.
This is specified up to weak homotopy equivalence by the \emph{Borel construction}
\begin{equation}
  \label{BorelConstruction}
  X\dslash G := X \times_G E G
  \,,
\end{equation}
which we take as our definition.
\end{defn}

\begin{example}[Homotopy quotient of trivial $G$-actions]
  \label{HomotopyQuotientOfTrivialGAction}
  The homotopy quotient  
   of the (unique) $G$-action
  on the point $\ast$
  is the classifing space \eqref{ClassifyingSpace}:
  $$
    \ast \dslash G
    :=
    (\ast \times E G)/G
    =
    (E G)/G
    =
    B G
    \,.
  $$
  More generally, for a \emph{trivial} $G$-space $X$  (so that $x\cdot g = x$ for all $g$),
  the homotopy quotient is simply
  $$
    X \dslash G
    \;\simeq\;
    X \times B G
    \,.
  $$
\end{example}

\begin{remark}[Maps related to the homotopy quotient]
  \label{MapsRelatedToHomotopyQuotient}
The homotopy quotient (Def. \ref{HomotopyQuotient}) is naturally equipped with the following maps of interest:

\item {\bf (i)}
The ordinary quotient projection \eqref{QuotientProjection} factors canonically up to homotopy via the homotopy quotient as
$$
  \pi_G
  \;:\;
  \xymatrix{
    X \ar[r]
    &
    X \dslash G
    \ar[r]
    &
    X / G
  }\!.
$$
Indeed, choosing any point $p \in E G$ (which is unique up to homotopy since $EG$ is contractible), the factorization is obtained
by the sequence of maps on the left-hand side of the diagram
\begin{equation}
  \label{HomotopyQuotientReceiving}
  \raisebox{40pt}{
  \xymatrix{
    X \ar[rr]^-{ x \mapsto (x,p) } \ar@/_1.6pc/[dd]_{\pi_G} \ar[d]^-{ x \mapsto [x,p] } && X \times E G \ar[d]
    \\
    X \dslash G  \ar@{=}[rr] \ar[d]  && X \times_G E G \ar[d]
    \\
    X/G \ar@{=}[rr] && X \times_G \ast
  }
  }
\end{equation}
On the right-hand side we first take the quotient projection by the diagonal $G$-action on $X\times EG$ and then project out the
 $EG$ factor via the ($G$-equivariant) map $EG\to \ast$.
If the $G$-action on $X$ is free, the comparison map $X\dslash G\to X/G$ is a weak homotopy equivalence
(see e.g. \cite{Kor}). 

\item {\bf (ii)} The homotopy quotient $X\dslash G$ is equipped with a canonical map to the classifying space \eqref{ClassifyingSpace}
\begin{equation}
  \label{CanonicalCocycleOnHomotopyQuotient}
   X\dslash G \longrightarrow B G,
\end{equation}
obtained via the map $X\to \ast$ as
\begin{equation}
  \label{CanonicalCocycle}
  X\dslash G
  =
  X \times_{G} E G \xymatrix{\ar[r]&} \ast \times_G\,  E G = (E G)/G = B G
  \,.
\end{equation}
\end{remark}

\begin{prop}[Extension of homotopy quotient]
  \label{ExtensionOfHomotopyQuotientEquivalentToOriginalSpace}
Any $G$-space $X$ is weakly homotopy equivalent to the $G$-extension (Def. \ref{GExt}) of its homotopy quotient $X\dslash G$ (Def. \ref{HomotopyQuotient}) along the canonical map \eqref{CanonicalCocycleOnHomotopyQuotient}:
  \begin{equation}
    \label{ExtOfHomotopyQuotient}
    \xymatrix{
      \mathrm{Ext}_G( X \dslash G )
      \ar[rr]^-{}^-{ \simeq_{\mathrm{whe}} }
      &&
      X.
    }
  \end{equation}
\end{prop}

\begin{proof}
Unwinding the definitions, the extension in question is obtained as the pullback
\begin{equation}
  \label{PullbackDescriptionOfExtOfHomotopyQuotient}
  \mathrm{Ext}_G( X \dslash G ) = (X \times_G E G) \underset{B G}{\times} E G,
\end{equation}
which is homeomorphic \emph{as a $G$-space} to  $X \times E G$, via the map
\begin{align}
 \label{ExtOfHomotopyQuotientForm}
(X \times_G E G) \underset{B G}{\times} E G
&\longrightarrow X\times EG
\\
\notag
\big([x,e_1], e_2\big) &\longmapsto
\big(x\cdot (e_1^{-1}\cdot e_2), e_2 \big)
%
\end{align}
where we have used the comparison map \eqref{ComparisonMap} fiberwise over $BG$.
Thus we have a map
\[
\mathrm{Ext}_G(X\dslash G)\cong X\times EG \longrightarrow X
\]
which is a weak homotopy equivalence by weak contractibility of $EG$.
\end{proof}

\begin{remark}[$\mathrm{Ext}_G(X\dslash G)$ as a free resolution]
The induced action  on $X\dslash G\times_{BG} E G$
in \eqref{PullbackDescriptionOfExtOfHomotopyQuotient}
is always free, even if the action on $X$ is not.
Additionally, the isomorphism \eqref{ExtOfHomotopyQuotientForm} is manifestly $G$-equivariant for the diagonal $G$-action on $X \times E G$.
Hence Prop. \ref{ExtensionOfHomotopyQuotientEquivalentToOriginalSpace} is saying that $\mathrm{Ext}_G(X\dslash G)$ is a \emph{resolution} of $X$ by a free $G$-space.
For example, if $X = \ast$ then we have $\ast\dslash G \times_{BG} E G \cong E G$.
\end{remark}

We now turn to the description of the $G$-cyclification functor, which extends the cyclification functor from \cite{FSS17, Higher-T}.

\begin{defn}[Mapping space out of $G$]
\label{MappingSpaceOutOfG}
For a topological space $Y$, write $\mathrm{Maps}(G,Y)$ for the space of continuous maps $G\to Y$.
\footnote{we will always assume that all topological spaces are compactly generated, so that $\mathrm{Maps}(G,Y)$
is the exponential object in the category of compactly generated spaces---this completely specifies the topology.}
This mapping space is regarded as equipped with the
$G$-action
\begin{align*}
\mathrm{Maps}(G,Y)\times G&\longrightarrow \mathrm{Maps}(G,Y)\\
(f,g)&\longmapsto \big[(f\cdot g)\colon h\mapsto f(hg^{-1})\big],
\end{align*}
which, equivalently, is the conjugation action on maps of $G$-spaces where $Y$ is regarded as having the trivial $G$-action.
\end{defn}

\begin{defn}[$G$-Cyclification]
\label{GCyc}
For $G$ a topological group and $Y$ a topological space, the \emph{$G$-cyclification} of $Y$ is the map
\[
\mathrm{Maps}(G,Y)\dslash G\longrightarrow \ast \dslash G =BG
\]
obtained by forming the homotopy quotient (Def. \ref{HomotopyQuotient})
of the mapping space $\mathrm{Maps}(G,Y)$ (Def. \ref{MappingSpaceOutOfG}).
\item {\bf (i)} This assignment extends to a functor
\begin{align*}
\mathrm{Cyc}_G\colon \mathrm{Spaces}&\longrightarrow \mathrm{Space}_{/BG}
\\
Y
&\longmapsto \mathrm{Maps}(G,Y)\times_G EG.
\end{align*}
\item {\bf (ii)} For the special case that $G = S^1$ is the circle group we omit the subscript \lq\lq $G$'' and write simply
\begin{equation}
  \label{Cyclic}
  \mathrm{Cyc}(X) = \mathcal{L}(X) :=  \mathrm{Maps}(S^1, X) \dslash S^1.
\end{equation}
This is the homotopy quotient of the \emph{free loop space} of $X$ by the rigid rotation action on loops.
The cohomology of $\mathrm{Cyc}(X)$ is the \emph{cyclic cohomology} of $X$, whence the terminology and notation.
\end{defn}

The key fact relating the functors $\mathrm{Ext}_G$ and $\mathrm{Cyc}_G$ in this setting is the following:
\begin{theorem}[The $\mathrm{Ext}$/$\mathrm{Cyc}$-adjunction]
  \label{GCycExtAdjunction}
Let $G$ be a topological group. Then
  \item {\bf (i)} The functors $\mathrm{Ext}_G$ (Def. \ref{GExt}) and $\mathrm{Cyc}_G$ (Def. \ref{GCyc}) are adjoints,
  with $\mathrm{Ext}_G$ the left and $\mathrm{Cyc}_G$ the right adjoint:

 \vspace{-3mm}
  $$
    \xymatrix{
      \mathrm{Spaces}
      \ar@{<-}@<+6pt>[rr]^-{ \mathrm{Ext}_G }
      \ar@<-6pt>[rr]_-{ \mathrm{Cyc}_G }^-{\bot}
      &&
      \mathrm{Spaces}_{/B G}\;.
    }
  $$

  \vspace{-3mm}
  \item {\bf (ii)}
  The unit of the adjunction
  $$
    \xymatrix{
      X
      \ar[r]^-{\eta_{{}_X}   }
      &
      \mathrm{Cyc}_G(\mathrm{Ext}_G(X))
    }
  $$
 is the map that sends $x\in X$ to the equivalence class of the map $G\to \mathrm{Ext}_G(X)|_x$ obtained
 by choosing any image of the neutral element $e\in G$ and extending to all of $G$ by the group action:
  $$
    \eta_{{}_X}  \;:\; x \longmapsto \big[ G \overset{\simeq}{\longrightarrow} \mathrm{Ext}_G(X)\vert_x \big].
  $$
  \item {\bf (iii)}
  The counit of the adjunction
  $$
    \xymatrix@R=1.5em{
      \mathrm{Ext}_G(\mathrm{Cyc}_G(Y))
      \ar[dr]_-{ \simeq_{\mathrm{whe}} }
      \ar[rr]^-{ \epsilon_{{}_Y} }
      && Y
      \\
      & \mathrm{Maps}(G,Y) \ar[ur]_-{\mathrm{ev}_e}
    }
  $$
  is the composite of a weak homotopy equivalence to $\mathrm{Maps}(G,Y)$ followed by
  evaluation at the neutral element.
\end{theorem}
\begin{proof}
To show {\bf (i)},
given $(c\colon X \to B G) \in \mathrm{Spaces}_{BG}$ and $Y \in \mathrm{Spaces}$,
  we must produce a natural bijection of sets
  \begin{equation}
  \label{eqn:CycExtCondition}
  \mathrm{Hom}
      \Big(
        X \underset{B G}{\times} E G
        \,,\,
        Y
      \Big)
      \cong
      \mathrm{Hom}_{/BG} \big(X, \mathrm{Maps}(G,Y)\times_G EG\big).
  \end{equation}
On the one hand, given a map
\[
\xymatrix@R=1em{
  X \ar[rr]^-{\phi} \ar[dr]_-{c}&& \mathrm{Maps}(G,Y)\times_G EG\ar[dl]
  \\
  &BG&
  }
\]
that sends $x\mapsto [f_x, p_x]$, we define the map
$\Phi(\phi)\colon X\underset{BG}{\times} EG  \longrightarrow Y$
via
$
(x,p) \longmapsto f_x (p^{-1}\cdot p_x)$.
It is easy to check that this assignment does not depend on the choice of representative of the class $[f_x, p_x]$, and the
argument of $f_x$ is determined by the comparison map \eqref{ComparisonMap}, since $p, p_x$ lie in the same fiber of $EG$ over $BG$.

  Conversely, given a map
  \begin{align*}
  \Psi\colon X\underset{BG}{\times} EG & \longrightarrow Y\\
  (x,p)&\longmapsto y_{x,p},
  \end{align*}
we define the map $\psi(\Psi)\colon X \to \mathrm{Maps}(G,Y)\times_G EG$ via the assignment
$
x\mapsto [y_{x,p_x\cdot (-)^{-1}}, p_x]
$,
where $p_x \in EG$ is such that $(x,p_x)\in X\times_{BG} EG$ and $y_{x,p_x\cdot (-)^{-1}}$ denotes the map $G\to Y$ given by the assignment
$g \mapsto y_{x,p_x\cdot g^{-1}}$.
It is straightforward to see that $\psi(\Psi)$ is well-defined and indeed determines a map of spaces over $BG$.

We now check that that the assignments $\phi\mapsto \Phi(\phi)$ and $\Psi\mapsto \psi(\Psi)$ are inverses of each other.
If $\phi\colon x\mapsto [f_x, p_x]$ is as above, then we have
\[
\psi(\Phi(\phi))\colon x\longmapsto \Big[f_x\big((-)\cdot (p'_x)^{-1} \cdot p_x\big), p'_x\Big] = [f_x, p_x],
\]
where $p'_x\in EG$ is any point such that $(x,p_x)\in X\times_{BG} \, EG$.
Similarly, $\Phi(\psi(\Psi))= \Psi$ for all maps of spaces $\Psi \colon X\to \mathrm{Cyc}_G(Y)$ over $BG$.
Indeed, writing $\Psi\colon (x,p)\mapsto y_{x,p}$ as above then
\[
\Phi(\psi(\Psi))\colon (x,p) \longmapsto y_{x, p_x\cdot (p_x^{-1}\cdot p)} = y_{x,p}.
\]
This gives us a bijection on hom-sets of the desired form \eqref{eqn:CycExtCondition}, which is manifestly natural in $(c\colon X\to BG)$ and $Y$.
This completes the proof of {\bf (i)}.

As to {\bf (ii)}, we recall that the component of the unit at $(X\to BG)$ is the adjunct of the identity map on $\mathrm{Ext}_G(X)$.
By the above, this map is described by the assignment
  $$
    x \longmapsto [ (x,p_x \cdot (-)^{-1}), p_x ]\;,
  $$
  where $p_x\in EG$ is any point such that $(x,p_x)\in X\times_{BG} EG$.
But this is precisely of the claimed form: for each $x$ we choose a point $(x,p_x)$ in the fiber of $\mathrm{Ext}_G(X)$ over $X$ which is the image of the neutral element of $G$.
Then an arbitrary element $g\in G$ is sent to $(x, p_x\cdot g^{-1})$, which  determines a homeomorphism $G\cong \mathrm{Ext}_G(X)|_x$.
Passing to the homotopy quotient removes the choice ambiguity.

For {\bf (iii)}, the component of the counit at $Y$ is the adjunct of the identity on $\mathrm{Cyc}_G(Y)$.
By the above, this is simply the map
\begin{align*}
\big(\mathrm{Maps}(G,Y)\times_G EG\big)\underset{BG}{\times} EG
&\longrightarrow Y\\
\big([f,p], p'\big) &\longmapsto f\big((p')^{-1}\cdot p\big).
\end{align*}
In the proof of Prop. \ref{ExtensionOfHomotopyQuotientEquivalentToOriginalSpace}, we saw that
\begin{align*}
\kappa\colon \big(\mathrm{Maps}(G,Y)\times_G EG\big)\underset{BG}{\times} EG
&\longrightarrow \mathrm{Maps}(Y, G)\times EG\\
\big([f,p], p'\big) & \longmapsto \big(f\cdot (p^{-1}\cdot p'), p' \big)
\end{align*}
is a homeomorphism of $G$-spaces, so that the counit factors as
\[
\xymatrix{
  \big(\mathrm{Maps}(G,Y)\times_G EG\big)\underset{BG}{\times} EG
  \ar[rr]^-{\simeq_{\mathrm{whe}}}
  &&
  \mathrm{Maps}(G,Y)
  \ar[rr]^-{\mathrm{ev}_e}
  &&
  Y,
  }
\]
where we have used that $EG\to \ast$ is a weak homotopy equivalence.
This completes the proof.
\end{proof}

\begin{remark}[Extension to the homotopy categories]
\label{ExtCycAdjunctiononHomotopyCats}
The result we have just proven establishes $\mathrm{Ext}/\mathrm{Cyc}$-adjunction as an ordinary adjunction between
categories of topological spaces.  However, as we are interested in the corresponding adjunction between the \emph{homotopy}
categories, some additional points are in order.
\begin{itemize}
  \item Since the universal $G$-principal bundle $EG\to BG$ is always a (Serre) fibration, taking fiber products with this
  map preserves weak homotopy equivalences. In particular, the functor $\mathrm{Ext}_G$ is \emph{homotopical} and so
  descends to functor on derived categories $\mathrm{Ho}(\mathrm{Spaces}_{/BG})\to \mathrm{Ho}(\mathrm{Spaces})$.

  \item The homotopical properties of the cyclification functor are more involved.
  Indeed, $\mathrm{Cyc}_G$ may fail to be a homotopical functor in general since $Y\mapsto \mathrm{Maps}(G,Y)$ need
  not preserve weak homotopy equivalences, though this problem evaporates if $G$ is a CW complex.
  In this article, we are primarily interested in $G= S^1$, in which case $\mathrm{Cyc}=\mathrm{Cyc}_{S^1}$ \emph{is} a
  homotopical functor and so does determine a functor between the corresponding homotopy categories.
\end{itemize}
In summary, we have that Theorem \ref{GCycExtAdjunction} presents the adjunction between homotopy categories
\[
\xymatrix{
      \mathrm{Ho}\big(\mathrm{Spaces}\big)
      \ar@{<-}@<+6pt>[rr]^-{ \mathrm{Ext} }
      \ar@<-6pt>[rr]_-{ \mathrm{Cyc} }^-{\bot}
      &&
      \mathrm{Ho}\big(\mathrm{Spaces}_{/BS^1}}\!\big)
\]
and, therefore (upon further localization), between \emph{rational} homotopy categories, which is our primary focus in this article.
For more general $G$, the adjunction of Theorem \ref{GCycExtAdjunction} may fail to descend to homotopy categories.
There are various ways of remedying this issue, but this takes us beyond the scope of the present article.
\end{remark}

Below we will be mainly concerned with the $\mathrm{Ext}/\mathrm{Cyc}$-adjunction in the \emph{rational} approximation.
To see how this works, we first establish good models for the rationalization of the cyclification functor:

\begin{remark}[Minimal models for cyclic loop spaces {\cite{VigueBurghelea}}]
  \label{MinimalDGCModelForCyclicLoopSpace}
  Let $X$ be a simply-connected topological space of finite rational type (Def. \ref{ClassicalHomotopyCategories}) and
  let $(\mathbb{Q}[ \{\omega_i\} ], d_X)$ be a corresponding minimal DG-algebra (Prop. \ref{SullivanEquivalence} {\bf{(iv)}}). Then

 \item {\bf (i)}
  a minimal DG-algebra model for the free loop space
  $\mathcal{L}(X) := \mathrm{Maps}(S^1,X)$ is obtained by adjoining a second copy $\{s \omega_i\}$ of the generators, where the degrees of these additional \lq\lq looped generators'' satisfy $\mathrm{deg}(s\omega_i) = \mathrm{deg}(\omega_i)-1 = i-1$, and with differential given by
  $$
    d_{{}_{\mathcal{L}(X)}} \omega_i := d_{{}_X} \omega_i\,,
    \qquad \quad
    d_{{}_{\mathcal{L}(X)}} s \omega_i := - s ( d_{{}_{X}} \omega_i )
    \,,
  $$
  where on the right $s$ is uniquely extended from a linear map on generators to a graded derivation of degree $-1$.

 \item {\bf (ii)}
  a minimal DG-algebra model for the cyclic loop space $\mathrm{Cyc}(X)$ (Def. \ref{GCyc}) is obtained from this by adjoining one more generator
  $\widetilde{\omega}_2$ in degree 2, and taking the differential to be given by
  $$
    d_{{}_{\mathrm{Cyc}(X)}} \widetilde{\omega}_2 = 0
    \,,
    \qquad \quad
    d_{{}_{\mathrm{Cyc}(X)}} w = d_{{}_{\mathcal{L}(X)}} w + \widetilde{\omega}_2 \wedge s w
    \,,
  $$
  where on the right $w \in \{\omega_i, s \omega_i\}$.
\end{remark}


\begin{example}[Minimal model for $\mathrm{Cyc}(S^4)$ {\cite[Example 3.3]{FSS16a}}]
\label{MinimalDGCAlgebraModelForCyclicSpaceOfFourSphere}
By Remark \ref{MinimalDGCModelForCyclicLoopSpace}, a minimal DG-algebra model for the cyclification  of the 4-sphere is given by
$$
  \mathcal{O}(\mathrm{Cyc}(S^4))
  \;\simeq\;
  \mathbb{Q}[h_3, h_7, \omega_2, \omega_4, \omega_6]
  \Bigg/
  \left(
    \begin{aligned}
      d h_3 & = 0
      \\[-1mm]
      d h_7 & = -\tfrac{1}{2} \omega_4 \wedge \omega_4 + \omega_6 \wedge \omega_2
      \\[-1mm]
      d \omega_2 & = 0
      \\[-1mm]
      d \omega_4 & = h_3 \wedge \omega_2
      \\[-1mm]
      d \omega_6 & = h_3 \wedge \omega_4
    \end{aligned}
  \right).
$$
This exhibits the structure of a DG-algebra over $\mathbb{Q}[h_3]/(d h_3= 0)$, hence exhibiting
a rational model for a map
\begin{equation}
  \label{OverS3CycS4}
\mathrm{Cyc}(S^4)\longrightarrow S^3.
\end{equation}
\end{example}
\begin{example}[$\mathrm{Cyc}(S^4)$ covers 6-truncated twisted K-theory, rationally]
  \label{CyclificationOf4SphereReceives6TruncationOfTwistedK}
  Ex. \ref{MinimalDGCAlgebraModelForCyclicSpaceOfFourSphere} reveals a close relationship between $\mathrm{Cyc}(S^4)$ and the 6-truncation of $\Omega^{\infty-2}_{B^3 \mathbb{Q}} \left( \mathrm{ku} \dslash  B S^1 \right)$ in the rational approximation.
  In terms of the minimal models of Lemma \ref{TwistedKModel}, the 6-truncation $\Omega^{\infty-2}_{B^3 \mathbb{Q}} \left( \mathrm{ku} \dslash  B S^1 \right)\langle 6 \rangle$ is obtained simply by setting all $\omega_{\bullet>6}$ to zero.
We then have the following morphisms in the rational homotopy category:
  $$
    \xymatrix@R=-2pt{
      \Omega^{\infty-2}_{B^3 \mathbb{Q}} \left( \mathrm{ku} \dslash  B S^1 \right)
      \ar@{->}[r]^-{\tau_6}
      &
      \Omega^{\infty-2}_{B^3 \mathbb{Q}} \left( \mathrm{ku} \dslash  B S^1 \right)\langle 6 \rangle
      &
      \mathrm{Cyc}(S^4)
      \ar[l]_-{p}
      \\
      h_3
       &
      h_3
        \ar@{|->}[l]
        \ar@{|->}[r]
        &
      h_3
      \\
      \omega_2
        &
      \omega_2
        \ar@{|->}[l]
        \ar@{|->}[r]
      & \omega_2
      \\
      \omega_4
        &
      \omega_4
        \ar@{|->}[l]
        \ar@{|->}[r]
      & \omega_4
      \\
      \omega_6
        &
      \omega_6
        \ar@{|->}[l]
        \ar@{|->}[r]
      & \omega_6
      \\
      \omega_8 && h_7
      \\
      \vdots
    }
  $$
In Theorem \ref{TwistedKTheoryInsideFiberwiseStabilizationOfATypeOrbispaceOf4Sphere} we encounter lifts $\widehat \phi$
of morphisms of rational homotopy types
  $\phi\colon X \overset{\phi}{\longrightarrow} \mathrm{Cyc}(S^4)$ through this zig-zag, i.e.,
  maps $\widehat{\phi}$ making the following diagram of maps of rational homotopy types commute:
  \begin{equation}
    \label{LiftsThroughZigZag}
    \raisebox{45pt}{
    \xymatrix@C=6em{
      &&
      \Omega^{\infty-2}_{B^3 \mathbb{Q}} \left( \mathrm{ku} \dslash  B S^1 \right)
      \ar[d]^{ \tau_6 }
      \\
      &&
      \Omega^{\infty-2}_{B^3 \mathbb{Q}} \left( \mathrm{ku} \dslash  B S^1 \right)\langle 6\rangle
      \\
      X_{\mathbb{Q}}
      \ar[rr]^-{\phi}
      \ar@{-->}@/^1pc/[uurr]^{\widehat \phi}
      &&
      \mathrm{Cyc}(S^4).
      \ar[u]_-{p}
    }
    }
  \end{equation}
\end{example}

\begin{prop}[Minimal DG-module for fiberwise stabilization of $\mathrm{Cyc}(S^4)$]
  \label{MinimaldgModuleForFiberwiseStabilisationOfCyclicSpaceOf4Sphere}

 \item {\bf (i)} A minimal DG-module (Def. \ref{MinimalDGModule}) modelling the fiberwise stabilization
  (Prop. \ref{AdjunctionStabilization}) of the
  cyclic loop space of the 4-sphere (Ex. \ref{MinimalDGCAlgebraModelForCyclicSpaceOfFourSphere})
  over $S^3$ (via \eqref{OverS3CycS4}) is
  $$
    \mathcal{M}_{\mathbb{Q}[S^3]}
    \left(
        \Sigma^\infty_{+,S^3} \mathrm{Cyc}(S^4)
    \right)
    \;\simeq\;
    \frac{
      \mathbb{Q}[ h_3, \omega_2, \omega_4, \omega_6 ]
    }
    {
      (\omega_6 \wedge \omega_2 -\tfrac{1}{2} \omega_4 \wedge \omega_4)
    }
  \Bigg/
  \left(
    \begin{aligned}
      d h_3 & = 0
      \\[-1mm]
      d \omega_2 & = 0
      \\[-1mm]
      d \omega_4 & = h_3 \wedge \omega_2
      \\[-1mm]
      d \omega_6 & = h_3 \wedge \omega_4
    \end{aligned}
  \right)
  \;\in\;
  \mathbb{Q}[h_3]\mathrm{\mbox{-}Mod}\;.
  $$
  Here the module structure is the evident one induced by multiplication in $\mathbb{Q}[h_3]$.

 \item {\bf (ii)} There is a quasi-isomorphism from this minimal model to the DG-module underlying the minimal DG-algebra from Ex. \ref{MinimalDGCAlgebraModelForCyclicSpaceOfFourSphere}
  $$
    \frac{
      \mathbb{Q}[ h_3, \omega_2, \omega_4, \omega_6 ]
    }
    {
      (\omega_6 \wedge \omega_2 -\tfrac{1}{2} \omega_4 \wedge \omega_4)
    }
    \Bigg/
    \left(
      {\begin{aligned}
        d h_3 & = 0
        \\[-1mm]
        d \omega_2 & = 0
        \\[-1mm]
        d \omega_4 & = h_3 \wedge \omega_2
        \\[-1mm]
        d \omega_6 & = h_3 \wedge \omega_4
      \end{aligned}}
    \right)
    \xrightarrow{\;\;\simeq_{\mathrm{qi}}\;\;}
    \mathbb{Q}[ h_3, h_7, \omega_2, \omega_4, \omega_6 ]
    \Bigg/
    \left(
      {\begin{aligned}
        d h_3 & = 0
        \\[-1mm]
        d h_7 & = -\tfrac{1}{2}\omega_4 \wedge \omega_4 + \omega_2 \wedge \omega_6
        \\[-1mm]
        d \omega_2 & = 0
        \\[-1mm]
        d \omega_4 & = h_3 \wedge \omega_2
        \\[-1mm]
        d \omega_6 & = h_3 \wedge \omega_4
      \end{aligned}}
    \right)
  $$
  given by any choice of linear splitting of the underlying quotient map of graded algebras, for
  example, by the map sending equivalence classes on the left to their unique representatives on the right that are
  at most linear in $\omega_4$.
\end{prop}

\begin{proof}
 For {\bf(i)}: by Prop. \ref{FiberwiseSuspensionSpectrumdgModel} and Lemma \ref{MinimalDgModulesForFiberSpectra}
  the underlying graded $\mathbb{Q}[h_3]$-module is the free $\mathbb{Q}[h_3]$-module
  on the cohomology of the homotopy cofiber of
  $$
    \mathbb{Q}[h_3]
    \longrightarrow
    \mathbb{Q}[ h_3, h_7, \omega_2, \omega_4, \omega_6 ]
    \Bigg/
    \left(
      {\begin{aligned}
        d h_3 & = 0
        \\[-1mm]
        d h_7 & = -\tfrac{1}{2}\omega_4 \wedge \omega_4 + \omega_2 \wedge \omega_6
        \\[-1mm]
        d \omega_2 & = 0
        \\[-1mm]
        d \omega_4 & = h_3 \wedge \omega_2
        \\[-1mm]
        d \omega_6 & = h_3 \wedge \omega_4
      \end{aligned}}
    \right),
  $$
  where the minmal DG-algebra on the right is from Ex. \ref{MinimalDGCAlgebraModelForCyclicSpaceOfFourSphere}.
  This cofiber cohomology is evidently the graded-commutative algebra
  $$
    \frac{
      \mathbb{Q}[\omega_2, \omega_4, \omega_6 ]
    }
    {
      (\omega_6 \wedge \omega_2 -\tfrac{1}{2} \omega_4 \wedge \omega_4)
    }\;,
  $$
  obtained as the quotient by the two-sided tensor ideal generated by $\omega_6 \wedge \omega_2 -\tfrac{1}{2} \omega_4\wedge\omega_4$.
  The graded vector space underlying the minimal DG-module is therefore
  $$
    \mathbb{Q}[h_3]
    \otimes
    \frac{
      \mathbb{Q}[\omega_2, \omega_4, \omega_6 ]
    }
    {
      (\omega_6 \wedge \omega_2 -\tfrac{1}{2} \omega_4 \wedge \omega_4)
    }
    \;\simeq\;
    \frac{
      \mathbb{Q}[ h_3, \omega_2, \omega_4, \omega_6 ]
    }
    {
      (\omega_6 \wedge \omega_2 -\tfrac{1}{2} \omega_4 \wedge \omega_4)
    }\;.
  $$
  The differential on this must be such that fiberwise stabilization does not change the cohomology (by Prop. \ref{FiberwiseSuspensionSpectrumdgModel}).
  This completely determines the differential, fixing it as claimed.
  The second point {\bf (ii)} now follows at once.
\end{proof}

%
%
%
%

\section{The A-type orbispace of the 4-sphere}
\label{TheATypeOrbispaceOfThe4Sphere}

In this section we consider a particular circle action on the 4-sphere, as well as the induced homotopy quotient, which we call
the \emph{A-type orbispace of the 4-sphere} (see Def. \ref{ATypeOrbispaceOf4Sphere} and Rem. \ref{OrbispaceTerminology} below).
We first provide an informal string-theoretic motivation for considering this space in Rem. \ref{TheATypeQuotientFromSpacetime},
and then substantiate this by a more formal mathematical analysis.
After establishing some results on the rational homotopy type of the A-type orbispace in Sec. \ref{RationalHomotopyTypeOfATypeOrbispace},
our main result Theorem \ref{TwistedKTheoryInsideFiberwiseStabilizationOfATypeOrbispaceOf4Sphere} shows that, rationally, there is
a copy of twisted K-theory in the fiberwise stabilization of the A-type orbispace, fibered over the 3-sphere.
In Sec. \ref{TheMechanism}, we demonstrate how this result witnesses the phenomenon of gauge enhancement of M-branes.
%
%

\begin{defn}[The A-type orbispace of the 4-sphere]
  \label{ATypeOrbispaceOf4Sphere}
Writing $S^4$ as the unit sphere in $\mathbb{R}^5 =\mathbb{R}\oplus \mathbb{C}^2$, the identification
  \begin{equation}
    \label{SU2ActionOn4Sphere}
    S^4 \;=\; S( \mathbb{R} \oplus \!\!\!\!\! \xymatrix{ \mathbb{C}^2 \ar@(ul,ur)[]^{ {\rm SU}(2)_L } } \!\!\!\!\! )
  \end{equation}
  shows that $S^4$ inherits an action of ${\rm SU}(2)$.
  Specifically, on the right-hand side above we are referring to the defining linear representation of ${\rm SU}(2)$ on $\mathbb{C}^2$, regarded as a \emph{left} action.
  This restricts along the canonical inclusion $S^1 \simeq {\rm U}(1)\hookrightarrow {\rm SU}(2)$ to define an $S^1$-action on $S^4$.
We refer to the corresponding homotopy quotient \eqref{BorelConstruction}
  \begin{equation}
    \label{ATypeOrbispace}
    S^4 \dslash S^1
    \;\simeq\;
    S^4 \times_{S^1} E S^1
  \end{equation}
 as the \emph{A-type orbispace of the 4-sphere}.
\end{defn}


\begin{remark}[$A$-series vs. $S^1$]
  \label{OrbispaceTerminology}
  The terminology in Def. \ref{ATypeOrbispaceOf4Sphere} is motivated as follows:
  the finite subgroups of $SU(2)$ have a famous ADE-classification, corresponding to the simply-laced Dynkin diagrams.
  The finite subgroups in the A-series are cyclic and, up to conjugation, are all subgroups of
  the canonical copy of $S^1$ inside ${\rm SU}(2)$:
  $$
  \xymatrix{
    \mathbb{Z}_{n+1} \; \ar@{^{(}->}[r]& S^1 \simeq {\rm U}(1) \; \ar@{^{(}->}[r] & {\rm SU}(2)
    }.
  $$
  The $S^1$-action considered in Def. \ref{ATypeOrbispaceOf4Sphere} is thus the limiting case (as $n \to \infty$)
  of the A-series actions.
  Now the homotopy quotient of the smooth 4-sphere such a finite group action
  is an \emph{orbifold}, hence an \emph{A-type orbifold} for an A-series group action
  (see \cite{ADE} for further discussion).
    More generally, homotopy quotients by
  (possibly non-finite) topological groups are \emph{orbispaces} \cite{HenriquesGepner07}, whence our terminology.
\end{remark}

The following result is immediate:
\begin{prop}[Quotient and fixed points of the A-type orbispace]
  \label{SystemOfFixedPointsAndQuotientsOfATypeActionOn4Sphere}
  For the A-type $S^1$-action on $S^4$ (Def. \ref{ATypeOrbispaceOf4Sphere}), it holds that:

 \item {\bf (i)}
  The ordinary quotient space is the 3-sphere:
  $
    S^4 / S^1 \;\simeq\; S^3
   $.
  Hence, via 
 \eqref{HomotopyQuotientReceiving} there is a canonical map from the
  A-type orbispace of the 4-sphere to the 3-sphere:
  \begin{equation}
    \label{MapFromATypeOrbispaceTo3Sphere}
    \raisebox{12pt}{\xymatrix@R=6pt{
      S^4   \dslash  S^1
      \ar[rr]
      &&
      S^3
      \\
      S^4 \times_{S^1} E S^1
      \ar[rr]_{ \mathrm{id}\times_{S^1} p }
      \ar@{=}[u]
      &&
      S^4 \times_{S^1} \ast
      \ar@{=}[u]
    }}
  \end{equation}

\vspace{-3mm}
\item  {\bf (ii)} The space of $S^1$-fixed points is the 0-sphere, included as two antipodal points
  $$
    S^0 \;=\; \left(S^4\right)^{S^1}
    \xymatrix{\ar@{^{(}->}[r]&}
     S^4
    \,.
  $$
In summary, we have the following system of spaces over $S^3$:
\begin{equation}
  \label{QuotientOfS4ByS1OverS3}
  \raisebox{20pt}{\xymatrix@R=1pt@C=4em{
    \overset{
    \mbox{
        \tiny
        \begin{tabular}{c}
          \emph{Fixed}\;\;
          \\
          \emph{points}\;\;
        \end{tabular}}}
    {\overbrace{S^0 = \big(S^4\big)^{S^1}}}
    \ar[dddrr]
    \;\ar@{^{(}->}[r]
    &
    \overset{
    \mbox{
        \tiny
        \begin{tabular}{c}
        \emph{4-sphere}\;\;\;
        \end{tabular}}}
    {\overbrace{S^4}}
    \ar[dddr]
    \ar[r]
    &
    \overset{
    \mbox{
        \tiny
        \begin{tabular}{c}
          \emph{Homotopy}\;\;
          \\
          \emph{quotient}\;\;
        \end{tabular}}}
    {\overbrace{S^4\dslash S^1}}
    \ar[ddd]
    \ar[r]
    &
    \overset{
    \mbox{
        \tiny
        \begin{tabular}{c}
          \emph{Naive}\;\;
          \\
          \emph{quotient}\;\;
        \end{tabular}}}
    {\overbrace{S^4/ S^1}}
    \ar@{=}[dddl]
    \\
    \\
    \\
    &&
    S^3
  }}
\end{equation}
\end{prop}

Below in Sec. \ref{TheMechanism} we regard the (rationalization of the) A-type orbispace of the 4-sphere 
as the \emph{coefficient} of a generalized cohomology theory.
However, as explained in \cite[Sec. 2.2]{ADE}, the 4-sphere coefficient here ultimately originates as a
factor in a black M5-brane spacetime $\sim \mathrm{AdS}_7 \times S^4$.
 With this in mind, the spaces appearing in \eqref{QuotientOfS4ByS1OverS3}
readily explain those spaces appearing in 
the string theory literature.

\begin{remark}[The A-type orbispace from black M5-brane geometry]
 \label{TheATypeQuotientFromSpacetime}

\item {\bf{(i)}} The near-horizon geometries of black M2-brane and black M5-brane solutions of
11-dimensional supergravity are given by $\mathrm{AdS_4} \times S^7$ and
$\mathrm{AdS}_7 \times S^4$, respectively \cite{Gueven92}.
Both of the spherical factors admit natural maps to the four-sphere, namely the quaternionic Hopf fibration
$H_\mathbb{H}\colon S^7 \to S^4$ and the identity map $S^4 \to S^4$, and these maps generate the torsion-free
homotopy of $S^4$. It is natural to posit that $S^4$ is the coefficient for a \emph{nonabelian} cohomology theory
(in this case \emph{cohomotopy}) that measures M-brane charge in the spirit of Dirac charge quantization \cite{Freed00},
at least rationally \cite{S-top, cohomotopy, FSS16a}.
%
\begin{equation}
  \label{SpacetimeMaps}
  \raisebox{50pt}{\xymatrix@R=1em@C=5em{
    \big[
     \underset{
       \tiny
       \begin{tabular}{c}
         black M2-brane
         \\
         spacetime
       \end{tabular}
      }{\underbrace{ \mathrm{AdS}_4 \times S^7}}
    \ar[r]^-{{\rm pr}_2}
    \ar@/^2pc/[rr]^{ \mbox{ \tiny \begin{tabular}{c} one unit of \\ M2-brane charge \end{tabular} } }
    &
    \underset{
      \!\!\!\!\!\!\!\!\!\!\!\!\!\!
      \!\!\!\!\!\!\!\!\!\!\!\!\!\!
      \mbox{
        \tiny
        \begin{tabular}{c}
          sphere around
          \\
          M2-brane
          \\
          singularity
        \end{tabular}
      }
      \!\!\!\!\!\!\!\!\!\!\!\!\!\!
      \!\!\!\!\!\!\!\!\!\!\!\!\!\!
    }{
    \underbrace{
      S^7
    }}
    \ar[r]^-{ H_{\mathbb{H}} }
    &
    S^4
    \big]
    &
    \in \big[Y, S^4\big]
    \\
    &&
    &
    \mbox{
      \footnotesize
      \begin{tabular}{c}
         cohomotopy classes
        \\
        of $Y$ in degree 4
      \end{tabular}
    }
    \\
    \big[
     \underset{
       \tiny
       \begin{tabular}{c}
         black M5-brane
         \\
         spacetime
       \end{tabular}
      }{\underbrace{ \mathrm{AdS}_7 \times S^4}}
    \ar[r]^-{{\rm pr}_2}
    \ar@/^2pc/[rr]^{ \mbox{ \tiny \begin{tabular}{c} one unit of \\ M5-brane charge \end{tabular} } }
    &
    \underset{
      \!\!\!\!\!\!\!\!\!\!\!\!\!\!
      \!\!\!\!\!\!\!\!\!\!\!\!\!\!
      \mbox{
        \tiny
        \begin{tabular}{c}
          sphere around
          \\
          M5-brane
          \\
          singularity
        \end{tabular}
      }
      \!\!\!\!\!\!\!\!\!\!\!\!\!\!
      \!\!\!\!\!\!\!\!\!\!\!\!\!\!
    }{
    \underbrace{
      S^4
    }}
    \ar[r]^{ \mathrm{id}}
    &
    S^4
    \big]
    &
    \in \big[Y, S^4\big]
  }}
\end{equation}
\item {\bf (ii)}
More generally, the black M5-brane may sit inside an $\mathrm{MK6}$, which itself is located at the singular locus of a global orbifold
$$
  \mathrm{AdS}_7 \times S^4 \dslash G_{\mathrm{ADE}}
$$
(see \cite[Ex. 2.7]{ADE} for a precise statement and for pointers to the literature),
where  $G_{\mathrm{ADE}} \subset SU(2)$ is a finite subgroup acting on the 4-sphere via the identification
$$
  \xymatrix{S^4 \ar@(ul,ur)[]^{ G_{\mathrm{ADE}} }}
  \;\simeq\;
  S( \mathbb{R} \oplus \!\!\!\xymatrix{ \mathbb{C}^2 \ar@(ul,ur)[]^{ G_{\mathrm{ADE}} } } \!\!\! )
  \,.
$$
In order for the 4-sphere charge coefficient to be able to measure the unit charges of such M5-branes sitting at ADE singularities
in a manner generalizing \eqref{SpacetimeMaps},
it must be equipped with that same group action.
The resulting \emph{equivariant} cohomotopy theory for M-branes is the subject of \cite{ADE}.

\item {\bf (iii)}  Our current focus is on the A-series subgroups, which up to conjugation are the cyclic subgroups
$
  \mathbb{Z}_{n+1} \hookrightarrow  S^1 = {\rm U}(1) \hookrightarrow  {\rm SU}(2)
$,
as in Rem. \ref{OrbispaceTerminology}.
By analogy with the case for M2-branes as in \cite[p. 3]{ABJM08}, we may interpret the $S^1$-action as being that of the
M-theory circle fibration over 10d type IIA supergravity.
With this interpretation in mind, passage to the finite A-type orbifold quotient
$$
  \mathrm{AdS}_7 \times S^4 \dslash \mathbb{Z}_{n+1}
$$
corresponds to shrinking the M-theory circle fiber, and hence the coupling constant of non-perturbative type IIA string theory,
by the factor $n+1$. The limit $n \to \infty$, in which the cyclic groups $\mathbb{Z}_{n+1}$ exhausts the group $S^1$, corresponds to the limit of perturbative type IIA string theory. Via the maps of \eqref{HomotopyQuotientReceiving}:
%
$$
  \scalebox{.9}{
  \xymatrix@R=7pt{
    \fbox{
      \begin{tabular}{c}
        M-theoretic
        \\
        near horizon spacetime
        \\
        of black M5-brane
      \end{tabular}
    }
    &
    \mathrm{AdS}_7 \times S^4 \hspace{-.75mm}
    \ar[d]
    \\
    \fbox{
      \begin{tabular}{c}
        M-theoretic
        \\
        near horizon spacetime
        \\
        of black M5-brane at A-type singularity
        \\
        for coupling $g/(n+1)$
      \end{tabular}
    }
    &
    \mathrm{AdS}_7 \times S^4 \dslash \mathbb{Z}_{n+1}
    \ar[d]
    \\
    \fbox{
      \begin{tabular}{c}
        Type IIA string-theoretic
        \\
        near horizon spacetime
        \\
        of black NS5-brane inside black D6-brane
      \end{tabular}
    }
    &
    \mathrm{AdS}_7 \times S^4 \dslash S^1
    \ar[d]
    \\
    \fbox{
      \begin{tabular}{c}
        Type IIA string-theoretic
        \\
        near horizon spacetime
        \\
        of black NS5-brane
      \end{tabular}
    }
    &
    \mathrm{AdS}_7 \times S^3.
  }
  }
$$
Applying the same logic as before, we might expect that the A-type orbispace $S^4 \dslash S^1$ serves as the charge quantization coefficient when M-branes
 are identified with their dual
incarnations as D-branes in type IIA string theory. That this is indeed the case, up to a subtlety related to fiberwise stabilization,
is essentially our result on gauge enhancement.

\item {\bf (iv)} While the A-type orbispace $S^4 \dslash S^1$ has not previously featured in the string theory literature,
the \emph{ordinary} quotient space $S^4 / S^1 \simeq S^3$ has been considered in this context.
We briefly survey the related literature:
\begin{itemize}
\item The dimensional reduction of 11-dimensional supergravity on the 4-sphere factor yields
a maximal ${\rm SO}(5)$-gauged supergravity
in seven dimensions \cite{PNT}. The consistency of this reduction is established in
\cite{NVvN} and a systematic classification of such reductions is given in \cite{FS2}.
On the other hand, the reduction of type IIA
supergravity on $S^3$ leads to an
${\rm SO}(4)$-gauged supergravity
in seven dimensions.
To compare these two gauged supergravity theories,
one needs a means of breaking the ${\rm SO}(5)$ gauge symmetry.
In
\cite{CLPST} the comparison between the two reductions is achieved using the singular scaling limit of $S^4$
opening up to $S^3 \times \R$, based on earlier
arguments \cite{HW,CLLP}. The consistency of
such reductions is studied and established in \cite{CLP}.

\item  Reductions with less symmetry are
also possible, for instance by gauging only a left-acting ${\rm SU}(2)$ subgroup of
${\rm Spin}(4) \cong {\rm SU}(2)_L \times {\rm SU}(2)_R$ \cite{CS}.
In \cite{NV}, this was achieved using a singular limit of the
$S^4$ reduction of 11d supergravity. In \cite[Sec. 2.2]{ADE} it is explained
how the distinction between these actions relates to the 4-sphere detecting
black M5-branes as well as black M2-branes at singularities.
\end{itemize}
\end{remark}

\begin{remark}[Other circle actions on the 4-sphere]
Circle actions on spheres form one of the most important problems in the theory of transformation
groups. For  a compact connected topological group $G$ acting non-trivially
on the 4-sphere, requiring orbits to be of dimension $\leq 1$ immediately forces $G$ to be  the circle group.
However, the action may not be equivalent to a differentiable one
\cite{MZ}.
Furthermore, there are infinitely many nonlinear circle actions on $S^4$
\cite{Pao}.
Since $S^4$ is compact, it follows (by applying \cite[Theorem 7.33]{FOT}) that in any case the fixed point set $F$
will have the same Euler characteristic as $S^4$, namely 2.
Since the sum of dimensions of the cohomology groups of $F$ is always at most as large as the corresponding sum for the space $S^4$
\cite[Theorem 7.37]{FOT}, this forces 
$\dim H^{\rm ev}(F; \Q)=2$ and $\dim H^{\rm odd}(F; \Q)=0$.
Away from the trivial case ($F\simeq_\mathbb{Q} S^4$), this implies that there are only two possibilities for the rational homotopy type of the fixed point space: either $F$ is rationally a 0-sphere (union of two points), or $F$ is a rational homology 2-sphere.
The latter case is described in \cite[Ex. 7.39]{FOT}; in this article we deal exclusively with the former case.
\end{remark}

\subsection{Rational homotopy type of the A-type orbispace}
  \label{RationalHomotopyTypeOfATypeOrbispace}

In this section we study the rational homotopy type of the A-type orbispace of the 4-sphere (Def. \ref{ATypeOrbispaceOf4Sphere})
and apply the $\mathrm{Ext}/\mathrm{Cyc}$-adjunction to it. 
The main result is Prop. \ref{MinimalDGCAlgebraModelForATypeOrbispace} below, but first we establish some preliminary results:

\begin{lemma}[Rational homotopy and cohomology of
$S^4 \dslash S^1$]
  \label{HomotopyGroupsOfS4OverS1}
  For \emph{every} $S^1$-action on the 4-sphere, the resulting homotopy quotient
  $S^4 \dslash S^1$
  has the following properties:
    \item {\bf (i)}
    Its rational homotopy groups are
    \begin{equation}
      \label{ATypoeOrbispaceRationalHomotopyGroups}
      \pi_\bullet^{\mathbb{Q}}(S^4 \dslash S^1) := \pi_\bullet(S^4 \dslash S^1)\otimes \mathbb{Q}\simeq
      \begin{cases}
      \mathbb{Q} & \mbox{\emph{in dimensions 2, 4 and 7}}\\
      \,0& \mbox{\emph{otherwise}}.
      \end{cases}
    \end{equation}
    \item {\bf (ii)}
     Its rational cohomology groups are
     \begin{equation}
       \label{ATypeOrbispaceRationalCohomology}
       H^\bullet(S^4 \dslash S^1,\mathbb{Q})
       \;\simeq\;
       \langle
           \widetilde \omega_0
       \rangle
       \oplus
       \langle
           \omega_2
       \rangle
       \oplus
       \underset{k \in \mathbb{N}}{\bigoplus}\,
       \big\langle \,
             \omega_2^{\wedge(k+2)}
             \,,\,
             \omega_4 \wedge \omega_2^{\wedge k}
       \, \big\rangle
       \,,
     \end{equation}
     so that $\dim H^{2k}(S^4\dslash S^1,\mathbb{Q})$ is $1$ for $k=0,1$, is $2$ for $k\geq 2$, and the odd cohomology vanishes.
\end{lemma}
\begin{proof}
  The first statement follows with the long exact sequence of rational homotopy groups induced by the homotopy fiber sequence
  \[
  S^1 \longrightarrow S^4 \longrightarrow S^4 \dslash S^1\,.
  \]
  The second statement follows with the corresponding multiplicative Serre spectral sequence in rational cohomology (though we make no claims regarding the \emph{algebra} structure on cohomology---the notation is merely suggestive of the manner in which these classes arise in the Serre spectral sequence).
\end{proof}

\begin{lemma}[DG-algebra model for general $S^4 \dslash  S^1$]
  \label{dgcAlgebraModelForATypeOrbispaceOf4Sphere}
  For \emph{every} $S^1$-action on the 4-sphere, the minimal DG-algebra (via Prop. \ref{SullivanEquivalence}) of the
  resulting homotopy quotient (Def. \ref{HomotopyQuotient})
   is of the form
  \begin{equation}
    \label{dgcModelForHomotopyQuotientOf4SphereByCircleAction}
    \mathcal{O}\left(
      S^4 \dslash S^1
    \right)
    \;\simeq\;
    \mathbb{Q}[ \omega_2, \omega_4, h_7 ]\bigg/
    \left(
      \begin{aligned}
        d \omega_2 & = 0
        \\[-1mm]
        d \omega_4 & = 0
        \\[-1mm]
        d h_7
          & =
        -\tfrac{1}{2} \omega_4 \wedge \omega_4
        \\
        & \phantom{=}
        +
        c_1 \, \omega_2^{\wedge 4}
        +
        c_2 \, \omega_2^{\wedge 2} \wedge \omega_4
      \end{aligned}
    \right)
    \;\in\;
    \mathrm{DGCAlg}
  \end{equation}
  for some coefficients $c_1, c_2 \in \mathbb{R}$.
\end{lemma}

\begin{proof}
  By Prop. \ref{SullivanEquivalence}, the minimal DG-algebra model of $S^4 \dslash S^1$ has the following properties:
  \begin{enumerate}
    \item as a graded algebra, it is generated by the rational homotopy groups;
    \item the differential on the minimal model is such that the cochain cohomology reproduces the rational cohomology of $S^4\dslash S^1$.
  \end{enumerate}
By Lemma \ref{HomotopyGroupsOfS4OverS1}, we therefore have that the underlying graded commutative algebra of the minimal model
  of \emph{any} $S^4\dslash S^1$ is $\mathbb{Q}[ \omega_2, \omega_4, h_7 ]$.
  By the second item in Lemma \ref{HomotopyGroupsOfS4OverS1} and for degree reasons, the differential is necessarily of the form
  \[
  d\omega_2 =0,
  \qquad\quad
  d\omega_4 =0,
  \qquad\quad
  dh_7 \neq 0.
  \]
There is a homotopy fiber sequence
\[
S^4 \longrightarrow S^4 \dslash S^1 \longrightarrow BS^1,
\]
which in the rational models is reflected by the requirement that setting $\omega_2$ to zero in $\mathcal{O}(S^4\dslash S^1)$ produces a (necessarily minimal) DG-algebra model for $S^4$.
Comparing with Ex. \ref{MinimalDgcAlgebraModelFor4Sphere}, this means  that
\[
dh_7 = -\tfrac{1}{2} \omega_4\wedge \omega_4 + \text{terms of degree 8 at most linear in $\omega_4$},
\]
which completes the proof.
\end{proof}

\begin{lemma}[Rational Ext/Cyc-adjunction unit at $S^4\dslash S^1$]
  \label{ExtCycAdjunctionForATypeOrbispace}
  For any $S^1$-action on $S^4$,
  the composite of the 
  Ext/Cyc-adjunction unit (Theorem \ref{GCycExtAdjunction})
  at $S^4 \dslash S^1$ with the cyclification of the equivalence of Prop. \ref{ExtensionOfHomotopyQuotientEquivalentToOriginalSpace} is presented by the map of minimal DG-algebra models
  \begin{equation}
    \label{AdjunctionUnitOnRationalATypeOrbispaceISH}
    \raisebox{45pt}{\xymatrix@R=-3pt@C=4em{
      S^4 \dslash S^1
        \ar[r]^-{\eta_{S^4 \dslash S^1}}
      &
      \mathrm{Cyc}\,\mathrm{Ext}(S^4 \dslash S^1)
      \ar[r]^-{\simeq_{\mathrm{whe}}}
      &
      \mathrm{Cyc}(S^4)
      \\
      \omega_2
      \vphantom{
      \omega_2^{\wedge 3}
      }
       && \omega_2 \ar@{|->}[ll]
      \\
      \omega_4
      \vphantom{
      \omega_2^{\wedge 3}
      }
      && \omega_4 \ar@{|->}[ll]
      \\
      {
       c_1 \, \omega_2^{\wedge 3}
        + c_2 \, \omega_2 \wedge \omega_4
      }
      && \omega_6 \ar@{|->}[ll]
      \\
      0
      \vphantom{
      \omega_2^{\wedge 3}
      }
      && h_3 \ar@{|->}[ll]
      \\
      h_7
      \vphantom{
      \omega_2^{\wedge 3}
      }
      && h_7 \ar@{|->}[ll]
    }}
  \end{equation}
  where $c_1, c_2 \in \mathbb{R}$ are the constants as in \eqref{dgcModelForHomotopyQuotientOf4SphereByCircleAction}.
\end{lemma}
\begin{proof}
We determine what the map $S^4\dslash S^1 \to \mathrm{Cyc}(S^4)$ does on generators of the rational homotopy groups.
With this we can posit a map on minimal models, which turns out to be uniquely specified for degree reasons.
To begin, we observe that the degree-3 generator $h_3$ must be sent to zero (since $S^4 \dslash S^1$ has no free
homotopy in dimension $3$). The degree-2 generator $\omega_2$ on the right of \eqref{AdjunctionUnitOnRationalATypeOrbispaceISH}
is sent to the generator of the same name on the left: all morphisms considered are over $BS^1$, and $\omega_2$ generates
the minimal model of this space (Ex. \ref{MinimalDGCAlgebraModelForClassifyingSpace}).
The zig-zag identity of the $\mathrm{Ext}/\mathrm{Cyc}$-adjunction gives us a commuting diagram
  $$
    \xymatrix{
      S^4
      \ar@{}[r]|-{\simeq_{\mathrm{whe}}}
      &
      \mathrm{Ext}(S^4\dslash S^1)
      \ar[drrr]_{\mathrm{id       }}
      \ar[rrr]^-{ \mathrm{Ext}(\eta_{S^4\dslash S^1}) }
      &&&
      \mathrm{Ext}\,\mathrm{Cyc}\,\mathrm{Ext}(S^4\dslash S^1)
      \ar[d]^{ \epsilon_{\mathrm{Ext}(S^4  \dslash S^1)} }
      \ar@{}[r]|-{ \simeq_{\mathrm{whe}} }
      &
      \mathrm{Maps}(S^1, S^4)
      \ar[d]^{ \mathrm{ev}_0 }
      \\
      &&
      &&
      \mathrm{Ext}(S^4\dslash S^1)
      \ar@{}[r]|-{\simeq_{\mathrm{whe}}}
      &
      S^4
    }
  $$
  with weak homotopy equivalences as shown due to Prop. \ref{ExtensionOfHomotopyQuotientEquivalentToOriginalSpace}.
  Examining the DG-algebra model of the free loop space (Remark \ref{MinimalDGCModelForCyclicLoopSpace}),
  this means that the unit $\eta$ sends non-shifted (non-looped) algebra generators to themselves. These generators are $\omega_4$ and $h_7$.

  So far we have defined the desired map of DG-algebras on the generators $\omega_2$, $\omega_4$, $h_3$, and $h_7$.
  This map respects the differentials on $\omega_2$, $\omega_4$ and $h_3$, whereas respect for the differential on $h_7$
  $$
    \xymatrix{
      {\begin{aligned}
        & -\tfrac{1}{2} \omega_4 \wedge \omega_4
        \\
        & + \omega_2 \wedge \big( c_1 \omega_2^{\wedge 3} + c_2  \omega_2\wedge\omega_{4}  \big)
      \end{aligned}}
      &&
      {\begin{aligned}
        & -\tfrac{1}{2} \omega_4 \wedge \omega_4
        \\
        &
        + \omega_2 \wedge \omega_6
      \end{aligned}}
      \ar@{|-->}[ll]
      \\
      h_7
      \ar@{|->}[u]^-d
      &&
      h_7
      \ar@{|->}[ll]
      \ar@{|->}[u]_-d
    }
  $$
forces $\omega_6 \mapsto c_1 \omega_2^{\wedge 3} + c_2 \omega_2 \wedge \omega_4$.
This completely determines the map of minimal models \eqref{AdjunctionUnitOnRationalATypeOrbispaceISH}.
\end{proof}

\begin{example}[DG-algebra model for homotopy quotient of the trivial action]
  \label{dgcAlgebraModelForHomotopyQuotientOfTrivialAction}
The homotopy quotient of the trivial $S^1$-action on $S^4$ is $S^4 \dslash S^1 \simeq S^4 \times BS^1$ (Ex.\ref{HomotopyQuotientOfTrivialGAction}).
  The minimal DG-algebra model of this product space is obtained by setting $c_1=c_2 =0$ in Lemma \ref{dgcAlgebraModelForATypeOrbispaceOf4Sphere}:   $$
    \mathcal{O}\left(
      S^4 \times B S^1
    \right)
    \;\simeq\;
    \mathbb{Q}[ \omega_2, \omega_4, h_7 ]\bigg/
    \left(
      \begin{aligned}
        d \omega_2 & = 0
        \\[-1mm]
        d \omega_4 & = 0
        \\[-1mm]
        d h_7
          & =
        -\tfrac{1}{2} \omega_4 \wedge \omega_4
      \end{aligned}
    \right)
    \;\in\;
    \mathrm{DGCAlg}
    \,.
  $$
\end{example}

\begin{lemma}[A-type action on 4-sphere is rationally trivial]
  \label{ATypeActionIsRationallyTrivial}
  The A-type circle action on the 4-sphere (Def. \ref{ATypeOrbispaceOf4Sphere}) is rationally trivial.
  That is, the action is represented in rational homotopy theory by the coprojection map of DG-algebras
  \begin{align*}
  \mathcal{O}(S^4) &\longrightarrow  \mathcal{O}(S^1)\otimes \mathcal{O}(S^4)\\
  \eta &\longmapsto 1\otimes \eta.
  \end{align*}
 In particular, the A-type orbispace $S^4 \dslash S^1$ is equivalent to the rationalization of the trivial action (Ex. \ref{dgcAlgebraModelForHomotopyQuotientOfTrivialAction}).
\end{lemma}
\begin{proof}
It is sufficient to argue on minimal models:
the action $\mu\colon S^1\times S^4 \to S^4$ determines a dual map in the category of DG-algebras
$\mu^\ast\colon \mathcal{O}(S^4) \to \mathcal{O}(S^1)\otimes \mathcal{O}(S^4)$ via Prop. \ref{SullivanEquivalence}.
Since minimal models are cofibrant and fibrant in the model structure on DG-algebras (see \cite{Hess06} for a review),
the map $\mu^\ast$ is homotopic to a map between minimal DG-algebra models:
\[
\mathbb{Q}[\omega_4, \omega_7]
  \bigg/
  \left(
    \begin{aligned}
      d \omega_4 & = 0 \\[-1mm]
      dh_7 &= -\tfrac{1}{2} \omega_4 \wedge \omega_4
    \end{aligned}
  \right)
  \xrightarrow{\quad \nu \quad}
 \mathbb{Q}[\theta_1, \omega_4, \omega_7]
  \bigg/
  \left(
    \begin{aligned}
      d \theta_1 & = 0
      \\[-1mm]
      d \omega_4 & = 0 \\[-1mm]
      dh_7 &= -\tfrac{1}{2} \omega_4 \wedge \omega_4
    \end{aligned}
  \right)\!.
\]
The source of this map is the minimal DG-algebra model of the 4-sphere (Ex. \ref{MinimalDgcAlgebraModelFor4Sphere}),
and the extra degree-1 generator $\theta_1$ appearing in the target corresponds to $\pi_1 (S^1)= \mathbb{Z}$.
The map $\nu$ is completely determined by the images of the generators $\omega_4$ and $h_7$.

Now, ${\rm SU}(2)$ acts on $S^4$ via the inclusion ${\rm SU}(4)\hookrightarrow {\rm SO}(5)$ (compare with \eqref{SU2ActionOn4Sphere}).
In particular, the ${\rm SU}(2)$-action preserves the round volume form on $S^4$.
Restricting along $S^1 \hookrightarrow {\rm SU}(2)$ and observing that the generator $\omega_4$ represents the round
volume form in cohomology forces $\nu(\omega_4) = \omega_4$.
Up to non-zero scaling, the only way to define $\nu$ on the degree-7 generator $h_7$ that respects the differential
is $\nu(h_7) = h_7$.
%
%
\end{proof}

\begin{remark} Some remarks on the above results are in order:
\item {\bf (i)} In the above proof, we refer to the degree-7 generator in the minimal model of $S^4$ as $h_7$, in line with the notation used throughout this section and in Sec. \ref{TheMechanism}.
This generator was called $\omega_7$ in Ex. \ref{MinimalDgcAlgebraModelFor4Sphere}.

\item  {\bf (ii)} In interpreting expression \eqref{dgcModelForHomotopyQuotientOf4SphereByCircleAction}
it may be worthwhile to view passing to rational homotopy theory as 
a homotopical analogue of forming first derivatives: that the A-type action on the 4-sphere is rationally trivial 
is analogous to finding that the derivative of some non-trivial function on the real line vanishes at the origin.
That is to say, the A-type orbispace does \emph{not} itself split as a product, but it does in the rational approximation.
This turns out to be crucial for our gauge enhancement mechanism .
In the companion article \cite{ADE} we go further and work in \emph{equivariant} rational homotopy theory, which captures a great deal more information.
%
\end{remark}

In summary, we have the following:
\begin{prop}[Minimal DG-algebra model for the A-type orbispace]
  \label{MinimalDGCAlgebraModelForATypeOrbispace}
  The minimal DG-algebra model of the A-type orbispace of the
  4-sphere (Def. \ref{ATypeOrbispaceOf4Sphere}) is
  $$
    \mathcal{O}\left(
      S^4 \dslash S^1
    \right)
    \;\simeq\;
    \mathbb{Q}[ \omega_2, \omega_4, h_7 ]\bigg/
    \left(
      \begin{aligned}
        d \omega_2 & = 0
        \\[-1mm]
        d \omega_4 & = 0
        \\[-1mm]
        d h_7
          & =
        -\tfrac{1}{2} \omega_4 \wedge \omega_4
      \end{aligned}
    \right)
    \;\in\;
    \mathrm{DGCAlg}
    \,.
  $$
  Furthermore, the unit of the $\mathrm{Ext}\dashv\mathrm{Cyc}$-adjunction (Theorem \ref{GCycExtAdjunction})
  on the A-type orbispace, composed with the equivalence \eqref{CanonicalCocycle} from Prop. \ref{ExtensionOfHomotopyQuotientEquivalentToOriginalSpace},
  pulls back the generators of the DG-algebra model for $\mathrm{Cyc}(S^4)$ (Ex. \ref{MinimalDGCAlgebraModelForCyclicSpaceOfFourSphere}) as follows:
  \begin{equation}
    \label{AdjunctionUnitOnRationalATypeOrbispace}
    \raisebox{40pt}{
    \xymatrix@R=-2pt@C=4em{
      S^4 \dslash S^1
        \ar[r]^-{\eta_{S^4 \dslash  S^1}}
      &
      \mathrm{Cyc}\,\mathrm{Ext}(S^4 \dslash S^1)
      \ar[r]^-{ \mathrm{Cyc}(\kappa) }_-{\simeq_{\mathrm{whe}}}
      &
      \mathrm{Cyc}(S^4)
      \\
      \omega_2 && \omega_2 \vphantom{h_7}\ar@{|->}[ll]
      \\
      \omega_4 && \omega_4\vphantom{h_7} \ar@{|->}[ll]
      \\
      0
      && \omega_6 \vphantom{h_7}\ar@{|->}[ll]
      \\
      0 && h_3 \vphantom{h_7}\ar@{|->}[ll]
      \\
      h_7 && h_7 \ar@{|->}[ll]
    }
    }
  \end{equation}
\end{prop}

\begin{proof}
By Lemma \ref{ATypeActionIsRationallyTrivial}, the minimal DG-algebra model for A-type homotopy quotient $S^4 \dslash S^1$ 
coincides with that of the trivial action, given by Ex. \ref{dgcAlgebraModelForHomotopyQuotientOfTrivialAction},
hence is given by setting $c_1= c_2= 0$ in  \eqref{dgcModelForHomotopyQuotientOf4SphereByCircleAction}.
We conclude with Lemma \ref{ExtCycAdjunctionForATypeOrbispace}.
\end{proof}

This concludes our discussion of the rational homotopy type of the A-type orbispace of the 4-sphere.
Next, we discuss the fiberwise stabilization of $S^4 \dslash S^1 \to S^4/S^1= S^3$ in rational homotopy theory.

\subsection{Fiberwise stabilized $\mathrm{Ext}$/$\mathrm{Cyc}$-unit of the A-type orbispace }
  \label{RationalUnitOnAType}

In Prop. \ref{MinimalDGCAlgebraModelForATypeOrbispace} we described the rationalization of the unit
$\eta_{S^4 \dslash S^1}$ of the $\mathrm{Ext}$/$\mathrm{Cyc}$-adjunction (Theorem \ref{GCycExtAdjunction})
on the A-type orbispace of the 4-sphere (Def. \ref{ATypeOrbispaceOf4Sphere}).
In the rational approximation, we may regard this map as lying over
the classifying space $B^2 S^1 \simeq B^3 \mathbb{Z} \simeq_{\mathbb{Q}} B^3 \mathbb{Q} \simeq_{\mathbb{Q}} S^3$
and discuss its \emph{fiberwise stabilization} (Prop. \ref{AdjunctionStabilization}) in
 rational parametrized stable homotopy theory (Theorem \ref{RationalParameterizedSpectradgModel}).

\medskip
The main result of this section is Theorem \ref{TwistedKTheoryInsideFiberwiseStabilizationOfATypeOrbispaceOf4Sphere},
which states that the fiberwise stabilization of the A-type orbispace contains two summands of twisted connective K-theory,
and characterizes lifts through the fiberwise stabilization adjunction unit in terms of lifting from 6-truncated to untruncated
twisted K-theory (cf. Ex. \ref{CyclificationOf4SphereReceives6TruncationOfTwistedK}). We interpret these lifts
as being the rational image of gauge enhancement of M-branes in  Sec. \ref{TheMechanism} below.

\medskip
The next result appears in \cite{RS}. We spell out the proof in some detail, since we will need certain details in the proof of
our main Theorem \ref{TwistedKTheoryInsideFiberwiseStabilizationOfATypeOrbispaceOf4Sphere}.
In order to make certain features more apparent, we use a different naming convention for algebra generators
than is used in \cite{RS}:
$$
  \begin{tabular}{c||ccccccc}
    \cite{{RS}}:
    &
    $a$
    &
    1
   &
    $c_{2n}$
    &
    $c_{2n+1}$
    &
    $e$
    &
    $\gamma_{2n}$
    &
    $\gamma_{2n+1}$
   \\
    \hline
    here:
    &
    $h_3$
    &
    $\widetilde\omega_0\vphantom{\Big(}$
    &
    $\omega^L_{2n+2}$
    &
    $\omega^R_{2n+4}$
    &
    $\omega_2$
    &
    $\widetilde\omega_{2n}$
    &
    $\omega_{2n}$
  \end{tabular}
$$

\begin{prop}[Minimal DG-module for fiberwise stabilization of the A-type orbispace]
  \label{MinimalDGModels}
  We have the following table of minimal DG-module models (Def. \ref{MinimalDGModule}) describing fiberwise stabilizations of the spaces \eqref{QuotientOfS4ByS1OverS3} over the 3-sphere:
%
  \vspace{.3cm}
  \begin{center}
  \begin{tabular}{|c||c|c|}
  \hline
    {\bf Fibration}
      &
    {\begin{tabular}{c} Vector space underlying \\ minimal DG-model \end{tabular}}
      &
    {\begin{tabular}{c} Differential of \\ minimal DG-model \end{tabular}}
    \\
    \hline
    \hline
    $
      \raisebox{20pt}{
      \xymatrix{
        S^0 = \left( S^4\right)^{S^1}
        \ar[d]
        \\
        S^3
      }}
    $
    &
    $
    \mathbb{Q}[h_3]
      \otimes
    \left\langle
      \omega^{L}_{2 p}, \omega^R_{2 p }
      \,\vert\,
      p \in \mathbb{N}
    \right\rangle
    $
    &
    $
      d
      \;:\;
      \left\{
        \begin{aligned}
          \omega^L_0 & \mapsto 0 &&  \multirow{2}{*}{ \bigg\}\,\footnotesize  $(\mathrm{ku} \dslash BS^1)$ }
          \\
          \omega^L_{2p + 2} & \mapsto h_3 \otimes \omega^L_{2 p }
          \\
          \omega^R_0 & \mapsto 0 && \multirow{2}{*}{ \bigg\}\,\footnotesize $(\mathrm{ku} \dslash BS^1)$ }
          \\
          \omega^R_{2p + 2} & \mapsto h_3  \otimes \omega^R_{2p}
        \end{aligned}
      \right.
    $
      \\
    \hline
    $
      \raisebox{20pt}{
      \xymatrix{
        S^4
        \ar[d]
        \\
        S^3
      }}
    $
    &
    $
    \mathbb{Q}[h_3]
      \otimes
    \left\langle
      \widetilde\omega_{2 p}, \omega_{2 p+ 4}
      \,\vert\,
      p \in \mathbb{N}
    \right\rangle
    $
    &
    $
      d
      \;:\;
      \left\{
        \begin{aligned}
          \widetilde\omega_0 & \mapsto 0 && \multirow{2}{*}{ \bigg\}\,\footnotesize $(\mathrm{ku} \dslash BS^1)$ }
          \\
          \widetilde\omega_{2p + 2} & \mapsto h_3 \otimes \widetilde\omega_{2 p }
          \\
          \omega_4 & \mapsto 0 && \multirow{2}{*}{  \bigg\}\,\footnotesize $(\Sigma^4 \mathrm{ku} \dslash  BS^1)$ }
          \\
          \omega_{2p + 6} & \mapsto h_3  \otimes \omega_{2p+4}
        \end{aligned}
      \right.
    $
      \\
    \hline
    $
      \raisebox{20pt}{
      \xymatrix{
        S^4\dslash S^1
        \ar[d]
        \\
        S^3
      }}
    $
    &
    $
    \mathbb{Q}[h_3, \omega_2]
      \otimes
    \left\langle
      \widetilde\omega_{2 p}, \omega_{2 p + 4}
      \,\vert\,
      p \in \mathbb{N}
    \right\rangle
    $
    &
    $
      d
      \;:\;
      \left\{
        \begin{aligned}
          \widetilde\omega_0 & \mapsto 0 && \multirow{2}{*}{\bigg\}\,\footnotesize $(\mathrm{ku} \dslash BS^1)$ }
          \\
          \widetilde\omega_{2p + 2} & \mapsto h_3 \otimes \widetilde\omega_{2 p }
          \\
          \omega_2 & \mapsto 0 && \multirow{3}{*}{\Bigg\}\, \footnotesize  $(\Sigma^2 \mathrm{ku} \dslash  BS^1)$ }
          \\
          \omega_4 & \mapsto h_3 \wedge \omega_2 \otimes \widetilde \omega_0
          \\
          \omega_{2p + 6} & \mapsto h_3  \otimes \omega_{2p + 4}
        \end{aligned}
      \right.
    $
    \\ \hline
  \end{tabular}
  \end{center}
\end{prop}

\vspace{1mm}
\noindent Beware of the special placement of the element $\omega_2$ in the last line, as a generator of
(the graded vector space underlying) a whole graded-commutative algebra.
On the far right of the table we are highlighting that there are two sequences of differentials in each case, differing
only in the degrees in which they start, and that each are of the same form as those for the minimal DG-model of the labelled shifted
twisted K-theory spectrum (Lemma \ref{TwistedKModel}). Observe, moreover, that the second sequence in the last line really starts
with the element $\omega_2 \otimes \widetilde \omega_0$.

\begin{proof}
To determine the minimal DG-models, in each case we use that:
 \begin{itemize}
   \item by the particular form of the minimal DG-algebra model of $S^3$ (Ex. \ref{MinimalDgcAlgebraModelFor3Sphere}),
   the minimal DG-modules in question
   are necessarily \emph{free} modules over $\mathbb{Q}[h_3]$; and
   \item
   according to \eqref{FormulaFiberwiseSuspensionSpectra},
 the cochain cohomology of the minimal DG-module must coincide with the rational cohomology of the total space of the corresponding fibration.
 \end{itemize}
 In all present cases of interest, these two constraints have a unique solution.

 We begin by adjoining additional closed generators $\omega_k$ to capture the cohomology of the total space.
 But by the free module structure, this also makes the element $h_3 \otimes \omega_k$ appear, which must be killed off in cohomology.
 To do this, we introduce new generators with prescribed differential, which produce new spurious elements that need to be killed off, and so on.
Explicitly, we have:
 \item {\bf (i)} We start with the case $S^0 \to S^3$.  Write the 0-sphere as the
 disjoint union of a \lq\lq left'' and a \lq\lq right'' point
  $$
    S^0 = \ast^L \coprod \ast^R
    \,.
  $$
  In cohomology, the map $S^0 \xrightarrow{\;\;\pi\;\;} S^3$ is
  $$
    \xymatrix@C=3pt{
      H^\bullet(S^0, \mathbb{Q})
      = \left\langle [\omega^L_0], [\omega^R_0] \right\rangle
      \ar@{<-}[d]^{\pi^\ast}
      &&&
      [\omega_0^L] + [\omega_1^R]
      &&
      0
      \\
      H^\bullet(S^3, \mathbb{Q})
       = \left\langle [1], [h_3] \right\rangle
      &&&
      [1]
      \ar@{|->}[u]
      &&
      [h_3].
      \ar@{|->}[u]
    }
  $$
  Thus, if the DG-model for the fiberwise stabilization of $S^0$ over $S^3$ is to be of the form
  $$
    \xymatrix@C=3pt@R=1.6em{
      \left\langle 1, h_3\right\rangle
      \otimes
      \left\langle
        \omega_0^L, \omega_0^R, \cdots
      \right\rangle
      &&&
      1 \otimes ( \omega^L_0 + \omega^R_0 )
      &&
      h_3 \otimes ( \omega^L_0 + \omega^R_0 )
      \\
      \underset{
        =\mathbb{Q}[h_3]
      }{
      \underbrace{
        \left\langle 1, h_3\right\rangle
      }}
      \ar[u]_{\pi^\ast}
      &&&
      1
      \ar@{|->}[u]
      &&
      h_3
      \ar@{|->}[u]
    }
  $$
  with
  $$
    d 1 = 0,
 \qquad
     d h_3 = 0 ,
     \qquad  d \omega_0^{L/R} = 0
    \,,
  $$
  then arguing by induction proves the claim;
 firstly, there must be additional generators
  $\omega_2^{L/R}$ in order to remove the elements $h_3 \otimes \omega_0^{L/R}$ from cohomology:
  $$
    d \omega_2^{L/R} = h_3 \otimes \omega_0^{L/R}
    \,.
  $$
  But this means that the $h_3 \otimes \omega_2^{L/R}$ are closed, since
  $$
    \begin{aligned}
      d\big(
        h_3 \otimes \omega_2^{L/R}
      \big)
      & =
      \underset{
        = 0
      }{
      \underbrace{
        (d h_3)
      }} \otimes \omega^{L/R}_2
      -
      h_3 \otimes d \omega^{L/R}_2
      \\
      & =
      -
      \underset{
        = 0
      }{
      \underbrace{
        h_3 \wedge h_3
      }}
      \otimes \omega^{L/R}_0\;.
    \end{aligned}
  $$
  In order to remove these elements from cohomology, we need to introduce new elements $\omega_4^{L/R}$
  with differential
  $$
    d \omega^{L/R}_4 = h_3 \otimes \omega^{L/R}_2
    \,,
  $$
  and so on.

 \item {\bf (ii)} We now consider the case $S^4 \xrightarrow{\;\;\pi\;\;} S^3$, which in cohomology is the assignment
  $$
    \xymatrix@C=3pt{
      H^\bullet(S^4, \mathbb{Q})
      = \left\langle [\widetilde\omega_0], [\omega_4] \right\rangle
      \ar@{<-}[d]^{\pi^\ast}
      &&&
      [\widetilde\omega_0]
      &&
      0
      \\
      H^\bullet(S^3, \mathbb{Q})
      = \left\langle [1], [h_3] \right\rangle
      &&&
      [1]
      \ar@{|->}[u]
      &&
      [h_3].
      \ar@{|->}[u]
    }
  $$
  Hence, if the minimal DG-model is to be of the form
  $$
    \xymatrix@C=3pt@R=1.5em{
      \left\langle 1, h_3\right\rangle
      \otimes
      \left\langle
        \widetilde\omega_0, \omega_4, \cdots
      \right\rangle
      &&&
      1 \otimes \widetilde\omega_0
      &&
      h_3 \otimes \widetilde\omega_0
      \\
        \left\langle 1, h_3\right\rangle
      \ar[u]_{\pi^\ast}
      &&&
      1
      \ar@{|->}[u]
      &&
      h_3
      \ar@{|->}[u]
    }
  $$
  with
  $$
    d 1 = 0,
    \qquad  d \widetilde\omega_0 = 0,
    \qquad
    d \omega_4 = 0\;,
  $$
 then there need to be elements
  $\widetilde\omega_2$ and $\omega_6$ that remove $h_3 \otimes \widetilde\omega_0$ and $h_3 \otimes \omega_4$ from cohomology via
  $$
    d \widetilde\omega_2 = h_3 \otimes \widetilde\omega_0
    \,,
    \qquad\quad
    d \omega_6 = h_3 \otimes \omega_4
    \,.
  $$
  But this implies that
  $$
    d( h_3 \otimes \widetilde\omega_2 ) = 0
    \,,
    \qquad\quad
    d( h_3 \otimes \omega_6 ) = 0,
  $$
  so that to we must introduce additional generators $\widetilde\omega_4$ and $\omega_8$ to make these elements exact.
  But then $h_3\otimes \widetilde\omega_4$ and $h_3\otimes \omega_8$ are closed, so that we must add $\widetilde\omega_6$ and $\omega_{10}$ to remove them from cohomology.
  The result follows by induction.

\item {\bf (iii)} Now consider the main case of interest, the A-type orbispace $S^4 \dslash S^1$.
The cohomology of the total space was determined in Lemma \ref{HomotopyGroupsOfS4OverS1}, so that in cohomology the map $S^4\dslash S^1 \to S^3$ is of the form
  $$
    \xymatrix@C=3pt{
      H^\bullet(S^4 \dslash S^1 , \mathbb{Q})
      =
      \langle [\widetilde \omega_0], [\omega_2], [\omega_4], [\omega_2 \wedge \omega_2], \cdots \rangle
      \ar@{<-}[d]^{\pi^\ast}
      &&&
      [\widetilde\omega_0]
      &&
      0
      \\
      H^\bullet(S^3, \mathbb{Q})
      = \left\langle [1], [h_3] \right\rangle
      &&&
      [1]
      \ar@{|->}[u]
      &&
      [h_3]
      \ar@{|->}[u]
    }
  $$
  First consider $\widetilde \omega_0$ and of $\omega_2(\equiv \omega_2\otimes \widetilde\omega_0)$: analogously to {\bf (ii)} above, these induce sequences of additional generators $\widetilde \omega_{2p}$ and $\omega_{2p+4}$
  with differentials as claimed.
  But adding in these generators already implies the existence of a further cohomology class in degree four, exhibited by the cocycle
  \begin{equation}
    \label{SecretDegree4Class}
    \widehat \omega_4 := 1 \otimes \omega_4 - \omega_2 \otimes \widetilde \omega_2
  \end{equation}
By degree reasons, the only potential primitives of $\widehat{\omega}_4$ are non-zero multiples of $h_3$, but this is closed.
Thus we have already found the correct cohomology in dimensions $\leq 5$.
 Extending in powers of $\omega_2$, we find that the cocycles
 \[
 \omega_2^{\wedge(k+2)}\otimes \widetilde\omega_0\;\;\mbox{ and }\;\;\omega_2^{\wedge k}\otimes \widehat{\omega}_4
 \]
 recover the correct $2k+4$-dimensional cohomology of $S^4\dslash S^1$, and no new non-trivial cocycles arise in odd degrees.
 This completes the proof.
\end{proof}

At this point, let us pause to provide some  intuition for what is going on in Prop. \ref{MinimalDGModels}:

\begin{remark}[Interpretation of fiberwise stabilization of A-type orbispace]
  \label{InterpretationOfFiberwiseSuspensionOfATypeOrbispace}
  Due to the rational homotopy equivalence
  $
    S^3 \simeq_{\mathrm{whe},\mathbb{Q}} B^2 S^1
  $,
  the homotopy fiber of any point inclusion $\ast \to S^3$ is, rationally, the classifying space $B S^1$ \eqref{ClassifyingSpace}:
  $$
    \mathrm{hofib}\big(
      \xymatrix{
        \ast
        \ar[r]
        &
        S^3
      }     \big)
    \;\simeq_{\mathrm{whe},\mathbb{Q}}\;
    B S^1
    \,.
  $$
  Accordingly, the homotopy fiber of $S^0 \to S^3$ is, rationally, the disjoint union of
  two copies of $B S^1$:
  $$
    \mathrm{hofib}\Big( \!
      \xymatrix{
        \big(S^4\big)^{S^1} = \ast\displaystyle\coprod\ast
        \ar[r]
        &
        S^3
      }   \!  \Big)
    \;\simeq_{\mathrm{whe},\mathbb{Q}}\;
    B S^1 \coprod B S^1
    \,.
  $$
Forming fiberwise suspension spectra (Prop. \ref{AdjunctionStabilization}) really means that we stabilize these homotopy fibers, which means that the fiberwise suspension spectrum of $S^0 \to S^3$
  is a parametrized spectrum whose fiber over any point is
  $$
    \Sigma^\infty_+ B S^1 \oplus \Sigma^\infty_+ B S^1
    \;\simeq_{\mathrm{swhe}, \mathbb{Q}}\;
    \mathrm{ku} \oplus \mathrm{ku}
  $$
(cf. Ex. \ref{RationalSnaithTheorem}).
  Hence, the fiberwise suspension spectrum $\Sigma^\infty_{+, S^3} S^0$ is rationally equivalent to the direct sum
  of two copies of twisted connective K-theory,
  which by comparison in Lemma \ref{TwistedKModel} is precisely what item
  {\bf  (i)} in Prop. \ref{MinimalDGModels} asserts.

  The second item in Prop. \ref{MinimalDGModels} can also be heuristically understood in similar terms: as the cartoon picture
  \begin{center}
   \includegraphics[width=.4\textwidth]{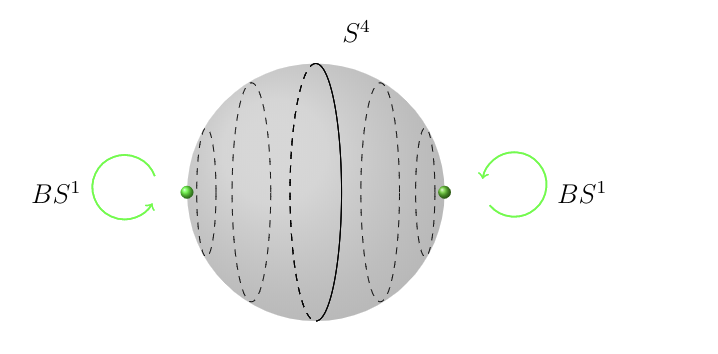}
  \end{center}
  for the
  suspended Hopf action from Def. \ref{ATypeOrbispaceOf4Sphere} indicates, the homotopy fiber now is
  some mixture of the two copies of $B S^1$ attached to the fixed points, and a
  copy of $S^1$ attached to all the other points. Prop. \ref{dgModelForS0Inclusion} below shows that, accordingly,
  there is one unshifted copy of $\Sigma^\infty_{+} B S^1$ in the fiber spectrum of $\Sigma^\infty_{+,S^3} S^4$
  that pulls back diagonally into the direct sum of the suspension spectra associated with the two fixed points.
\end{remark}

We record the following useful consequences of Prop. \ref{MinimalDGModels}:

\begin{lemma}[Comparison map between DG-module models of A-type orbispace]
  \label{ComparisonMapBetweenModelsForFiberwiseStabilizationOfATypeOrbispace}
  A quasi-isomorphism of DG-modules over $\mathbb{Q}[h_3]$ from the minimal
  DG-model of $\Sigma^\infty_{+,S^3}\left( S^4 \dslash S^1 \right)$ (Prop. \ref{MinimalDGModels})
  to the (non-minimal) DG-model underlying the DG-algebra model of Prop. \ref{MinimalDGCAlgebraModelForATypeOrbispace} 
  is determined on generators as follows:
  \begin{gather}
    \label{ComparisonMapBetweenModels}
    \mathclap{\raisebox{70pt}{\xymatrix@R=-2pt{
      \mathbb{Q}[\omega_2,\omega_4, h_7]\Bigg/
      \left(
        {\begin{aligned}
          d \omega_2 & = 0
          \\
          d \omega_4 & = 0
          \\
          d h_7 & = - \tfrac{1}{2} \omega_4 \wedge \omega_4
        \end{aligned}}
      \right)
      &
      \mathbb{Q}[h_3, \omega_2] \otimes \langle \widetilde \omega_{2p}, \omega_{2p + 4} \rangle\Bigg/
      \left(
      {\begin{aligned}
        d \widetilde \omega_0 & = 0
        \\
        d \widetilde \omega_{2p + 2} & = h_3 \otimes \widetilde \omega_{2 p}
        \\
        d \omega_2 & = 0
        \\
        d \omega_4 & = h_3 \wedge \omega_2 \otimes \widetilde \omega_0
        \\
        d \omega_{2 p+ 6} & = h_3 \otimes \omega_{2p + 4}
      \end{aligned}}
      \right)
      \ar[l]_-{ \simeq_{\mathrm{qi}} }
      \\
      \\
      0
      \ar@{<-|}[r]
      &
      \mathrlap{h_3\vphantom{\widetilde \omega_{2p+2}}}
      \\
      1
      \ar@{<-|}[r]
      &
      \mathrlap{\widetilde \omega_0\vphantom{\widetilde \omega_{2p+2}}}
      \\
      0
      \ar@{<-|}[r]
      &
      \mathrlap{\widetilde \omega_{2p+2}}
      \\
      \omega_2
      \ar@{<-|}[r]
      &
      \mathrlap{\omega_2\vphantom{\widetilde \omega_{2p+2}}}
      \\
      \omega_4
      \ar@{<-|}[r]
      &
      \mathrlap{\omega_4\vphantom{\widetilde \omega_{2p+2}}}
      \\
      0
      \ar@{<-|}[r]
      &
      \mathrlap{\omega_{2p+6}\vphantom{\widetilde \omega_{2p+2}}}
    }}}
  \end{gather}
  where both wedge products as well as tensor products (including powers of $\omega_2$) on the right are sent to wedge products on the left.
\end{lemma}

\begin{proof}
  The only non-evident point to note is that this map is indeed an isomorphism on cohomology in degree 4.
  But this follows from \eqref{SecretDegree4Class}, since $\widehat{\omega}_4\mapsto \omega_4$.
%
\end{proof}

\begin{prop} [DG-model for $S^0$-$S^3$-$S^4$ system]
\label{dgModelForS0Inclusion}
The commutative diagram \eqref{QuotientOfS4ByS1OverS3}
$$
  \xymatrix@R=1em{
    S^0
    \ar[dr]
    \; \ar@{^{(}->}[rr]
    &&
    S^4
    \ar[dl]
    \\
    & S^3
  }
$$
is represented on rational fiberwise suspension spectra in terms of the DG-modules of Prop. \ref{MinimalDGModels}
by the commuting diagram
\begin{gather*}
\mathclap{
  \xymatrix@C=3.3em{
    \left(\hspace{-1mm}
      {\begin{array}{l}
        d \omega^{L/R}_0 = 0
        \\
        d \omega^{L/R}_{2p + 2} = h_3\otimes \omega^{L/R}_{2p}
      \end{array}}
      \hspace{-1mm}
    \right)
    &&&&
    \left(\hspace{-1mm}
      {\begin{array}{l}
        d \widetilde{\omega}_{0} = 0\,, d \widetilde{\omega}_{2p+2} = h_3 \otimes \widetilde{\omega}_{2p}
        \\
        d \omega_{4} = 0\,, d \omega_{2p+6} = h_3 \otimes \omega_{2p+4}
      \end{array}}
      \hspace{-1mm}
    \right)
    \ar[llll]_-{\tiny
      \left(\hspace{-2mm}
        {\begin{array}{l}
           \;\;\;\;\widetilde{\omega}_{2p} \mapsto (\omega^{L}_{2p} + \omega^R_{2p})
           \\
            \omega_{2p + 4} \mapsto 0
        \end{array}}
        \hspace{-2mm}
      \right)
    }
    \\
    &&
    \left(\hspace{-1mm}
      {\begin{array}{l}
        d 1 =0
        \\
        d h_3 = 0
      \end{array}}
      \hspace{-1mm}
    \right)
    \ar@{->}[urr]_{\hspace{1cm}\tiny
      \left( \hspace{-2mm}
        \begin{array}{l}
          \;\; 1 \mapsto 1 \otimes \widetilde{\omega}_0
          \\
          h_3 \mapsto h_3 \otimes \widetilde{\omega}_0
        \end{array}
       \hspace{-2mm}\right)
    }
    \ar@{->}[ull]^{\hspace{-2.4cm}\tiny
      \left( \hspace{-2mm}
        \begin{array}{l}
          \;\; 1 \mapsto 1 \otimes ( \omega^L_0 + \omega^R_0 )
          \\
          h_3 \mapsto h_3 \otimes ( \omega^L_0 + \omega^R_0 )
        \end{array}
       \hspace{-2mm}\right)
    }
  }\!\!.}
\end{gather*}
\end{prop}
\begin{proof}
  We know that the various cohomology generators are mapped as follows:
  $$
    \xymatrix@R=1.6em{
      1 \otimes ([\omega^L_0] + [\omega^R_0])
      &&
      1 \otimes [\widetilde{\omega}_0]
      \ar@{|->}[ll]
      \\
      &
      [1]
      \ar@{|->}[ur]
      \ar@{|->}[ul]
    }
  $$
  This implies the statement along the lines of the proof of Prop. \ref{MinimalDGModels}.
\end{proof}

So far, we have shown that copies of rational twisted K-theory appear as summands in the rational fiberwise stabilization
of the A-type orbispace.
Our next aim is to analyze how the $\mathrm{Ext}/\mathrm{Cyc}$-adjunction unit compares these copies to the copy of
rational 6-truncated twisted K-theory to be found inside the rational cyclification of the 4-sphere (Ex. \ref{CyclificationOf4SphereReceives6TruncationOfTwistedK}).
For this, we will need to know how the adjunciton unit is represented in terms of minimal DG-modules:

\begin{prop}[Fiberwise stabilization of $\mathrm{Ext}/\mathrm{Cyc}$-unit at A-type orbispace
via minimal DG-models]
  \label{FiberwiseStabilizationOfUnitOnMinimalModels}
  The fiberwise stabilization of the component of the $\mathrm{Ext}$/$\mathrm{Cyc}$-unit at the
  A-type orbispace of the 4-sphere
  $$
    \xymatrix{
      \Sigma^\infty_{+,S^3} \left( S^4\dslash S^1 \right)
      \ar[rr]^-{ \Sigma^\infty_{+,S^3}(\eta) }
      &&
      \Sigma^\infty_{+,S^3} \mathrm{Cyc}(S^4)
    }
  $$
  is represented on minimal DG-modules (Prop. \ref{MinimaldgModuleForFiberwiseStabilisationOfCyclicSpaceOf4Sphere}
  and Prop. \ref{MinimalDGModels}) by a homomorphism of DG-modules over $\mathbb{Q}[h_3]$,
  shown as a dotted map in \eqref{TheMap}, which has the following properties:

    \item \hypertarget{FirstThreeGeneratorsAreMappedCorrectly}{{\bf (i)}}
      The generators $\omega_2$, $\omega_4$ and $\omega_6$ from Prop. \ref{MinimaldgModuleForFiberwiseStabilisationOfCyclicSpaceOf4Sphere}
      are sent to the generators of the same name in Prop. \ref{MinimalDGModels}:
      \begin{equation}
      \label{AdjunctionUnitFiberwiseStabilizedOnATypeOrbispace}
        \raisebox{15pt}{
        \xymatrix@R=.5pt{
          \omega_2 \otimes \tilde \omega_0
          \ar@{<-|}[rr]^-{ \Sigma^\infty_{+,S^3}(\eta)^\ast }
          &&
          \omega_2
          \\
          1 \otimes \omega_4
          \ar@{<-|}[rr]^-{ }
          &&
          \omega_4
          \\
          1 \otimes \omega_6
          \ar@{<-|}[rr]^-{ }
          &&
          \omega_6
        }
        }
      \end{equation}
    \item \hypertarget{NotInImage}{{\bf (ii)}}
     The elements
     $$
       1 \otimes \omega_{2p + 8} \,,\;\; p\in \mathbb{N}
     $$
     are \emph{not} in its image.
   \item \hypertarget{RelativeCellInclusion}{{\bf (iii)}}
    It is a minimal relative DG-module inclusion (Def. \ref{MinimalDGModule}).
\end{prop}

\begin{proof}
 A general representative for $\Sigma^\infty_{+,S^3}(\eta)$ on minimal DG-models
 is a dotted morphism that makes the following diagram commute \emph{up to homotopy},
 but we claim that a representative with the stated properties exists that makes this diagram even  \emph{strictly} commutative:
  \begin{gather}
    \label{TheMap}
    \mathclap{
    \hspace{0cm}
    \scalebox{.85}{
    \xymatrix@R=2em@C=-2em{
      \Sigma^\infty_{+,S^3}\left( S^4 \dslash S^1 \right)
      \ar[r]^-{ \Sigma^\infty_{+,S^3}\left(\eta_{S^4 \dslash S^1}\right) }
      &
      \Sigma^\infty_{+,S^3} \mathrm{Cyc}\,\mathrm{Ext}\left( S^4  \dslash  S^1 \right)
      \ar[r]^-{\Sigma^\infty_{S^3} \mathrm{Cyc}(\kappa)}_-{\simeq_{\mathrm{whe}}}
      &
      \Sigma^\infty_{+,S^3} \mathrm{Cyc}(S^4)
      \\
      \mathbb{Q}[\omega_2,\omega_4, h_7]\Bigg/\!
      \left(
        {\begin{aligned}
          d \omega_2 & = 0
          \\
          d \omega_4 & = 0
          \\
          d h_7 & = - \tfrac{1}{2} \omega_4 \wedge \omega_4
        \end{aligned}}
      \right)
      &&
      \mathbb{Q}[ h_3, h_7, \omega_2, \omega_4, \omega_6 ]
      \Bigg/\!
      \left(
        {\begin{aligned}
          d h_3 & = 0
          \\
          d h_7 & = -\tfrac{1}{2} \omega_4 \wedge \omega_4 + \omega_2 \wedge \omega_6
          \\
          d \omega_2 & = 0
          \\
          d \omega_4 & = h_3 \wedge \omega_2
          \\
          d \omega_6 & = h_3 \wedge \omega_4
        \end{aligned}}
      \right)
      \ar[ll]|-{\footnotesize
        \begin{aligned}
          h_3 & \mapsto 0
          \\
          h_7 & \mapsto h_7
          \\
          \omega_2 & \mapsto \omega_2
          \\
          \omega_4 & \mapsto \omega_4
          \\
          \omega_6 & \mapsto 0
        \end{aligned}
      }
      \\
      \\
      \\
      \\
      \mathbb{Q}[h_3, \omega_2] \otimes \langle \widetilde \omega_{2p}, \omega_{2p + 4} \rangle\Bigg/\!
      \left(
      {\begin{aligned}
        d h_3 & = 0
        \\
        d \widetilde \omega_0 & = 0
        \\
        d \widetilde \omega_{2p + 2} & = h_3 \otimes \widetilde \omega_{2 p}
        \\
        d \omega_2 & = 0
        \\
        d \omega_4 & = h_3 \wedge \omega_2 \otimes \widetilde \omega_0
        \\
        d \omega_{2 p+ 6} & = h_3 \otimes \omega_{2p + 4}
      \end{aligned}}
      \right)
      \ar@{-> }[uuuu]^-{\simeq_{\mathrm{qi}}}_-{\footnotesize
        \begin{aligned}
          h_3 & \mapsto 0
          \\
          \widetilde \omega_0 & \mapsto 1
          \\
          \widetilde \omega_{2p + 2} & \mapsto 0
          \\
          \omega_2 & \mapsto \omega_2
          \\
          \omega_4 & \mapsto \omega_4
          \\
          \omega_{2p+6} & \mapsto 0
        \end{aligned}
      }
       &&
      \frac{
        \mbox{$\mathbb{Q}[ h_3, \omega_2, \omega_4, \omega_6 ]$}
      }{
        \mbox{$\left( \omega_2 \wedge \omega_6 -\tfrac{1}{2}\omega_4 \wedge \omega_4 \right)$}
      }\Bigg/\!
      \left(
        {\begin{aligned}
          d h_3 & = 0
          \\
          d \omega_2 & = 0
          \\
          d \omega_4 & = h_3 \wedge \omega_2
          \\
          d \omega_6 & = h_3 \wedge \omega_4
        \end{aligned}}
      \right).
      \ar@{-> }[uuuu]^-{\simeq_{\mathrm{qi}}}_-{\footnotesize
        \begin{aligned}
          h_3 & \mapsto h_3
          \\
          \omega_2 & \mapsto \omega_2
          \\
          \omega_4 & \mapsto \omega_4
          \\
          \omega_6 & \mapsto \omega_6
        \end{aligned}
      }
      \ar@{..>}[ll]|-{\footnotesize
        \color{gray}
        \begin{aligned}
          h_3 & \mapsto h_3
          \\
          \omega_2 & \mapsto \omega_2 \otimes \widetilde \omega_0
          \\
          \omega_4 & \mapsto 1 \otimes \omega_4
          \\
          \omega_6 & \mapsto 1 \otimes \omega_6
          \\
          &\quad \vdots
        \end{aligned}
      }
    }
    }}
  \end{gather}
  Here the top horizontal morphism is the map from Prop. \ref{MinimalDGCAlgebraModelForATypeOrbispace}, regarded as a map between underlying DG-modules, while in the bottom row we have the corresponding minimal models.
  The left-hand vertical morphism is \eqref{ComparisonMapBetweenModels}
  from Lemma \ref{ComparisonMapBetweenModelsForFiberwiseStabilizationOfATypeOrbispace}, while the right-hand vertical morphism is provided by Prop. \ref{MinimaldgModuleForFiberwiseStabilisationOfCyclicSpaceOf4Sphere}.
%
%

  Firstly, we show that if such a dotted morphism exists, then it must satisfy
  (\hyperlink{FirstThreeGeneratorsAreMappedCorrectly}{ i}):
  To start with, the dotted morphism necessarily sends $h_3 \mapsto h_3$, since it is a homomorphism
   of free $\mathbb{Q}[h_3]$-modules.
  Next, in order for the underlying linear maps to commute, the dotted morphism needs to send
  $$
    \omega_2 \mapsto \omega_2 \otimes \widetilde \omega_0 + \cdots
    \,,
  $$
  where the ellipsis indicates a term of degree 2 that vanishes under the left vertical map. The only possibility for such a term is a scalar multiple
  of $1 \otimes \widetilde \omega_2$,
  but then the respect for the differential
  $$
    \xymatrix@R=1.5em{
      c \, h_3 \otimes \widetilde \omega_0
      &&
      0
      \ar@{|->}[ll]
      \\
       \omega_2 \otimes \widetilde \omega_0
        +  c \, 1 \otimes \widetilde \omega_2
      \ar[u]^-{d}
      &&
      \omega_2
      \ar@{|->}[u]_-d
      \ar@{|->}[ll]
    }
  $$
  enforces $c = 0$.

  A similar argument applies to $\omega_4$: we must have
  $$
    \omega_4 \mapsto 1 \otimes \omega_4 + \cdots
  $$
  where the ellipsis indicates a term of degree 4 that vanishes under the left vertical map, which must be of the form
  $c_1 \, \omega_2 \wedge \widetilde \omega_2 + c_2 \, 1 \otimes \widetilde \omega_4$ for some $c_1, c_2 \in \mathbb{R}$. But then respect for the differentials
  $$
    \xymatrix{
         h_3 \wedge \omega_2 \otimes \widetilde \omega_0
        +  c_1 \, \omega_2 \wedge h_3 \otimes \widetilde \omega_0
        + c_2 \, h_3 \otimes \widetilde \omega_2
      &&
      h_3 \wedge \omega_2
      \ar@{|->}[ll]
      \\
 1 \otimes \omega_4
        +  c_1 \, \omega_2 \otimes \widetilde \omega_2
        +  c_2 \, 1 \otimes \widetilde \omega_4
      \ar@{|->}[u]^d
      &&
      \omega_4
      \ar@{|->}[u]_d
      \ar@{|->}[ll]
    }
  $$
  implies $c_1 = 0$ and $c_2 = 0$.

  Finally, we must have that
  $$
    \omega_6 \mapsto 0 + \cdots,
  $$
  where the ellipsis is a term of degree 6 that vanishes under the left vertical map. Hence
  $$
    \omega_6
      \mapsto
    c_1 \, 1 \otimes \omega_6
    +
    c_2 \, 1 \otimes \widetilde \omega_6
    +
    c_3 \, \omega_2 \otimes \widetilde \omega_4
    +
    c_4 \, \omega_2^{\wedge 2}\otimes \widetilde \omega_2
  $$
  for some $c_1, c_2, c_3, c_4 \in \mathbb{R}$.
  Now the respect for the differentials
  $$
    \xymatrix{
      {\begin{aligned}
        c_1 \, h_3 \otimes \omega_4
        +
        c_2 \, h_3 \otimes \widetilde \omega_4
        +
        c_3 \, \omega_2 \wedge h_3 \otimes \widetilde \omega_2
        +
        c_4 \, \omega_2^{ \wedge 2} \wedge h_3 \otimes \widetilde \omega_0
      \end{aligned}}
      &&
      h_3 \wedge \omega_4
      \ar@{|->}[ll]
      \\
      {\begin{aligned}
        c_1 \, 1 \otimes \omega_6
        +
        c_2 \, 1 \otimes \widetilde \omega_6
        +
        c_3 \, \omega_2 \otimes \widetilde \omega_4
        +
        c_4 \, \omega_2 \wedge \omega_2 \otimes \widetilde \omega_2
      \end{aligned}}
      \ar@{|->}[u]
      &&
      \omega_6
      \ar@{|->}[u]
      \ar@{|->}[ll]
    }
  $$
  implies $c_1 = 1$ and $c_2 = c_3 =c_4 = 0$.
  This establishes (\hyperlink{FirstThreeGeneratorsAreMappedCorrectly}{{i}}), provided that the dotted morphism actually exists.

  We now prove that a dotted morphism in \eqref{TheMap}
  does exist as claimed.
  Indeed, we have just seen that the images of the generators $h_3$, $\omega_2$, $\omega_4$, and $\omega_6$ are predetermined.
  Extending this map as a $\mathbb{Q}[h_3,\omega_2]$-module homomorphism,
  we have defined the dotted morphism on all terms that are at most linear in $\omega_4$ and $\omega_6$.
  We must now extend this map to higher order terms in these generators---but by the relation $\tfrac{1}{2}\omega_4 \wedge \omega_4 = \omega_2 \wedge \omega_6$, which holds in the minimal DG-module in the
  bottom right of \eqref{TheMap},
  we may consider those unique representatives which are at most linear in $\omega_4$.
  Hence we need only find consistent images for all elements of the form
  $$
    (\omega_6)^{\wedge n}, \,\;\;\;\; (\omega_6)^{\wedge(n-1)} \wedge \omega_4
    \phantom{AAAA}
    \mbox{ for $n \geq 2$ }
    \,.
  $$
  We claim that there are unique coefficients $a_n, b_n \in \mathbb{N} \subset \mathbb{Q}$ such that the assignment
  \begin{equation}
    \label{ExtComp}
    \raisebox{12pt}{
    \xymatrix@R=1pt{
      b_n \, \omega_2^{\wedge (n-1)} \otimes \omega_{4n}
      &&
      \mathrlap{(\omega_6)^{\wedge(n-1)} \wedge \omega_4}
      \ar@{|->}[ll]
      \\
      a_n \, \omega_2^{\wedge (n-1)} \otimes \omega_{4n+2}
      &&
      \mathrlap{(\omega_6)^{\wedge n} }
      \ar@{|->}[ll]
    }
    }
    \phantom{AAAAAAAAAAA}
    \mbox{for $n \geq 2$}.
  \end{equation}
  respects the differentials and makes the underlying linear maps of \eqref{TheMap} commute.
  We argue by induction: for the initial case $n=2$, respect for the differentials on $\omega_6 \wedge \omega_4$
means that the following square has to commute
\begin{equation}
  \label{FirstSquare}
  \raisebox{20pt}{\xymatrix{
    3 h_3 \wedge \omega_2 \otimes \omega_6
    &&
    3 h_3 \wedge \omega_2 \wedge \omega_6
    \ar@{|->}[ll]
    \\
    b_2 \, \omega_2 \otimes \omega_8
    \ar@{|->}[u]^d
    &&
    \omega_6 \wedge \omega_4
    \,,
    \ar@{|->}[ll]
    \ar@{|->}[u]_d
  }}
\end{equation}
where the top left entry shows the image of the top right element.
To see this, notice that:

\vspace{-2mm}
\begin{enumerate}[{\bf (a)}]
\item  in computing the differential on the right of \eqref{FirstSquare}, we are using the relation $\tfrac{1}{2} \omega_4 \wedge \omega_4 = \omega_2 \wedge \omega_6$
in the minimal DG-module model from Prop. \ref{MinimaldgModuleForFiberwiseStabilisationOfCyclicSpaceOf4Sphere},
to uniquely represent all differentials by representatives that are at most linear in $\omega_4$; and

\vspace{-1mm}
\item  the definition \eqref{ExtComp} applies only to terms at least quadratic in $\omega_6, \omega_4$.
For the image of the top right element in \eqref{FirstSquare} we instead use \eqref{AdjunctionUnitFiberwiseStabilizedOnATypeOrbispace} and respect for the $\mathbb{Q}[h_3,\omega_2]$-module structure.
\end{enumerate}

\vspace{-2mm}
\noindent But the differential of the bottom left element in \eqref{FirstSquare} is $b_2 \, \omega_2 \wedge h_3 \otimes \omega_6$,
and hence the square \eqref{FirstSquare} commutes
precisely if
$
  b_2 = 3
$.
In an analogous manner, we see that the respect for the differential on $\omega_6 \wedge \omega_6$ means that the
following square has to commute:
$$
  \xymatrix{
    2 b_2\,  h_3 \wedge \omega_2 \otimes \omega_8
    &&
    2 h_3 \wedge \omega_6 \wedge \omega_4
    \ar@{|->}[ll]
    \\
    a_2 \, \omega_2 \otimes \omega_{10}
    \ar@{|->}[u]^d
    &&
    \omega_6 \wedge \omega_6
    \,,
    \ar@{|->}[ll]
    \ar@{|->}[u]_d
  }
$$
where again the top left element shown is the image of the top right element, now obtained via \eqref{ExtComp}.
We see that the differential on the left has this same image precisely if
$$
  a_2 =  2 b_2  = 6
  \,.
$$
For the inductive argument, we assume that the claim
\eqref{ExtComp} holds for $(n-1)$, with $n \geq 3$. Then we need to assure the
commutativity of the squares
$$
  \xymatrix{
    n b_n \, h_3 \wedge \omega_2^{\wedge(n-1)} \otimes \omega_{4n}
    &&
    n\, h_3\wedge (\omega_6)^{\wedge(n-1)} \wedge \omega_4
    \ar@{|->}[ll]
    \\
    a_n \, \omega_2^{\wedge (n-1)} \otimes \omega_{4n+2}
    \ar@{|->}[u]^d
    &&
    (\omega_6)^{\wedge n}
    \ar@{|->}[u]_d
    \ar@{|->}[ll]
  }
$$
and
$$
  \xymatrix{
    \left( 2n-1\right) a_{n-1} \, h_3 \wedge \omega_2^{\wedge (n-2)}\otimes \omega_{4n-2}
    &&
    \left( 2n-1\right)\,  h_3 \wedge \omega_2 \wedge (\omega_6)^{\wedge(n-1)}
    \ar@{|->}[ll]
    \\
    b_n \, \omega_2^{\wedge(n-1)} \otimes \omega_{4n}
    \ar@{|->}[u]^d
    &&
    (\omega_6)^{\wedge(n-1)} \wedge \omega_4\,.
    \,
    \ar@{|->}[ll]
    \ar@{|->}[u]_d
  }
$$
By the inductive hypothesis, the second square above implies that $b_n =(2n-1) a_{n-1}$, whereas the first square implies
that $a_n = n\cdot b_n$. By induction, this gives a map of DG-$\mathbb{Q}[h_3, \omega_2]$-modules map as per the
dotted arrow in \eqref{TheMap} via the specification \eqref{ExtComp}.
Let us note for completeness that the coefficients in \eqref{ExtComp} are given recursively by
\[
\frac{a_n}{a_{n-1}} = 2n^2-n, \qquad\quad b_n = \frac{a_n}{n}\phantom{AAAAA}
    \mbox{for $n > 2$},
\]
with $a_2 =6$, $b_2 =3$.
This completes the construction of the dotted morphism in \eqref{TheMap}, and it is now easy to see that that diagram commutes.


By construction, it is clear that the elements
$1 \otimes \omega_{2p+8}$, $p\geq 0$ are not in the image of the dotted morphism.
This proves item (\hyperlink{NotInImage}{{ii}}).

Finally, the construction of the dotted morphism via (\hyperlink{FirstThreeGeneratorsAreMappedCorrectly}{{i}}) and \eqref{ExtComp}
is a map between minimal DG-modules (over $\mathbb{Q}[h_3]$) which is moreover manifestly an injective cell map in that is injectively maps generators to generators.
It follows that this map is a relative cell complex inclusion, in fact a  minimal DG-module inclusion.
This completes the proof.
\end{proof}

\begin{remark}
The proof of (\hyperlink{RelativeCellInclusion}{iii}) depends crucially upon the fact that the elements $\omega_2^{\wedge n}$
for varying $n$ are \emph{distinct} as generators for the minimal DG-modules appearing along the bottom of Diagram \eqref{TheMap}.
Despite being a potential source of great confusion, we persist in using this notation because of the utility of being able to
work \lq\lq multiplicatively in $\omega_2$'' in the argument above.
\end{remark}

We are now in a position to state and prove our main technical result:
\begin{theorem}
[K-theory \lq\lq detruncation'' by lifting against fiberwise stabilization of A-type orbispace]
  \label{TwistedKTheoryInsideFiberwiseStabilizationOfATypeOrbispaceOf4Sphere}
Regarding the A-type orbispace of the 4-sphere (Def. \ref{ATypeOrbispaceOf4Sphere}) as fibered over the 3-sphere 
  via the canonical map \eqref{QuotientOfS4ByS1OverS3} from Prop. \ref{SystemOfFixedPointsAndQuotientsOfATypeActionOn4Sphere},
   \begin{equation}
   \label{ATypeOrbispaceOverB3Q}
   \xymatrix{
     S^4 \dslash S^1 \ar[r]
     &
     S^3
     \ar[r]^-{\simeq_\mathbb{Q}}
     &
     B^3 \mathbb{Q}\,,
     }
      \end{equation}
    we have that:

 \item {\bf (i)} Rationally, the fiberwise stabilization (Prop. \ref{AdjunctionStabilization}) of \eqref{ATypeOrbispaceOverB3Q}
  contains a copy of the 2-fold suspension of twisted connective K-theory as a direct summand:
  $$
    \xymatrix{
      \Sigma^2\, \mathrm{ku} \dslash B S^1\;
      \ar@{^{(}->}[rr]^-{\iota}
      &&
      \Sigma^\infty_{+, S^3 }\left( S^4 \dslash S^1  \right)
    }.
  $$

 \item {\bf (ii)}
  A map $\phi$ of rational homotopy types over $B^3 \mathbb{Q}$
  \[
  \xymatrix@R=1em{
    X\ar[rr]^-{\phi} \ar[dr]_-{\mu_3} &&\mathrm{Cyc}(S^4) \ar[dl] \ar[r]^-{p}& \Omega^{\infty-2}_{B^3 \mathbb{Q}}
    \left( \mathrm{ku} \dslash B S^1\right)\langle 6\rangle
    \ar@/^0.5pc/[dll]
    \\
    & B^3 \mathbb{Q} &
    }
  \]
lifts from rational 6-truncated twisted K-theory 
 to rational untruncated twisted K-theory if and only if $\phi$ admits a lift through the fiberwise stabilization of the $\mathrm{Ext}/\mathrm{Cyc}$-unit $\eta_{S^4\dslash S^1}$:
%
  \[
  \mathclap{
    \xymatrix@R=1.6em@C=2.5em{
      &&
      \Omega^{\infty-2}_{B^3 \mathbb{Q}}\left( \mathrm{ku} \dslash B S^1\right)\;
      \ar@{->}[d]^-{\;\tau_6}
      \ar@{^{(}->}[rr]^-{\iota}
      &
      &
      \Omega^\infty_{S^3}
      \Sigma^\infty_{+,S^3}
      (S^4\dslash S^1)
      \ar[dd]^{ \Omega^\infty_{S^3} \Sigma^\infty_{+,S^3}\big( \eta_{{S^4 \dslash S^1}} \! \big) }
      \\
      &&
      \Omega^{\infty-2}_{B^3 \mathbb{Q}}\left( \mathrm{ku} \dslash B S^1\right)\langle 6\rangle
      \\
      X
      \ar[dr]_{\mu_3}
      \ar@{-->}@/^1pc/[rruu]|{\;\; \widehat{\phi} \;\;}
      \ar@{..>}@/^6.3pc/[rrrruu]|{ \;\;\widehat{ \mathrm{st}\circ \phi } \;\;}
      \ar[rr]^{\;\;\;\phi}
      &&
      \mathrm{Cyc}(S^4)
      \ar[u]_-{\; p}
      \ar[r]^-{ \mathrm{st} }
      \ar[dl]
      &
      \Omega^\infty_{S^3}
      \Sigma^\infty_{+,S^3}
      \mathrm{Cyc}(S^4)
      \ar[r]^-{\simeq}
      &
      \Omega^\infty_{S^3}
      \Sigma^\infty_{+,S^3}
      \mathrm{Cyc}\,\mathrm{Ext}(S^4\dslash S^1)
      \\
      & B^3 \mathbb{Q}
    }}
  \]
  That is, specifying a dotted lift $\widehat{\phi}$ is equivalent to specifying a dotted lift $\widehat{\mathrm{st}\circ \phi}$, where $\mathrm{st}$ is the unit of the fiberwise stabilization adjunction.
\end{theorem}

\begin{proof}
The rational parametrized spectra appearing on the right-hand side of the diagram have minimal DG-module models as in Prop. \ref{MinimalDGCAlgebraModelForCyclicSpaceOfFourSphere} and Prop. \ref{MinimalDGModels}.
On these minimal models, the vertical map on the right-hand side is represented by the dotted morphism in \eqref{TheMap}.
The fiberwise stabilization map $\mathrm{st}$ is represented in algebra by the right-hand vertical map in \eqref{TheMap}.
The claim {\bf (i)} was already established in Prop. \ref{MinimalDGModels}.
%

To prove {\bf (ii)}, we argue as follows.
Recall that by Prop. \ref{RationalInfiniteLoopSpace}, rational models for fiberwise infinite loop spaces are obtained by taking free relative algebras, which allows us to argue on module generators (equivalently, we argue on $\Sigma^\infty_{+, S^3}\dashv \Omega^\infty_{S^3}$-adjuncts).
By item (\hyperlink{FirstThreeGeneratorsAreMappedCorrectly}{i}) of Prop. \ref{FiberwiseStabilizationOfUnitOnMinimalModels},  a strict lifting of $\mathrm{st}\circ \phi$ through the stabilized $\mathrm{Ext}/\mathrm{Cyc}$-unit implies a strict lifting of $\phi$ through the the zig-zag morphism $\tau_6$---this is since the images of $\omega_2$, $\omega_4$ and $\omega_6$ are already specified in the rational model of $X$.
Since untruncated twisted K-theory includes as a direct summand via $\imath$, the converse is also true: a strict lifting of $\phi$ through the zig-zag $\tau_6$ extends by zero to the complementary direct summand in $\Omega_{S^3}\Sigma^\infty_{+, S^3} (S^4\dslash S^1)$.

%
To complete the proof, we argue that strict lifts 
of $p\circ \phi$ and $\mathrm{st}_{S^3}\circ \phi$ are equivalent to up-to-homotopy lifts.
Ultimately, this follows from
  item (\hyperlink{RelativeCellInclusion}{{\bf iii}}) in Prop. \ref{FiberwiseStabilizationOfUnitOnMinimalModels}. by using a standard argument in homotopy theory.
  For completeness, let us recall how this works: homotopy commutative squares in a homotopical category
  \begin{equation}
  \label{HomotopyCommSquare}
    \raisebox{20pt}{
    \xymatrix@C=5em@R=1.5em{
      X_1
      \ar[r]^-{ \hat f }_>{\ }="s"
      \ar[d]_{p}
      &
      Y_1
      \ar[d]^{q}
      \\
      X_2
      \ar[r]_-{f}^<{\ }="t"
      &
      Y_2
      \ar@{=>} "s"; "t"
    }
    }
  \end{equation}
  are equivalent to morphisms in the homotopy category of the corresponding arrow category.
  When working with a cofibrantly generated model category $\mathcal{C}$, the homotopy theory of the arrow category can be presented by the projective model structure on functors $\mathrm{Fun}(\Delta[1], \mathcal{C})$ (see e.g. \cite[Sec. A.2.8]{Lurie06}).
  General model category theory then implies that \emph{strictly} commutative squares are equivalent to homotopy commutative squares \eqref{HomotopyCommSquare} as soon as
  \begin{itemize}
  \item the source morphism $p\colon X_1 \to X_2$ is projectively cofibrant; and
  \item the target morphism $q\colon Y_1 \to Y_2$ is projectively fibrant.
  \end{itemize}
In terms of maps of DG-$\mathbb{Q}[h_3]$-modules, the fibrancy condition is always satisfied, whereas a map of DG-modules is projectively cofibrant if and only if it is a cofibration between cofibrant objects.
In the case at hand, observe that minimal DG-modules are cofibrant, and the dotted morphism in \eqref{TheMap} is a cofibration by Prop. \ref{FiberwiseStabilizationOfUnitOnMinimalModels} item (\hyperlink{RelativeCellInclusion}{iii}), as is the 6-truncation map of Ex. \ref{CyclificationOf4SphereReceives6TruncationOfTwistedK}.
Working with our rational models, this means that up-to-homotopy lifts of $p\circ \phi$  through $\tau_6$ are equivalent to strict lifts, and likewise for lifts of $\mathrm{st}\circ \phi$ through the stabilization of the adjunction unit.
This completes the proof.
\end{proof}

%
%


\section{Gauge enhancement of M-branes}
\label{TheMechanism}
In this section, we explain how Theorem \ref{TwistedKTheoryInsideFiberwiseStabilizationOfATypeOrbispaceOf4Sphere} provides a mechanism implementing gauge enhancement of fundamental M-branes.
Firstly, in Sec. \ref{With}, we explain the construction without local supersymmetry taken into account.
We present a solution to 
 \hyperlink{FirstRational}{\bf Open Problem, rational version 1} 
within the framework of rational homotopy theory that explains the core of the gauge enhancement mechanism via double dimensional reduction of flux forms as in \cite{MathaiSati04}. 
This focuses the problem on the enhancement from 6-truncated to untruncated twisted de Rham cohomology but falls short of exhibiting the existence and uniqueness of cocycles and their lifts,
which requires local supersymmetry.
Then, in Sec. \ref{TheAppearanceFromMTheoryOfTheFundamentalD6AndD8}, we pass to rational \emph{super} homotopy theory to explain the
 gauge enhancement of super M-branes from first principles.
 To be self-contained, we first recall how
super-spacetimes and super $p$-branes propagating in them are incarnated as super-cocycles in super-homotopy theory. We then explain how Theorem \ref{TwistedKTheoryInsideFiberwiseStabilizationOfATypeOrbispaceOf4Sphere} implements gauge enhancement of fundamental M-branes under double dimensional reduction to type IIA D-branes in the rational approximation.

\subsection{The mechanism without supersymmetry}
\label{With}

We first discuss the aspects of the gauge enhancement mechanism that are already apparent in rational homotopy, disregarding for the moment the crucial interaction with local supersymmetry. 
Being that local supersymmetry is a necessary ingredient for the very existence and classification of fundamental $p$-branes and their charges, we will incorporate it into our discussion in Sec. \ref{TheAppearanceFromMTheoryOfTheFundamentalD6AndD8}  below.
For the moment, however, we will simply assume that a collection of flux-like forms is provided to us on some spacetime without commenting on their origin in string or M-theory.

\medskip
We take as input data a smooth manifold $Y$ (to be thought of as our 11-dimensional background spacetime),
equipped with a map to the 4-sphere

\begin{equation}
  \label{Input}
  \xymatrix{
    Y
    \ar[rr]^-{ {\color{blue}(G_4, G_7)} }
    &&
    S^4\,,
  }
\end{equation}
to be thought of as measuring M-brane charge in $Y$.
In the approximation of rational homotopy theory (Def. \ref{ClassicalHomotopyCategories}), the datum of such a map \eqref{Input} is encoded by 
 a differential 4-form $G_4$ and a differential 7-form $G_7$ on $Y$ satisfying the relations
\begin{equation}
  \label{SugraTopologicalSector}
  (G_4,G_7)
  \;\sim \;
  \left\{
  \begin{aligned}
    d G_4 & = 0
    \\
    d G_7 & = - \tfrac{1}{2} G_4 \wedge G_4
    \,.
  \end{aligned}
  \right.
\end{equation}
This follows from the equivalence Prop. \ref{SullivanEquivalence}, using the minimal
 DG-algebra model for the 4-sphere from Ex. \ref{MinimalDgcAlgebraModelFor4Sphere}.
The relations \eqref{SugraTopologicalSector} have the form of the ``topological sector'' of the equations
of motion in 11-dimensional supergravity for the \emph{C-field strength} $G_4$ and its dual flux $G_7$,
 which we have identified with a rational homotopy map to the 4-sphere (this was first highlighted in \cite[Sec. 2.5]{S-top}).

\medskip
Observe that the contribution of the \lq\lq dual flux'' $G_7$ can be disentangled from that of the
 \lq\lq C-field strength'' $G_4$ by taking the homotopy pullback along the quaternionic Hopf fibration
 $H_{\mathbb{H}}\colon S^7 \to S^4$. This is an ${\rm SU}(2)$-principal bundle over $S^4$, hence classified
 by a map $\phi_\mathbb{H}\colon S^4 \to B {\rm SU}(2)$ to the classifying space, as in \eqref{GBundleByPullback}.
The cohomotopy cocycle \eqref{Input} decomposes along the quaternionic Hopf fibration as follows:

\begin{equation}
  \label{Decompo}
  \raisebox{50pt}{
  \xymatrix@R=1.2em{
    \widehat Y
    \ar[dd]
    \ar[rr]^-{ \color{blue} G_7}
    \ar@{}[ddrr]|{\binom{\rm homotopy}{\rm pullback}}
    &&
    S^7
    \ar[dd]^-{ H_{\mathbb{H}} }
    \\
    \\
    Y
      \ar[rr]|-{ \;(G_4, G_7)\; }
      \ar[ddr]_-{ \color{blue} G_4 }
    && S^4
    \ar[d]^-{ \phi_{\mathbb{H}} }
    \\
    && B {\rm SU}(2)
    \ar[dl]^{ \mathrm{c_2} }
    \\
    & { B^4 \mathbb{Q}, }
  }}
\end{equation}
where $\mathrm{c}_2 \colon B{\rm SU}(2)\to B^4 \mathbb{Q}$ is the second Chern class of the universal ${\rm SU}(2)$-bundle.
We are interested in the \emph{double dimensional reduction} of these data, assuming that $Y$ is a circle fibration as as in \eqref{DD}.
This means that we assume $Y$ sits in a homotopy fiber sequence
\begin{equation}
  \label{YFib}
  \raisebox{20pt}{
  \xymatrix@R=1.3em{
    *+[r]{ \color{blue} Y \simeq_{\mathrm{whe}} \mathrm{Ext}(X) }
    \ar[d]_{ \pi }
    \\
    X := Y \dslash S^1
    \ar[dr]
    \\
    & B S^1
    \,,
  }}
\end{equation}
hence that $Y$ is an \emph{$S^1$-extension} of $X$ as in Sec. \ref{TheAdjunction}. 
Consequently, our cohomotopy cocycle as in \eqref{Decompo} is
 now exhibited equivalently as a map out of a space in the image of the $\mathrm{Ext}$-functor \eqref{GExtF}:
$$
  \xymatrix@R=1.4em{
    {\color{blue}\mathrm{Ext}( X ) }
    \ar[rr]^-{ (G_4, G_7) }
    \ar[dr]_-{G_4}
    &&
    S^4
    \ar[dl]^{ c_2(\phi_\mathbb{H})}
    \\
    & B^4 \mathbb{Q}\,.
  }
$$
By the adjunction of Theorem \ref{GCycExtAdjunction}, 
maps out of the image of 
$\mathrm{Ext}$ are equivalent to maps into the image of the right adjoint 
 $\mathrm{Cyc}$, where
this 
equivalence is exhibited by
first applying 
$\mathrm{Cyc}$ and then precomposing with the adjunction unit
map $\eta_X\colon X \to  \mathrm{Cyc}\,\mathrm{Ext}(X)$.
(e.g. \cite[Sec. 3]{Borceux94}).
In summary, this means that 
the cohomotopy cocycle
on $Y$ (with rational components $(G_4, G_7)$) is naturally identified with a $\mathrm{Cyc}(S^4)$-valued
cocycle on $X$:

\begin{equation}
  \label{RatDDRed}
  \raisebox{30pt}{
  \xymatrix@C=20pt@R=1.5em{
    {\color{blue}X}
    \ar@[white]@/^1.25pc/[rrrr]^-{ \mbox{ \tiny \begin{tabular}{c} double dimensional reduction \\ of $(G_4, G_7)$ \end{tabular}} }
    \ar@[blue]@/^1pc/[drrrr]|-{\,{\color{blue}\widetilde{(G_4, G_7)}}\, }
    \ar[d]_{\color{blue}\eta_X}
    &&&&
    \\
    {\color{blue}\mathrm{Cyc}}\,\mathrm{Ext}(X)
      \ar[rrrr]|{\; {\color{blue}\mathrm{Cyc}}(G_4,G_7)\; }
      \ar[drr]
    &&&& {\color{blue}\mathrm{Cyc}}(S^4)\;,
    \ar[dll]
    \\
    && {\color{blue}\mathrm{Cyc}}(B^4 \mathbb{Q})
    \ar[d]
    \\
    &&
    B^3 \mathbb{Q}
  }}
\end{equation}
where the map $\mathrm{Cyc}(B^4 \mathbb{Q})\to B^3 \mathbb{Q}$ is obtained using the infinite loop space structure of $B^4\mathbb{Q}$; in terms of minimal DG-algebra models it is simply the map
\[
    \xymatrix@R=.5pt@C=4em{
      \mathrm{Cyc}(B^4 \mathbb{Q})
      \ar[rr]
      &&
      B^3 \mathbb{Q}\;.
      \\
      \omega_2\phantom{\,.} &&
      \\
      \omega_3\phantom{\,.} && \omega_3 \ar@{|->}[ll]
      \\
      \omega_4
      &&
    }
\]
Using the minimal DG-algebra model for $\mathrm{Cyc}(S^4)$ (Ex. \ref{MinimalDGCAlgebraModelForCyclicSpaceOfFourSphere}),
we have that the $\mathrm{Ext}/\mathrm{Cyc}$-adjunct of the cohomotopy cocycle $(G_4, G_7)$
is rationally determined by a collection of differential forms on $X$ satisfying the following relations:

\begin{equation}
  \label{DDRed}
  \widetilde{(G_4,G_7)}
  \;\;=\;\;
  \left\{
  \begin{aligned}
    d H_3 & = 0\;, \phantom{AAA}\phantom{BB}\phantom{CC}  d H_7  = F_2 \wedge F_6 - \tfrac{1}{2} F_4 \wedge F_4\;,
    \\
    \\
    d F_2 & = 0\;,
    \\
    d F_4 & = H_3 \wedge F_2\;,
    \\
    d F_6 & = H_3 \wedge F_4
    \,.
  \end{aligned}
  \right.
\end{equation}
These differential forms can also be seen as arising from the data $(G_4, G_7)$ via the Gysin sequence
\cite[Sec 4.2]{MathaiSati04}.
Interpreting $G_4$ and $G_7$
as M-brane flux forms in 11-dimensional supergravity, corresponding to the M2-brane and the M5-brane respectively,
then the forms in \eqref{DDRed} are naturally interpreted as flux forms of brane species in 10-dimensional type IIA  supergravity, together with their
Bianchi identities.
We notice that not all of the expected brane flux forms appear: we have
the
 $\mathrm{NS1}$-brane flux $H_3$, the RR-flux forms $F_{(2p+2\leq 6)}$ for the  $\mathrm{D}(2p \leq 4)$-branes,
 as well  the $\mathrm{NS5}$-brane flux $H_7$,\footnote{
The identity $d H_7  = F_2 \wedge F_6 - \tfrac{1}{2} F_4 \wedge F_4$
  for the type IIA $\mathrm{NS5}$-brane flux does not hold after fiberwise stabilization
  in \eqref{DDRedEnhc}---indeed, the flux form $H_7$ is part of the obstruction to completing the zig-zag truncation map $\tau_6$ in Ex. \ref{CyclificationOf4SphereReceives6TruncationOfTwistedK} to an actual homomorphism.} 
but the flux forms for the D6 and D8 branes are missing.

\medskip
So far we have merely recapitulated the mechanism of double dimensional reduction formalized via the
$\mathrm{Ext}/\mathrm{Cyc}$-adjunction in rational homotopy theory \cite{FSS13,FSS16b}.
We are now interested in how 
 the collection of flux forms in \eqref{DDRed} may be naturally enhanced so as to contain 
 flux forms $F_8$ and $F_{10}$. 
As discussed in Rem. \ref{TheATypeQuotientFromSpacetime}, operations on spacetime $Y$ should be accompanied
by corresponding operations on the 4-sphere brane charge coefficient. 
This means that we should consider the A-type circle action on the 4-sphere:
\begin{equation}
  \label{SphereExt}
  \raisebox{35pt}{\xymatrix@R=1.4em{
    &
    *+[l]{\color{blue} \mathrm{Ext}(S^4 \dslash S^1) \simeq_{\mathrm{whe}} S^4}
    \ar[d]
    \\
    & S^4 \dslash S^1 \ar[dl]
    \\
    B S^1
  }}
\end{equation}
Diagram \eqref{RatDDRed} can now be cast in a more symmetric form:
\begin{equation}
  \label{SymmetricAtLast}
  \raisebox{35pt}{
  \xymatrix@C=25pt@R=1.5em{
    {\color{blue}X}
    \ar[d]_-{\color{blue}\eta_{X}}
    \ar@[blue]@/^1pc/[drrrr]|-{\,{\color{blue}\widetilde{(G_4, G_7)}}\, }
    && &&
    {\color{blue}  S^4 \dslash S^1}
    \ar[d]^{\color{blue}\eta_{{}_{S^4 \dslash S^1}}}
    \\
    {\color{blue}\mathrm{Cyc}}\,\mathrm{Ext}(X)
      \ar[rrrr]|{\; {\color{blue}\mathrm{Cyc}}(G_4,G_7)\; }
      \ar[drr]
    &&&& {\color{blue}\mathrm{Cyc}}( \mathrm{Ext}(S^4 \dslash S^1) )\;.
    \ar[dll]
    \\
    && {\color{blue}\mathrm{Cyc}}(B^4 \mathbb{Q})
    \ar[d]
    \\
    &&
    B^3 \mathbb{Q}
  }
  }
\end{equation}
Diagram \eqref{SymmetricAtLast} poses the natural question of whether we can find a lift through the $\mathrm{Ext}/\mathrm{Cyc}$-unit:
%
\begin{equation}
  \label{IsIt}
  \raisebox{35pt}{
  \xymatrix@C=25pt@R=1.5em{
    X
    \ar@[blue]@{-->}@/^1.5pc/[rrrr]^-{ \mbox{ \tiny \begin{tabular}{c} enhanced \\ double dimensional reduction \\ of $(G_4,G_7)$ {\color{blue} ??}  \end{tabular} }   }
    \ar[d]_{\eta_X}
    && &&
    {S^4 \dslash S^1}
    \ar[d]^{\eta_{{}_{S^4 \dslash S^1}}}
    \\
    { \mathrm{Cyc}}\,\mathrm{Ext}(X)
      \ar[rrrr]|-{\;\mathrm{Cyc}(G_4,G_7)\; }
      \ar[drr]_{H_3}
    &&&& { \mathrm{Cyc}}( \mathrm{Ext}(S^4 \dslash S^1) )\;.
    \ar[dll]
    \\
    && B^3 \mathbb{Q}\,.
  }
  }
\end{equation}
%
%
In general, there are strong obstructions to such a lift: Prop. \ref{MinimalDGCAlgebraModelForATypeOrbispace}
implies, via \eqref{AdjunctionUnitOnRationalATypeOrbispace}, that for such a lift to exist,
the cocycle data in \eqref{DDRed} needs to satisfy the conditions that $F_6 = 0$ and $H_3 \wedge F_2 = 0$
(so that $d F_4 = 0$).
Once we take local supersymmetry into account, we will see that these conditions are too restrictive and that a direct lift as in \eqref{IsIt} simply does not exist.

\begin{remark}[Goodwillie calculus]
\label{GoodwillieCalculus}
At this point it behooves us to highlight that similarly to string theory and quantum field theory,  homotopy theory also has a concept of \emph{perturbative approximation}.
As we have touched upon at varied points, homotopy theory is an immensely rich and computationally demanding area of mathematics (see \cite{Ravenel03, HHR09} for good examples of this).
However, in striking analogy to Taylor series expansions in differential calculus, mapping spaces between homotopy types can often be approximated by a sequence of increasingly accurate approximations called the \emph{Goodwillie--Taylor tower} (see \cite{Kuhn04} and \cite[Ch. 6]{Lurie09}).
The first-order approximation in this \emph{Goodwillie perturbation theory} is provided by stabilization
adjunction between spaces and spectra, as discussed in Sec. \ref{RationalParameterizedStableHomotopyTheory}.
\end{remark}
In view of Rem. \ref{GoodwillieCalculus}, as we encounter strong obstructions in  homotopy theory as in \eqref{IsIt},
we are led to ask whether the obstruction persists ``perturbatively'', hence stage-wise in the Goodwillie--Taylor approximation. To first order in Goodwillie perturbation theory, this means that we should ask whether enhanced double
 dimensional reduction exists after \emph{fiberwise stabilization} (Prop. \ref{AdjunctionStabilization}) over $S^3 \simeq_\mathbb{Q} B^3\mathbb{Q}$:

\begin{equation}
  \label{Enhc}
  \raisebox{40pt}{
  \xymatrix@R=1.6em@C=4em{
    X
    \ar@[blue]@{-->}@/^2pc/[rrrr]^-{
      \mbox{
        \tiny
        \begin{tabular}{c}
          {\color{blue} perturbatively} enhanced
          \\
          double dimensional reduction
          \\
          of $(G_4,G_7)$  {\color{blue} ?? }
        \end{tabular}
      }
    }
    \ar[d]_{\eta_{{X}}}
    &&
    &&
    {\color{blue} \Omega^\infty_{S^3}\Sigma^\infty_{+,S^3}} \left( S^4 \dslash S^1 \right)
    \ar@[blue][d]^{ {\color{blue} \Omega^\infty_{S^3}\Sigma^\infty_{+,S^3}} \big( \eta_{{S^4 \dslash S^1}} \hspace{-.7mm} \big) }
    \\
    { \mathrm{Cyc}}\,\mathrm{Ext}(X)
      \ar[rr]^-{\mathrm{Cyc}(G_4,G_7)\; }
      \ar[drr]_{H_3}
    && \mathrm{Cyc}( S^4 )
    \ar[d]
    \ar@[blue][rr]^-{ \color{blue} \mathrm{st}}
    &&
    {\color{blue} \Omega^\infty_{S^3}\Sigma^\infty_{+,S^3}} \mathrm{Cyc}(S^4)
    \ar[dll]
    \\
    && B^3 \mathbb{Q}
  }
  }
\end{equation}
%
But this is precisely the question that is addressed by Theorem \ref{TwistedKTheoryInsideFiberwiseStabilizationOfATypeOrbispaceOf4Sphere}: the obstruction to a dotted morphsim as in \eqref{Enhc} is precisely the lift to untruncated twisted K-theory as in the diagram:
\begin{equation}
  \label{FactorEnhc}
  \hspace{-10mm}
  \mathclap{
  \raisebox{40pt}{
  \xymatrix@R=1.8em@C=40pt{
    X
    \ar@{-->}@/^2.4pc/@[blue][rrrr]^-{
      \mbox{
        \tiny
        \begin{tabular}{c}
           perturbatively enhanced
          \\
          double dimensional reduction
          \\
          of $(G_4,G_7)$  {\color{blue} ! }
        \end{tabular}
      }
    }
    \ar@{-->}@[blue][rr]^{ \color{blue} \widehat{(G_4,G_7)} }
    \ar[drr]|-{\;{\widetilde{(G_4, G_7)}}\; }
    \ar[d]_{\eta_{{X}}}
    &&
    {\color{blue} \Omega^{\infty-2}_{B^3 \mathbb{Q}}\left( \mathrm{ku} \dslash B S^1 \right) \,}
    \ar@[blue]@{..}[d]|-{\color{blue} \tau_6}
    \ar@{^{(}->}[rr]^-{\iota}
    &&
    \Omega^\infty_{S^3}\Sigma^\infty_{+,S^3} {S^4 \dslash S^1}
    \ar[d]^{ \Omega^\infty_{S^3}\Sigma^\infty_{+,S^3} \big( \eta_{{}_{S^4 \dslash S^1}}  \hspace{-.7mm}\big) }
    \\
    { \mathrm{Cyc}}\,\mathrm{Ext}(X)
      \ar[rr]|-{\;\mathrm{Cyc}(G_4,G_7)\; }
      \ar[drr]_{H_3}
    && \mathrm{Cyc}( S^4 )
    \ar[d]
    \ar[rr]^-{\mathrm{st}}
    &&
    \Omega^\infty_{S^3}\Sigma^\infty_{+,S^3}\mathrm{Cyc}(S^4)
    \ar[dll]
    \\
    && B^3 \mathbb{Q}\,.
  }
  }}
\end{equation}
%
That is, exhibiting a \emph{perturbatively} enhanced double dimensional reduction cocycle is equivalent to specifying an extension of the cocycle data \eqref{DDRed} to a cocycle of the form
\begin{equation}
  \label{DDRedEnhc}
  \widehat{(G_4,G_7)}
  \;\;=\;\;
  \left\{
  \begin{aligned}
    d H_3 & = 0
    \\
    \\
    d F_2 & = 0
    \\
    d F_4 & = H_3 \wedge F_2
    \\
    d F_6 & = H_3 \wedge F_4
    \\
    {\color{blue} d F_8 } & {\color{blue} = H_3 \wedge F_6 }
    \\
    {\color{blue} d F_{10} } & {\color{blue} = H_3 \wedge F_8 }
    \\
    &\;{\color{blue} \vdots }
  \end{aligned}
  \right.
\end{equation}
According to the discussion in Sec. \ref{Before}, such a cocycle exhibits the full gauge ehancement mechanism in rational homotopy theory.
That is,

\begin{quotation}
\noindent A {\bf solution} to \hyperlink{FirstRational}{\bf Open Problem, rational version 1} (p.~
\pageref{FirstRationalProblem}) is obtained as follows:
\begin{itemize}
  \item Double dimensional reduction of the M-theory flux data $(G_4, G_7)$ of \eqref{SugraTopologicalSector} is given by the $\mathrm{Ext}/\mathrm{Cyc}$-adjunct along the M-theory spacetime extension \eqref{YFib}; and

  \item Its \emph{perturbative gauge enhancement}, making all RR flux forms appear, is exhibited by lifting against the fiberwise stabilization of the $\mathrm{Ext}/\mathrm{Cyc}$-unit on the A-type orbispace of the 4-sphere \eqref{FactorEnhc}\eqref{DDRedEnhc}.
\end{itemize}
\end{quotation}


\begin{remark}[Copies of twisted K-theory]
Let us briefly pause to highlight a curious aspect of our result on gauge enhancement. By Prop. \ref{MinimalDGModels},
 there are in fact \emph{two} direct summands of rationalized twisted K-theory inside the rational fiberwise stabilization
 of the A-type orbispace. Theorem \ref{TwistedKTheoryInsideFiberwiseStabilizationOfATypeOrbispaceOf4Sphere}
assigns a special role to the shifted copy $\Sigma^2 \mathrm{ku} \dslash B S^1$, since it is this copy that obstructs the
lifting problem. The 2-fold suspension here reflects the fact that RR-flux forms in type IIA string theory ordinarily start
with D0-flux $F_2$ in degree 2. But there is also a rational copy of \emph{un}shifted twisted K-theory; a cocycle landing
in this this copy corresponds to a twisted de Rham cocycle whose lowest component is a 0-form instead of a 2-form.
The subtle issue of considering this possibility in the context of the K-theory classification of RR-flux is discussed in
\cite{MoS, MathaiSati04, BV} and revisited systematically, and in a refined form, in \cite{GS19}.
%
\end{remark}

\medskip

\subsection{The mechanism with local supersymmetry}
\label{TheAppearanceFromMTheoryOfTheFundamentalD6AndD8}
\label{FundamentalpBranes}

In this section, we finally present our solution to the problem of gauge enhancement for super M-branes in its formulation as \hyperlink{OpenRational}{\bf Open Problem, rational version 2}
(p.~\pageref{OpenRationalPage}).
We proceed as per the solution to
\hyperlink{FirstRational}{\bf Open Problem, rational version 1} detailed in the previous section, 
with the important caveat that we 
now work in the proper context for super $p$-branes, namely locally supersymmetric supergeometry.
Before presenting our solution, we briefly review some of the necessary background material 
(for a detailed introduction with an emphasis on the aspects of relevance here, see \cite{Schreiber16}).

\medskip
While there is no known mechanism in string theory that would enforce -- or even prefer -- \emph{global} supersymmetry,
\footnote{Models of particle physics obtained from dimensional reductions of M-theory on singular manifolds of $G_2$-holonomy (see \cite{Kane17})
are among the globally supersymmetric extensions of the Standard Model of particle physics that are, so far, still consistent with
experimental constraints \cite{BGK18}.
If and when supersymmetric extensions of the Standard Model are ruled out, 
then this will also rule out dimensional reductions of
M-theory on singular fiber manifolds of $G_2$-holonomy
as realistic models for particle physics.
However, this particular type of dimensional reduction is in no way 
dictated by the theory, and are certainly not generic amongst all possibilities, 
but were motivated by
the expectation of global supersymmetry in the first place. What \emph{is} dictated by the theory is \emph{local} supersymmetry, which is already present as soon as fermions are.}
what is predicted by perturbative string theory, and what is at the heart of its non-perturbative description
as in \cite{ADE}, is \emph{local} supersymmetry, in particular \emph{supergravity}.
An elegant  way to understand supergravity is via \emph{Cartan geometry}, a powerful generalization of
Riemannian geometry (see \cite[Ch. 1]{CapSlovak09} for an introduction).
In Cartan geometry, one fixes a local model space $V$ and a subgroup $G \subset \mathrm{Aut}(V)$ of its
automorphism group, and then studies spaces that look \emph{locally}, namely on the infinitesimal
neighborhood of every point,
like $V$ equipped with a reduction of the structure group to $G$:
$$
  \mbox{
    \begin{tabular}{|c|c|}
      \hline
      \begin{tabular}{c}
        Local model for
        \\
        Cartan geometry
      \end{tabular}
      &
      \raisebox{-9pt}{
        \xymatrix{
          V \ar@(ul,ur)[]^{ G }
        }
      }
      \\
      \hline
    \end{tabular}
  }
$$
In the physics literature, this is known as the method of \emph{local gauging via moving frames}
(following \cite{Cartan1923}),
while mathematically this is the study of \emph{(torsion-free) $G$-structures} \cite{Guillemin65}.
If one chooses the local model to be Minkowski spacetime
\begin{equation}
  \label{MinkowskiSpacetime}
  \mbox{
    \begin{tabular}{|c|c|}
      \hline
      \begin{tabular}{c}
        Local model for off-shell
        \\
        gravity
      \end{tabular}
      &
      \raisebox{-9pt}{
        \xymatrix{ \mathbb{R}^{d,1} \ar@(ul,ur)[]^{{\rm SO}(d,1) } }
      }
      \\
      \hline
    \end{tabular}
  }
\end{equation}
then the resulting Cartan geometries are precisely the pseudo-Riemannian manifolds modelling configurations in general relativity.
However, not all of these configurations are physically admissible (they are not all \lq\lq on-shell'') and
we need to impose the Einstein field equations to pick out the physically admissible geometries. 
In this Cartan-geometric formulation, the passage to supergravity is immediate:
taking the local model space to be instead \emph{super} Minkowski spacetime
\begin{equation}
  \label{SuperMinkowskiSpacetime}
  \mbox{
    \begin{tabular}{|c|c|}
      \hline
      \begin{tabular}{c}
        Local model for off-shell
        \\
        supergravity
      \end{tabular}
      &
      \raisebox{-9pt}{
        \xymatrix{ \mathbb{R}^{d,1 \vert \mathbf{N}} \ar@(ul,ur)[]^{ {\rm Spin}(d,1)}  }
      }
      \\
      \hline
    \end{tabular}
  }
\end{equation}
then the resulting Cartan geometries are supergravity super-spacetimes. This is spelled out explicitly in \cite{Lott90, EE},
and is implicit in much of the physics literature on supergravity, following \cite{CDF}.

\medskip
In the spirit of Kleinian geometry \cite{Klein1872}, the group $G$ acting on the local model space $V$ 
reflects extra structure on $V$. 
For plain Minkowski spacetime \eqref{MinkowskiSpacetime} with its Lorentz group action, this extra structure is
the Minkowski metric of special relativity. For the super-Minkowski spacetime \eqref{SuperMinkowskiSpacetime}
with its Lorentzian Spin group action, this
structure is the translational \emph{supersymmetry} super Lie algebra structure on $\mathbb{R}^{d,1\vert \mathbf{N}}$.
Recall that the only non-trivial component  of the super Lie bracket is the odd-odd spinor-to-vector pairing:
$$
  \xymatrix@R=-2pt{
    \mathbf{N} \otimes \mathbf{N}
    \ar[rr]^{ \overline{(-)}\Gamma(-) }
    &&
    \mathbb{R}^{d,1}
    \\
    \psi_1, \psi_2
    \ar@{|->}[rr]
    &&
    \left(\overline{\psi}_1 \Gamma^a \psi_2\right)_{a=0}^d\,.
  }
$$
The Chevalley--Eilenberg algebra of this super-Minkowski super Lie algebra is a super DG-algebra
\begin{equation}
  \label{CEAlgebraSuperMinkowski}
  \mathrm{CE}\big(
    \mathbb{R}^{d,1\vert \mathbf{N}}
  \big)
  \;:=\;
  \Bigg(
  \mathbb{R}\big[\!\!
    \underset{ \mathrm{deg} = (1,\mathrm{even}) }{\underbrace{(e^a)_{a =0}^d}} \!\!, \;
    \underset{  (1, \mathrm{odd}) }{\underbrace{ ( \psi^\alpha )_{\alpha = 1}^{N} }} \;
  \big]\Big/
  {\small
  \left(
    \begin{aligned}
      d e^a & = \overline{\psi}\Gamma^a \psi
      \\[-1mm]
      d \psi^\alpha & = 0
    \end{aligned}
  \right)}
  \Bigg)
  \;\;
  \in
  \mathrm{DGCSuperAlg}
  \,,
\end{equation}
like those used in rational homotopy theory in Def. \ref{dgAlgebrasAnddgModules}, but equipped with additional $\mathbb{Z}_2$-grading.
In view of the equivalence of Prop. \ref{SullivanEquivalence} we may think of \emph{super} DG-algebras of finite type
as the formal duals of rational \emph{superspaces} (see \cite[Sec. 3.2]{ADE} for details in this context):
\begin{equation}
  \label{RationalSuperSpaces}
  \mathrm{Ho}\left( \mathrm{SuperSpaces}_{\mathbb{R}, \mathrm{cn}, \mathrm{ft}, \mathrm{nil}} \right)
  \;:=\;
  \mathrm{Ho}
  \left(
    \mathrm{DGCSuperAlg}_{\mathrm{cn}, \mathrm{ft}}^{\mathrm{op}}
  \right).
\end{equation}
All of ordinary (henceforth \lq\lq bosonic'') rational
homotopy theory embeds into rational super homotopy theory by
\begin{itemize}
  \item replacing the ground field by $\mathbb{R}$ via $A \mapsto A\otimes_\mathbb{Q}\mathbb{R};
  $\footnote{From the point of view of homotopy theory, there is little difference between working over $\mathbb{Q}$ or over $\mathbb{R}$.}
  and
  \item regarding ordinary DG-algebras are having an additional $\mathbb{Z}_2$-grading in which all elements have trivial degree.
\end{itemize}
For instance, a morphism of the form
\vspace{-3mm}
$$
  \xymatrix{
    \mathbb{R}^{d,1\vert \mathbf{N}}
    \ar@(ul,ur)[]^{{\rm Spin}(d,1)}
    \ar[rr]^-{\mu_p}
    &&
    {B}^{p+1} S^1
  }
$$
in rational super homotopy theory is, equivalently, a ${\rm Spin}(d,1)$-invariant super Lie algebra cocycle of $\mathbb{R}^{p,1|\mathbf{N}}$ in degree $p+2$. 
The requirement of ${\rm Spin}(d,1)$-invariance forces such a cocycle to be a product of elements of the form
$$
  \mu_{p}
  \;:=\;
  \left(
    \overline{\psi} \Gamma_{a_1 \cdots a_p} \psi
  \right)
  \wedge e^{a_1} \wedge \cdots \wedge e^{a_p}
  \;\;\in
  \mathrm{CE}(
    \mathbb{R}^{d,1\vert \mathbf{N}}
  )
  \,.
$$
Remarkably, when the local model space is taken to be the
 $D = 11$, $\mathcal{N} = 1$ super-Minkowski spacetime
$$
  \mbox{
    \begin{tabular}{|c|c|}
      \hline
      \begin{tabular}{c}
        Local model for \emph{on-shell}
        \\
        $D = 11$, $\mathcal{N} = 1$ supergravity
      \end{tabular}
      &
      \raisebox{-9pt}{
        \xymatrix{ \mathbb{R}^{10,1 \vert \mathbf{32}} \ar@(ul,ur)[]^{{\rm Spin}(10, 1)}  }
      }
      \\
      \hline
    \end{tabular}
  }
$$
the condition of (super-)torsion freeness is already equivalent to the supergravity equations of motion \cite{CandielloLechner94, Howe97, FOS}.
That is, (super-)torsion-free $D = 11$, $\mathcal{N} =1$ (super-)Cartan geometries are precisely on-shell configurations of supergravity.
In fact, more is true: the supergravity equations of motion imply that the bifermionic components of the supergravity $C$-field curvature 4-form $G_4$ and its Hodge dual 7-form $G_7$ are covariantly constant on each super tangent space $\mathbb{R}^{10,1\vert \mathbf{32}}$ and constrained there to be of the form
\begin{equation}
  \label{M2M5Cochains}
  (G_4,G_7)_{\mathrm{fermionic}}
  =
  \left(
    \begin{aligned}
     \mu_{M2} &:= \tfrac{i}{2} \left(\overline{\Psi} \Gamma_{a_1 a_2} \Psi\right) \wedge E^{a_1} \wedge E^{a_2}\,,
      \\
      \mu_{M5}&:= \tfrac{1}{5!}
    \left(\overline{\Psi} \Gamma_{a_1 \cdots a_5} \Psi \right) \wedge E^{a_1} \wedge \cdots \wedge E^{a_5}
    \end{aligned}
  \right).
\end{equation}
In this expression, $(E^a, \Psi^\alpha)$ is the \emph{super-vielbein} field on super-spacetime, hence the supergeometric
analog of the vielbein field determining a Cartan geometry.  Torsion-freeness implies that on the infinitesimal neighborhood
of each point of super-spacetime, the super-vielbein reproduces the super left-invariant 1-form generators
$(e^a, \psi^\alpha)$ of $\mathrm{CE}\big( \mathbb{R}^{10,1\vert \mathbf{32}} \big)$ as in \eqref{CEAlgebraSuperMinkowski}.
Due to the \emph{Fierz identities} of Spin-representation theory \cite{DF}, the super-forms of \eqref{M2M5Cochains} satisfy the relations
\begin{equation}
  \label{M2M5Coc}
  \mu_{{}_{M2/M5}}
  \;=\;
  \left\{
  \begin{aligned}
    d \mu_{{}_{M2}} & = 0
    \\
    d \mu_{{}_{M5}} & = -\tfrac{1}{2} \mu_{{}_{M2}} \wedge \mu_{{}_{M2}}
  \end{aligned}
  \right.
\end{equation}
on each super tangent space $\mathbb{R}^{10,1\vert \mathbf{32}}$.
Moreover, this is
precisely the super tangent space-wise condition that witnesses propagation of fundamental super M2-branes \cite[(15)]{BST2} and
the fundamental super M5-branes \cite[(5)]{LT}, \cite[(6)]{BLNPST97} in the background super-spacetime (see \cite[Sec. 2.1]{ADE} for further background on fundamental super $p$-branes).

\medskip
After comparing the $\mathrm{M2}/\mathrm{M5}$-brane Fierz identity \eqref{M2M5Coc}
with the minimal DG-algebra model for the 4-sphere (Ex. \ref{MinimalDgcAlgebraModelFor4Sphere},
 also \eqref{SugraTopologicalSector}), we may summarize the state of affairs
using the language of (rational) super homotopy theory 
 as follows:
\begin{quote}{
11-dimensional super-Minkowski spacetime carries an exceptional super-cocycle in rational cohomotopy of degree four \cite{FSS16a, cohomotopy},
and 11-dimensional supergravtiy super-spacetimes together with fundamental M2/M5-branes propagating in them
are the \emph{higher} Cartan geometries \cite{Schreiber15, Wellen} locally modelled on the higher geometric data:}
\begin{equation}
  \label{LocalM}
  \mbox{
  \begin{tabular}{|c|c|}
    \hline
    \begin{tabular}{c}
      Local model for
      \\
      fundamental $\mathrm{M2}/\mathrm{M5}$-branes
      \\
      in $D = 11$, $\mathcal{N} = 1$ supergravity
      \\
    \end{tabular}
    &
    \raisebox{-9pt}{
    \xymatrix{
      \mathbb{R}^{10,1\vert \mathbf{32}}
      \ar@(ul,ur)[]^{{\rm Spin}(10,1) }
      \ar[rr]^-{ \mu_{{}_{M2/M5}} }
      &&
      S^4
    }
    }
    \\
    \hline
  \end{tabular}
  }
\end{equation}
\end{quote}

\medskip
\noindent Analogous statements hold for all fundamental branes that appear in string theory \cite{FSS13}. In particular,
the super-spacetimes of $D = 10$, $\mathcal{N} = (1,1)$ supergravity (that is, type IIA supergravity) are
Cartan geometries locally modelled on the type IIA super-Minkowski spacetime
$$
  \mbox{
    \begin{tabular}{|c|c|}
      \hline
      \begin{tabular}{c}
        Local model for
        \\
        $D = 10$, $\mathcal{N} = (1,1)$ supergravity
      \end{tabular}
      &
      \raisebox{-9pt}{
        \xymatrix{
          \mathbb{R}^{9,1 \vert \mathbf{16} + \overline{\mathbf{16}}} \ar@(ul,ur)[]^{ {\rm Spin}(9,1)}
        }
      }
      \\
      \hline
    \end{tabular}
  }
$$
The presence of the fundamental $\mathrm{F1}$- and $\mathrm{D}p$-branes propagating in these super-spacetimes is captured, as for the M2/M5-brane cocycles above, by the non-trivial super Lie algebra cocycles
\cite[(3.9)]{CGNSW97}:
\begin{equation}
  \label{DCoch}
  \begin{aligned}
    \mu_{{}_{F1}^{\mathrm{IIA}}}
    & =
    i \left( \overline \Psi \Gamma_{a} \Gamma_{10} \Psi\right) \wedge E^a
    \\
    \mu_{{}_{D0}}
    & =
    \left(\overline{\Psi} \Gamma_{10} \Psi\right)
    \\
    \mu_{{}_{D2}}
    & =
    \tfrac{i}{2} \left( \overline{\Psi} \Gamma_{a_1 a_2} \Psi \right) \wedge E^{a_1} \wedge E^{a_2}
    \\
    \mu_{{}_{D4}}
    & =
    \tfrac{1}{4!} \left( \overline{\Psi} \Gamma_{a_1 \cdots a_4} \Gamma_{10} \Psi \right) \wedge E^{a_1} \wedge \cdots \wedge E^{a_4}
    \\
    \mu_{{}_{D6}}
    & =
    \tfrac{i}{6!} \left( \overline{\Psi} \Gamma_{a_1 \cdots a_6} \Psi \right) \wedge E^{a_1} \wedge \cdots \wedge E^{a_6}
    \\
    \mu_{{}_{D8}}
    & =
    \tfrac{1}{8!} \left( \overline{\Psi} \Gamma_{a_1 \cdots a_8} \Gamma_{10} \Psi \right) \wedge E^{a_1} \wedge \cdots \wedge E^{a_8}
    \\
    \mu_{{}_{D10}}
    & =
    \tfrac{i}{10!} \left( \overline{\Psi} \Gamma_{a_1 \cdots a_{10}} \Psi \right) \wedge E^{a_1} \wedge \cdots \wedge E^{a_{10}}
    \\
    \big(\mu_{{}_{NS5}} &= \tfrac{1}{5!} \left(\overline{\Psi} \Gamma_{a_1 \dotsb a_5}\Psi\right) \wedge E^{a_1}\wedge \dotsb \wedge E^{a_5}\big)\;,
  \end{aligned}
\end{equation}
where we have also included here the NS5-brane cocycle $\mu_{{}_{NS5}}$ for completeness.
Using the Fierz identities for ${\rm Spin}(9,1)$-representations, one finds that these expressions satisfy the following relations:
\begin{equation}
  \label{DCoc}
  ( \mu_{{}_{F1/Dp}} )
  \;=\;
  \left\{
    \begin{aligned}
      d \mu_{{}_{F1}} & = 0
      \phantom{AAAABBBCCC}
    \big(d\mu_{{}_{NS5}} = \mu_{{}_{D0}}\wedge \mu_{{}_{D4}} -\tfrac{1}{2}\mu_{{}_{D2}}\wedge \mu_{{}_{D2}}\big)
      \\
      d \mu_{{}_{D(2p+4)}} & = \mu_{{}_{F1}} \wedge \mu_{{}_{D(p+2)}}\,,
    \end{aligned}
  \right.
\end{equation}
(\cite[App.~A]{CGNSW97}, \cite[(6.8), (6.9)]{CAIB00}, see also \cite[Theorem 4.16]{FSS16a},  \cite[Prop. 4.8]{FSS16b}).
This is the supersymmetric version of the cocycle \eqref{DDRedEnhc} that we had encountered in the previous section.

\medskip
By analogy with \eqref{LocalM}, using Lemma \ref{TwistedKModel} we may concisely summarize this in the language of rational super homotopy theory
as follows:
\begin{quote}
{
10-dimensional type IIA super-Minkowski spacetime carries an exceptional cocycle in rational 2-shifted twisted connective K-theory
and super-spacetimes of 10-dimensional type IIA supergravity together with fundamental $\mathrm{F1}/\mathrm{D}p$-branes
propagating inside them are the \emph{higher} Cartan geometries locally modelled on the higher geometric data:}
\begin{equation}
  \label{LocalIIA}
  \mbox{
    \begin{tabular}{|c|c|}
      \hline
      \begin{tabular}{c}
        Local model for
        \\
        fundamental $\mathrm{F1}/\mathrm{D}p$-branes
        \\
        in $D =10$, $\mathcal{N} = (1,1)$ supergravity
      \end{tabular}
      &
      \raisebox{10pt}{
      \xymatrix@R=1.7em{
        \mathbb{R}^{9,1\vert \mathbf{16} + \overline{\mathbf{16}}}
        \ar@(ul,ur)[]^{{\rm Spin}(9,1) }
        \ar[rr]^-{ \mu_{{}_{F1/Dp}} }
        \ar[dr]_{\mu_{{}_{F1}}}
        &&
        \mathrm{ku} \dslash B S^1
        \ar[dl]
        \\
        & B^3 \mathbb{Q}
      }
      }
      \\
      \hline
    \end{tabular}
  }
\end{equation}
\end{quote}
%
%
Working with higher Cartan geometry in this way, one finds large parts of the string/M-theory literature appearing in cocycle incarnations 
 (see \cite[Sec. 2]{ADE}
for a more detailed account).
The upshot is that all information about brane species and their behavior in string/M-theory is encoded in \emph{cohomological
data}, depending on the local model space.
The actual brane dynamics are provded by the higher-super-Cartan-geometric globalization of these local data.

\medskip
The problem of gauge enhancement for M-branes, therefore, reduces to the question of how 
double dimensional reduction
turns the local model \eqref{LocalM} for $D = 11$, $\mathcal{N} = 1$ superspacetime with its $\mathrm{M2}/\mathrm{M5}$-brane cocycle 
into the local model \eqref{LocalIIA}
for $D = 10$, $\mathcal{N} = (1,1)$ super-spacetime with its unified D-brane coycle. 
%
%
This works via the mechanism of Sec. \ref{With}: our input datum \eqref{Input} is now specifically the fundamental $\mathrm{M2}/\mathrm{M5}$-brane super cocycle \eqref{LocalM}

\begin{equation}
  \label{FundamentalM2M5}
  \xymatrix{
    \mathbb{R}^{10,1\vert \mathbf{32}}
    \ar[rr]^{ \color{blue} \mu_{{}_{M2/M5}} }
    &&
    S^4\,.
  }
\end{equation}
The supersymmetric version of the spacetime circle extension \eqref{YFib} is the extension of type IIA super-Minkowski spacetime
to $D =11$ super-Minkowski spacetime classified by the D0-brane cocycle \cite[Prop. 4.5]{FSS13} (this extension implements \lq\lq D0-brane condensation''):
\begin{equation}
  \label{11dExt}
  \xymatrix@R=1.4em{
    *+[r]{
      \color{blue}
      \mathbb{R}^{10,1\vert \mathbf{32}}
      \;\simeq\;
      \mathrm{Ext}\left( \mathbb{R}^{9,1\vert \mathbf{16} + \overline{\mathbf{16}}} \right)
    }
    \ar[d]
    \\
    \mathbb{R}^{9,1 \vert \mathbf{16} + \overline{\mathbf{16}}}
    \ar[dr]_{ \mu_{{}_{D0}} }
    \\
    & B S^1\;.
  }
\end{equation}
Thus, \eqref{FundamentalM2M5} can be recast in the form 
$$
  \xymatrix@R=1.3em{
    {\color{blue} \mathrm{Ext}\left(\mathbb{R}^{9,1\vert \mathbf{16} + \overline{\mathbf{16}}} \right)}
    \ar[rr]^-{ \mu_{{}_{M2/M5}} }
    \ar[dr]_{\mu_{{}_{M2}}}
    &&
    S^4\;,
    \ar[dl]
    \\
    & B^4 \mathbb{Q}
  }
$$
so that its double dimensional reduction is given as in \eqref{RatDDRed} by the $\mathrm{Ext}/\mathrm{Cyc}$-adjunct as in the following diagram:
\begin{equation}
  \label{SuperRatDDRed}
  \raisebox{35pt}{
  \xymatrix@C=30pt@R=1.5em{
    {\color{blue} \mathbb{R}^{9,1\vert \mathbf{16} + \overline{\mathbf{16}}} }
    \ar@{}[rrrr]^-{ \mbox{ \tiny \begin{tabular}{c} double dimensional reduction \\ of $\mu_{{}_{M2/M5}}$\\{} \end{tabular}} }
    \ar@[blue]@/^1pc/[drrrr]|-{{\color{blue}\;\widetilde{\mu}_{M2/M5}\;} }
    \ar[d]_-{\color{blue}\eta}
    &&&&
    \\
    {\color{blue}\mathrm{Cyc}}\mathrm{Ext}\big( \mathbb{R}^{9,1\vert \mathbf{16} + \overline{\mathbf{16}}} \big)
      \ar[rrrr]|{\; {\color{blue}\mathrm{Cyc}}( \mu_{{}_{M2/M5}})\; }
      \ar[drr]
    &&&& {\color{blue}\mathrm{Cyc}}(S^4)\;.
    \ar[dll]
    \\
    && {\color{blue}\mathrm{Cyc}}(B^4 \mathbb{Q})
    \ar[d]
    \\
    &&
    B^3 \mathbb{Q}
  }}
\end{equation}
In \cite[Prop. 3.8]{FSS16a} and \cite[Theorem 3.8]{FSS16b}, this double dimensional reduction cocycle
$\widetilde{\mu}_{{}_{M2/M5}}$ was computed 
 as:
\begin{equation}
  \label{SuperDDRed}
  \widetilde{\mu}_{{}_{M2/M5}}
  \;\;=\;\;
  \left\{
  \begin{aligned}
    d \mu_{{}_{F1}} & = 0 \phantom{AAA}\phantom{BB}\phantom{CC}  d \mu_{{}_{NS5}}  = \mu_{{}_{D0}} \wedge \mu_{{}_{D4}}
     - \tfrac{1}{2} \mu_{{}_{D2}} \wedge \mu_{{}_{D2}}\;.
    \\
    \\
    d \mu_{{}_{D0}} & = 0
    \\
    d \mu_{{}_{D2}} & = \mu{{}_{F1}} \wedge \mu_{{}_{D0}}
    \\
    d \mu_{{}_{D4}} & = \mu_{{}_{F1}} \wedge \mu_{{}_{D2}}
    \,.
  \end{aligned}
  \right.
\end{equation}
We, therefore, obtain the truncation of \eqref{DCoch}, \eqref{DCoc} that contains the fundamental
F1-brane cocycle $\mu_{F1}$, as well as the fundamental D$p$-brane cocycles $\mu_{Dp}$ for
$p\in \{0,2,4\}$, together with Bianchi identities.
We also obtain the NS5-brane cocycle $\mu_{NS5}$.
To enhance this double dimensional reduction picture, we could ask for a lift as in \eqref{IsIt}:
\vspace{-3mm}
\begin{equation}
  \label{SuperIsIt}
  \raisebox{35pt}{
  \xymatrix@C=30pt@R=1.5em{
    X
    \ar@{-->}@/^2pc/[rrrr]^{ \mbox{ \tiny \begin{tabular}{c} enhanced \\ double dimensional reduction \\ of $\mu_{{}_{M2/M5}}$ {\color{blue} ??}  \end{tabular} }   }
    \ar[d]_{\eta_{{}_{ \mathbb{R}^{9,1\vert \mathbf{16} + \overline{\mathbf{16}}} }}}
    && &&
    {S^4 \dslash S^1}
    \ar[d]^{\eta_{{}_{S^4 \dslash S^1}}}
    \\
    { \mathrm{Cyc}}\mathrm{Ext}\big( \mathbb{R}^{9,1 \vert \mathbf{16} + \overline{\mathbf{16}}}  \big)
      \ar[rrrr]|-{\;\mathrm{Cyc}( \mu_{{}_{M2/M5}})\; }
      \ar[drr]_{ \mu_{{}_{F1}} }
    &&&& { \mathrm{Cyc}}( \mathrm{Ext}(S^4 \dslash S^1) )\;.
    \ar[dll]
    \\
    && B^3 \mathbb{Q}
  }
  }
\end{equation}
%
But a dashed lift in \eqref{SuperIsIt} does \emph{not} exist:
Prop. \ref{MinimalDGCAlgebraModelForATypeOrbispace}
requires
 that for such a lift to exist,
the double dimensionally reduced cocycle data in \eqref{SuperDDRed} needs to satisfy the extra conditions
\begin{enumerate}[{\bf (i)}]
  \item $\mu_{{}_{D4}} = 0$; and
  \item $\mu_{{}_{F1}} \wedge \mu_{{}_{D0}} = 0$,
\end{enumerate}
both of which fail, the first by \eqref{DCoch} the second by \eqref{DCoc}.
%
If we now approach the problem with homotopy theoretic perturbation theory (Rem. \ref{GoodwillieCalculus}),
we might instead ask for a \emph{first-order} lift as in \eqref{Enhc}, hence a lift in the following diagram:
\begin{equation}
  \label{SuperEnhc}
   \hspace{-6mm}
  \mathclap{
  \raisebox{48pt}{
  \xymatrix@C=40pt{
    \mathbb{R}^{9,1 \vert \mathbf{16} + \overline{\mathbf{16}} }
    \ar@[blue]@{-->}@/^2pc/[rrrr]^{
      \mbox{
        \tiny
        \begin{tabular}{c}
          \color{blue} perturbatively enhanced
          \\
          double dimensional reduction
          \\
          of $\mu_{{}_{M2/M5}}$  {\color{blue} ?? }
        \end{tabular}
      }
    }
    \ar[d]_{\eta_{{\mathbb{R}^{9,1\vert \mathbf{16} + \overline{\mathbf{16}} }}}}
    &&
    &&
    {\color{blue} \Omega^\infty_{S^3}\Sigma^\infty_{+,S^3}} \left( S^4 \dslash S^1 \right)
    \ar[d]^{ {\color{blue} \Omega^\infty_{S^3}\Sigma^\infty_{+,S^3}} \big( \eta_{{}_{S^4 \dslash S^1}} \hspace{-.7mm} \big) }
    \\
    { \mathrm{Cyc}}\,\mathrm{Ext}\big( \mathbb{R}^{9,1\vert \mathbf{16} + \overline{\mathbf{16}}} \big)
      \ar[rr]^-{ \mathrm{Cyc}( \mu_{{}_{M2/M5}} )  }
      \ar[drr]_{ \mu_{{}_{F1}} }
    && \mathrm{Cyc}( S^4 )
    \ar[d]
    \ar@[blue][rr]^-{\color{blue}  \mathrm{st}}
    &&
    {\color{blue} \Omega^\infty_{S^3}\Sigma^\infty_{+,S^3}} \mathrm{Cyc}(S^4)
    \ar[dll]
    \\
    && B^3 \mathbb{Q}
  }
  }}
\end{equation}
%
As was the case for \eqref{FactorEnhc}, Theorem \ref{TwistedKTheoryInsideFiberwiseStabilizationOfATypeOrbispaceOf4Sphere} says that the lift in \eqref{SuperEnhc}
exists precisely if the lift $\widehat{\mu}_{{}_{M2/M5}}$ in the following diagram exists:
\begin{equation}
  \label{SuperFactorEnhc}
    \hspace{-7mm}
  \mathclap{
  \raisebox{47pt}{
  \xymatrix@C=30pt{
    \mathbb{R}^{9,1 \vert \mathbf{16} + \overline{\mathbf{16}}  }
    \ar@[blue]@{-->}@/^2.4pc/[rrrr]^{
      \mbox{
        \tiny
        \begin{tabular}{c}
           perturbatively enhanced
          \\
          double dimensional reduction
          \\
          of $\mu_{{}_{M2/M5}}$  {\color{blue} ! }
        \end{tabular}
      }
    }
    \ar@[blue]@{-->}[rr]^{ \color{blue} \widehat{\mu}_{{}_{M2/M5}} }
    \ar[drr]|-{\;{\widetilde{\mu}_{{}_{M2/M5}}} \;}
    \ar[d]_{\eta_{{}_{ \mathbb{R}^{9,1\vert \mathbf{16} + \overline{\mathbf{16}} } }}}
    &&
    {\color{blue} \Omega^{\infty-2}_{B^3 \mathbb{Q}}\left( \mathrm{ku} \dslash B S^1 \right) }\,
    \ar@[blue]@{..}[d]|-{\color{blue} \tau_6}
    \ar@{^{(}->}[rr]^-{\iota}
    &&
    \Omega^\infty_{S^3}\Sigma^\infty_{+,S^3}\left( S^4 \dslash S^1 \right)
    \ar[d]^{ \Omega^\infty_{S^3}\Sigma^\infty_{+,S^3} \left( \eta_{{}_{S^4 \; \dslash \; S^1}} \hspace{-.7mm} \right) }
    \\
    { \mathrm{Cyc}}\,\mathrm{Ext} \big(  \mathbb{R}^{9,1\vert \mathbf{16} + \overline{\mathbf{16}}  }  \big)
      \ar[rr]|-{\;\mathrm{Cyc}( \mu_{{}_{M2/M5}})\; }
      \ar[drr]_{ \mu_{{}_{F1}} }
    && \mathrm{Cyc}( S^4 )
    \ar[d]
    \ar[rr]^-{\mathrm{st}}
    &&
    \Omega^\infty_{S^3}\Sigma^\infty_{+,S^3}\mathrm{Cyc}(S^4)
    \ar[dll]
    \\
    && B^3 \mathbb{Q}\,.
  }
  }}
\end{equation}
%
As in \eqref{DDRedEnhc}, such a lift  $\widehat{\mu}_{{}_{M2/M5}}$ is equivalent to specifying an extension of the truncated data in \eqref{SuperDDRed} to a cocycle:
\begin{equation}
  \label{SuperDDRedEnhc}
  \widehat{\mu}_{{}_{M2/M5}}
  \;\;=\;\;
  \left\{
  \begin{aligned}
    d \mu_{{}_{F1}} & = 0
    \\
    \\
    d \mu_{{}_{D0}} & = 0
    \\
    d \mu_{{}_{D2}} & = \mu_{{}_{F1}} \wedge \mu_{{}_{D0}}
    \\
    d \mu_{{}_{D4}} & = \mu_{{}_{F1}} \wedge \mu_{{}_{D2}}
    \\
    {\color{blue} d \mu_{{}_{D6}} } & {\color{blue} = \mu_{{}_{F1}} \wedge \mu_{{}_{D4}} }
    \\
    {\color{blue} d \mu_{{}_{D8}} } & {\color{blue} = \mu_{{}_{F1}} \wedge \mu_{{}_{D6}} }
    \\
    {\color{blue} d \mu_{{}_{D10}} } & {\color{blue} = \mu_{{}_{F1}} \wedge \mu_{{}_{D8}} }
    \,,
  \end{aligned}
  \right.
\end{equation}
(after discarding the NS5-brane cocycle $\mu_{{}_{NS5}}$).
By \eqref{DCoch} and \eqref{DCoc}, such an extension exists and is precisely the required enhancement of the
double dimensional reduction of the fundamental M2/M5-brane cocycle by the missing $\mathrm{D}(p \geq 6)$-brane cocycles to the full $\mathrm{F1}/\mathrm{D}p$-brane cocycle of type IIA string theory
 with coefficients in rational twisted K-theory.
This is exactly the required gauge enhancement: 
%
%
%
\begin{quotation}
\noindent A {\bf solution} to \hyperlink{OpenRational}{\bf Open Problem, rational version 2} (p.~
\pageref{OpenRationalPage}) is given as follows:
\begin{itemize}
  \item While double dimensional reduction of the fundamental M2/M5-brane cohomotopy 4-cocycle $\mu_{{}_{M2/M5}}$ \eqref{M2M5Coc} is obtained as the $\mathrm{Ext}/\mathrm{Cyc}$-adjunct \eqref{SuperDDRed} along the M-theoretic super-spacetime extension \eqref{11dExt};

  \item Its \emph{perturbative gauge enhancement}, making the full combined F1/D$p$-brane cocycle appear is obtained by lifting through the fiberwise stabilization of the $\mathrm{Ext}/\mathrm{Cyc}$-unit on the A-type orbispace of the 4-sphere \eqref{SuperEnhc}.
  \end{itemize}
\end{quotation}

\begin{remark}[Higher-dimensional branes]
An often neglected point is that, with local supersymmetry taken into account, the system of relations \eqref{SuperDDRedEnhc}
is indeed non-trivial up to the degree shown. The purely bosonic part of the D8-flux form $F_{10}$ is necessarily closed
(being a 10-form on a 10-dimensional manifold). However,  its fermionic component which is proportional to
$\left(\overline{\Psi} \Gamma_{a_1 \cdots a_8}\Gamma_{10} \Psi\right) \wedge E^{a_1} \wedge \cdots \wedge E^{a_8}$
is only of bosonic degree 8, so that its differential as a super-form need not vanish (and, indeed, does not vanish \cite[(6.9)]{CAIB00}).
For the same reason, we include in \eqref{DCoch} and \eqref{SuperDDRedEnhc} the non-trivial component $\mu_{{}_{D10}}$
which ought  to correspond to a D10-brane, were it not for the fact that there is no bosonic aspect of a 10-brane in a
10-dimensional spacetime. One way to see the \lq\lq D10-brane contribution'' $\mu_{{}_{D10}}$ arise is to
consider the type IIB D-brane super cocycles $\mu_{{}_{D1}}, \mu_{{}_{D3}}, \mu_{{}_{D5}}, \mu_{{}_{D7}}, \mu_{{}_{D9}}$ \cite[Sec. 2]{IIBAlgebra}
and then apply super-geometric T-duality \cite[Theorem 5.3]{FSS16b}.
The existence of the charge structure for would-be D10 branes was also noticed in \cite[p.~30]{CallisterSmith09} by different means.
%
\end{remark}

\medskip
\medskip
\noindent {\bf Acknowledgements.} We are grateful to August{\'i} Roig and Martintxo Saralegi-Aranguren for discussion of
\cite{RS}, as well as to David Corfield, Ted Erler, Domenico Fiorenza, and David Roberts for useful comments.
We also thank the anonymous referee for their careful reading and helpful suggestions.
VBM acknowledges partial support of SNF Grant No. 200020\_172498/1. This research was partly supported by
the NCCR SwissMAP, funded by the Swiss National Science Foundation, and by the COST Action MP1405 QSPACE,
supported by COST (European Cooperation in Science and Technology).
\medskip


\end{document}